\documentclass[english,letter paper,12pt,reqno]{article}
\usepackage{etex} %
\usepackage{array}
\usepackage{varwidth}
\usepackage{bussproofs}
\usepackage{epigraph}
\usepackage{stmaryrd}
\usepackage{mathdots}
\usepackage{subcaption}
\usepackage{amsmath, amscd, amssymb, mathrsfs, accents, amsfonts,amsthm}
\usepackage[all]{xy}
\xyoption{rotate}
\usepackage{mathtools} %
\usepackage{tikz-cd}
\usepackage{bbm}
\usepackage{diagbox}
\usetikzlibrary{calc}
\usetikzlibrary{positioning}
\usetikzlibrary{decorations.pathmorphing}
\usepackage{adjustbox}

\newcommand{\tikzxmark}{%
\tikz[scale=0.23] {
    \draw[line width=0.7,line cap=round] (0,0) to [bend left=6] (1,1);
    \draw[line width=0.7,line cap=round] (0.2,0.95) to [bend right=3] (0.8,0.05);
}}
\newcommand{\tikzcmark}{%
\tikz[scale=0.23] {
    \draw[line width=0.7,line cap=round] (0.25,0) to [bend left=10] (1,1);
    \draw[line width=0.8,line cap=round] (0,0.35) to [bend right=1] (0.23,0);
}}

\DeclarePairedDelimiter\ket{\lvert}{\rangle}
\DeclarePairedDelimiterX\braket[2]{\langle}{\rangle}{#1 \delimsize\vert #2}
\DeclarePairedDelimiterX\inner[2]{\langle}{\rangle}{#1,#2}

\definecolor{Myblue}{rgb}{0,0,0.6}
\usepackage[a4paper,colorlinks,citecolor=Myblue,linkcolor=Myblue,urlcolor=Myblue,pdfpagemode=None]{hyperref}

\SelectTips{cm}{}

\theoremstyle{definition}
\newtheorem{theorem}{Theorem}[section]
\newtheorem{proposition}[theorem]{Proposition}
\newtheorem{lemma}[theorem]{Lemma}
\newtheorem{corollary}[theorem]{Corollary}

	\newtheorem{definition}[theorem]{Definition}
	\newtheorem{example}[theorem]{Example}
	\newtheorem{remark}[theorem]{Remark}

\newcommand{\detectA}{\texttt{detectA}}

\def\res{\operatorname{Res}}

\def\vacu{\ket{\emptyset}}
\def\be{\begin{equation}}
\def\ee{\end{equation}}

\DeclareMathOperator{\inc}{inc}

\def\comp{\underline{\textup{comp}}}
\def\contract{\;\lrcorner\;}

\newcommand{\prom}{\operatorname{Prom}}

\newcommand{\der}{\operatorname{Der}}

\newcommand{\lto}{\longrightarrow}
\newcommand{\bb}[1]{\mathbb{#1}}

\def\l{\,|\,}

\newcolumntype{?}{!{\vrule width 1pt}}

\title{Programs as Singularities}
\author{Daniel Murfet, Will Troiani}
\date{September 11, 2026}

\begin{document}

\def\ScoreOverhang{1pt}

\makeatletter
\DeclareRobustCommand{\rvdots}{%
  \vbox{
    \baselineskip4\p@\lineskiplimit\z@
    \kern-\p@
    \hbox{}\hbox{.}\hbox{.}\hbox{.}
  }}
\makeatother

\newcommand{\startproof}[1]{
	\AxiomC{#1}
	\noLine
	\UnaryInfC{$\vdots$}
}
\newcommand{\ax}{\operatorname{Ax}}
\newcommand{\proofvdots}[1]{\overset{\displaystyle #1}{\rvdots}}
\def\Res{\res\!}
\newcommand{\ud}[1]{\operatorname{d}\!{#1}}
\newcommand{\Ress}[1]{\res_{#1}\!}
\newcommand{\cat}[1]{\mathcal{#1}}
\newcommand{\xlto}[1]{\stackrel{#1}\lto}
\newcommand{\mf}[1]{\mathfrak{#1}}
\newcommand{\md}[1]{\mathscr{#1}}
\newcommand{\church}[1]{\underline{#1}}
\newcommand{\prf}[1]{\underline{#1}}
\newcommand{\den}[1]{\llbracket #1 \rrbracket}
\def\l{\,|\,}
\def\sgn{\textup{sgn}}
\def\cont{\operatorname{cont}}
\def\counit{\varepsilon}
\def\ptail{\underline{\operatorname{tail}}}
\def\phead{\underline{\operatorname{head}}}
\def\comp{\underline{\textup{comp}}}
\def\mult{\underline{\textup{mult}}}
\def\repeat{\underline{\textup{repeat}}}
\def\contract{\;\lrcorner\;}
\def\<{\langle} \def\>{\rangle}
\newcommand{\id}{\text{id}}
\newcommand{\del}{\partial}
\newcommand{\Inj}{\operatorname{Inj}}

\newcommand{\tTur}{\textbf{Tur}}
\newcommand{\tInt}{\textbf{int}}
\newcommand{\tBint}{\textbf{bint}}
\newcommand{\tList}{\textbf{list}}
\newcommand{\tTint}{\textbf{tint}}
\newcommand{\tMBool}{\textbf{${_m}$bool}}
\newcommand{\tNBool}{\textbf{${_n}$bool}}
\newcommand{\tSBool}{\textbf{${_s}$bool}}
\newcommand{\tHBool}{\textbf{${_h}$bool}}
\newcommand{\tSList}{{_s}\textbf{list}}

\newcommand{\reject}{\operatorname{reject}}
\newcommand{\accept}{\operatorname{accept}}

\newcommand{\pTapehead}{\text{\underline{tapehead}}}
\newcommand{\pAbsStep}{\text{\underline{absstep}}}
\newcommand{\pSymbol}{\text{\underline{symbol}}}
\newcommand{\pBoolWeak}{\text{\underline{boolweak}}}
\newcommand{\pLeft}{\text{\underline{left}}}
\newcommand{\pEval}{\text{\underline{eval}}}
\newcommand{\pRight}{\text{\underline{right}}}
\newcommand{\pState}{\text{\underline{state}}}
\newcommand{\pRecomb}{\text{\underline{recomb}}}
\newcommand{\pConcat}{\text{\underline{concat}}}
\newcommand{\pListconcat}{\text{\underline{listconcat}}}
\newcommand{\pTrans}{\text{\underline{trans}}}
\newcommand{\pDecomp}{\text{\underline{decomp}}}
\newcommand{\pHead}{\text{\underline{head}}}
\newcommand{\pTail}{\text{\underline{tail}}}
\newcommand{\pIntCopy}{\text{\underline{intcopy}}}
\newcommand{\pBintCopy}{\text{\underline{bintcopy}}}
\newcommand{\pNBoolCopy}{\text{${_n}$\underline{boolcopy}}}
\newcommand{\pStep}{\text{\underline{step}}}
\newcommand{\pRelStep}{\text{\underline{relstep}}}
\newcommand{\pBoolStep}{\text{\underline{boolstep}}}
\newcommand{\pExtract}{\text{\underline{extract}}}
\newcommand{\pPack}{\text{\underline{pack}}}
\newcommand{\pUnpack}{\text{\underline{unpack}}}
\newcommand{\pRead}{\text{\underline{read}}}
\newcommand{\pMultread}{\text{\underline{multread}}}
\newcommand{\pTensor}{\text{\underline{tensor}}}
\newcommand{\pComp}{\text{\underline{comp}}}
\newcommand{\pRepeat}{\text{\underline{repeat}}}
\newcommand{\pAdd}{\text{\underline{add}}}
\newcommand{\pMult}{\text{\underline{mult}}}
\newcommand{\pPred}{\text{\underline{pred}}}
\newcommand{\pIter}{\text{\underline{iter}}}
\newcommand{\pBooltype}{\text{\underline{booltype}}}
\newcommand{\pCast}{\text{\underline{cast}}}
\newcommand{\proj}{\operatorname{proj}}
\newcommand{\prob}{\bold{P}}
\newcommand{\probc}{\mathscr{P}}
\newcommand{\dkl}{D_{\operatorname{KL}}}

\newcommand{\nl}{\text{nl}} %
\newcommand{\dntn}[1]{\llbracket #1 \rrbracket} %
\newcommand{\dntntup}[1]{\langle\!\langle #1 \rangle\!\rangle} %
\newcommand{\dntnPT}[1]{\dntn{#1}_{\text{PT}}}
\newcommand{\dntnNL}[1]{\dntn{#1}_{\nl}}
\newcommand{\Sym}{\operatorname{Sym}}
\newcommand{\ND}{\text{ND}} %

\maketitle

\begin{abstract} We develop a correspondence between the structure of Turing machines and the structure of singularities of real analytic functions, based on connecting the Ehrhard-Regnier derivative from linear logic with the role of geometry in Watanabe's singular learning theory. The correspondence works by embedding ordinary (discrete) Turing machine codes into a family of ``noisy'' codes which form a smooth parameter space. On this parameter space we consider a potential function which has Turing machines as critical points. By relating the Taylor series expansion of this potential at such a critical point to combinatorics of error syndromes, we relate the local geometry to internal structure of the Turing machine.

The potential in question is the negative log-likelihood for a statistical model, so that the structure of the Turing machine and its associated singularity is further related to Bayesian inference. Two algorithms that produce the same predictive function can nonetheless correspond to singularities with different geometries, which implies that the Bayesian posterior can discriminate between distinct algorithmic implementations, contrary to a purely functional view of inference. In the context of singular learning theory our results point to a more nuanced understanding of Occam’s razor and the meaning of simplicity in inductive inference.
\end{abstract}

\setlength{\epigraphwidth}{0.7\textwidth}
\epigraph{All of this will lead to theories which are much less rigidly of an all-or-none nature than past and present formal logic. They will be of a much less combinatorial, and much more analytical, character. In fact there are numerous indications to make us believe that this new system of formal logic will move closer to another discipline which has been little linked in the past with logic. This is thermodynamics, primarily in the form it was received from Boltzmann, and is that part of theoretical physics which comes nearest in some of its aspects to manipulating and measuring information.}{J. von Neumann, Collected Works V, \cite[p. 304]{vonNeumann1963}}

\tableofcontents
\newpage

\section{Introduction}

We present a perspective on programs which relates them to the geometry of singularities in real algebraic geometry \cite{arnold1981singularity,greuel2007introduction,greuel2007singularity}. This continues a line of work on algebraic semantics of linear logic \cite{murfet2015sweedler,clift2020cofree,CliftMastersThesis} and connections of those semantics to geometry \cite{murfet2022eliminationcuteliminationmultiplicativelinear, troiani2025linearlogichilbertscheme} and statistical learning theory \cite{clift2021geometry,waring2021geometric}. In this introduction we explain the key ideas of the present work without requiring any familiarity with this background material.
\\

The concept of an \emph{algorithm} or \emph{program} had a long gestation. From a modern perspective we can see two primary threads: the refinement of logical debate and structured arguments starting with Aristotle and the refinement of the idea of an algebraic manipulation or calculation starting with Al-Khw\={a}rizm\={i}. Over time these threads came together in the work of Leibniz and Boole \cite{boole1854laws} then later Gentzen \cite{gentzen1935untersuchungen}, Turing \cite{turing1936computable}, Brouwer-Heyting-Kolmogorov \cite{troelstra2011history}, Church \cite{church1932set}, Curry and Howard \cite{howard1980formulae} and others to give us the modern synthesis of logic and computation. However, with the rise of machine learning and artificial intelligence the conceptual landscape under this synthesis has shifted, as already anticipated by Turing \cite{turing1948intelligent} and von Neumann \cite{von1958computer, vonNeumann1963}. Soon most sophisticated cognition on Earth, including the construction of programs, will occur in machines. Thus it may be appropriate to expand the proper domain of logic to include understanding the laws of thought that emerge in the minds of large learning machines \emph{other} than ourselves.\footnote{Interpretability is the new (old) logic.}

One way to take this seriously is to search for an embedding of the discrete world of proofs and programs into the continuous world of machine learning and statistical learning theory. We know one relevant example: our brains are continuous but can parametrise effectively discrete forms of reasoning and symbol manipulation \cite{von1958computer}. Our main desiderata for this embedding should be that it be a \emph{homomorphism}, that is, the internal structure of proofs and programs should be reflected in the geometric structure of the continuous space near the image of that program. When the geometry in question is that of the singularity of a potential function, there is a long tradition in mathematics and physics of relating the ``internal structure'' of points to the geometry of singularities \cite{cecotti1993classification,cartier2001mad,greuel2007singularity,arnold2012singularities}.
\\

In this paper we study such a homomorphism, that is, a structure preserving map from programs to singularities. In \cite{clift2018derivatives,clift2021geometry,waring2021geometric} the map was constructed by considering \emph{synthesis problems}, that is, the Bayesian inference of a (noisy) Turing machine code given input-output examples in the spirit of Solomonoff \cite{solomonoff1964formal}. In this paper we develop a language of \emph{error syndromes} in which to describe this map and exhibit the sense in which it is structure-preserving. Finally, we explain how this correspondence offers a new view on Bayesian statistics, which we term \emph{structural Bayesianism}. This point of view emphasises that according to standard principles of model selection we should not only prefer models that predict well but which further have an internal structure which is optimal. This provides some theoretical grounding for approaches that attempt to perform structural inference on neural networks, another important class of singular models \cite{hoogland2024developmental,wang2024differentiation}.
\\

\emph{Acknowledgements.} We warmly thank Billy Snikkers, Zach Furman and Rumi Salazar for participating in a series of lectures on this material at the University of Melbourne and Zach Furman for suggesting the relation to Fourier analysis of Boolean functions. We also thank Rumi Salazar and Billy Snikkers for their helpful identification of typos, and Chris Elliott for a helpful comment on the error-correction example in Section~\ref{section:error_correction}.

\subsection{Overview}

The paper in overview:

\begin{itemize}
    \item Section \ref{section:background}: we recall the general setup of synthesis problems and smooth relaxations of the execution of a UTM, and modulo the details of that smooth relaxation how we define a statistical model (Section \ref{section:background_iid}, Section \ref{section:details_model}). We then recall the definition of learning coefficients (Section \ref{section:background_llc}) and how this geometric invariant controls the Bayesian learning process according to singular learning theory (Section \ref{section:free_energy}).
    \item Section \ref{section:geometry}: having associated to a Turing machine $M$ a singularity $([M], L)$ we proceed to introduce what it means to study the geometry of such an object. This leads to the definition of \emph{influence functions} $g_{\bold{i}}(x)$.
    \item Section \ref{section:derivatives_errors}: we recall the basics of linear logic and its semantics (Section \ref{section:linear_logic_setup}). Then we give a sketch of how in this semantics inputs are copied and distributed throughout the ``programs'' that are one interpretation of linear logic proofs, and how when the inputs are distributions this has the semantics of ``resample whenever you use an input'' which leads to the fundamental notion of an error syndrome (Section \ref{section:eval_prog_vec}). This is further expanded in Section \ref{section:error_syndromes} which concludes with Theorem \ref{theorem:error_count}. This says that derivatives of denotations of plain proofs (an important case being the proof that encodes the execution of a UTM for $t$ timesteps) can be thought of in terms of counts of error syndromes.
    \item Section \ref{section:diff_utm}: we introduce a plain proof $\psi$ which encodes reading off the final state of a simulated machine after the execution of a particular UTM for $t$ steps (Section \ref{section:encoding_utm_step}) and then apply Section \ref{section:derivatives_errors} to that proof to define derivatives and express them in terms of error syndromes (Section \ref{section:error_syndromes_U}).
    \item Section \ref{sec:GeometryofInductiveInference}: we express the Taylor series expansion of the polynomial function $H(w)$ at a code $[M]$ in terms of counts of error syndromes (Theorem \ref{theorem:partial_formula_explicit}).
    \item Section \ref{sec:examples}: using Theorem \ref{theorem:partial_formula_explicit} we explain how to make three separate connections between program structure and geometry (summarised in Section \ref{section:pas} below).
    \item Section \ref{section:conclusion}: in the conclusion we draw some general lessons from this connection between programs and singularities for Bayesian inference and interpretability and introduce the philosophy of structural Bayesianism.
\end{itemize} 

In the rest of this section we sketch the main results. First we recall the general situation of Bayesian inference of Turing machine codes considered in \cite{clift2018derivatives,clift2021geometry,waring2021geometric}, which is depicted in Figure \ref{fig:utm_diagram}. 

\subsubsection{Inductive Inference of TMs}\label{section:overview_ind_inf}

We suppose there is a computable function $y = y(x)$ with $y \in \{0,1\}$ and $x \in \Sigma^*$ for some alphabet of symbols $\Sigma$ where in this paper $\Sigma^*$ denotes strings of positive length. That is, we are given a \emph{language} $L = \{ x \in \Sigma^* \mid y(x) = 1 \}$. The learning problem, which we refer to as a \emph{synthesis problem}, is to infer a Turing machine code that explains a given dataset $D_n = \{ (x_i, y_i) \}_{i=1}^n$ of examples with $y_i = y(x_i)$ and $x_i \sim q(x)$ for some given distribution over inputs with finite support $I \subseteq \Sigma^*$. We refer to a Turing machine $M$ as a \emph{classical solution} if $M(x) = y(x)$ for all $x \in I$.\footnote{We will be more precise about the meaning of $M(x)$ below, but note that it will involve a timeout.} For the statistical analysis below, we smooth both the observation law and the model probabilities as explained in Section \ref{section:details_model}, while retaining this deterministic target function.

\begin{figure}[t]
    \centering
    \includegraphics
        [width=0.8\textwidth]
        {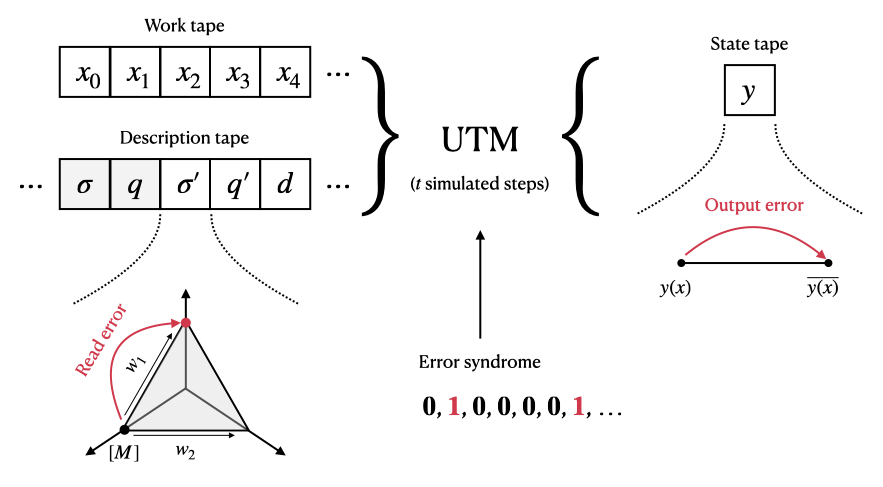}
    \caption{\label{fig:utm_diagram}%
        The parameter space $W$ with coordinates $w_1,\ldots,w_d$ parametrises ``noisy codes'' which specify for each entry on the description tape of a UTM a distribution over possible values. When executing such a noisy Turing machine code on an input $x$ for $t$ steps of the simulated machine the UTM can experience read errors when accessing the code on the description tape, and the pattern of such errors determines an error syndrome. Any particular error syndrome may flip the output of the simulated machine from the correct output $y = y(x)$ to the incorrect output $\overline{y(x)}$.
    }
\end{figure}

We can perform Bayesian inference to assign probabilities to various Turing machines $M$ represented by their codes $[M]$ for a particular universal Turing machine (UTM) which we denote $\mathcal{U}$. However, this approach does not ``see'' the internal structure of a Turing machine. On the theory that internal structure is revealed by the pattern of responses of a system to perturbations, we consider perturbations of the dynamical system $\mathcal{U}$ simulating $M$ of the following kind: when the UTM performs an operation, it \emph{reads from the description tape with some probability of error}, obtaining enough samples to simulate the machine from the first step. That is, we view the connection of the description tape to the rest of the UTM as a noisy channel. The code $w$ for a \emph{noisy Turing machine} specifies the precise errors that can occur and their probability. In this way we can include the ordinary Turing machine codes in a larger continuous space $W$. On this space we define the negative log-likelihood
\begin{equation}
L_n(w) = - \frac{1}{n} \sum_{i=1}^n \log p(y_i|x_i,w)
\end{equation}
for any noisy Turing machine code $w \in W$, and its average $L(w) = \mathbb{E}_{D_n}[ L_n(w) ]$ over all possible datasets. This is an analytic function and the global minima of $L$ include Turing machines $M$ which are solutions in the above sense. The singularity assigned in \cite{clift2021geometry,waring2021geometric} to the pair consisting of $M$ and the synthesis problem is the germ $([M], L)$ (see Definition \ref{defn:germ} below). Of course it remains to explain precisely how to assign probabilities to $y_i \in \{0,1\}$ given an input $x_i$ and a noisy Turing machine code $w$ (we will get to this in Section \ref{section:background_iid}).

\subsubsection{Error Syndromes}

With this preparation we can sketch the idea of error syndromes, influence functions and how they relate the structure of $M$ to the geometry of $([M], L)$. An \emph{error syndrome} $\gamma$ records the pattern of errors encountered while $\mathcal{U}$ simulates $M$ for some fixed number $t$ of steps. Let $M$ be a classical solution and let $\mathcal{U}(x, \gamma)$ denote the ``output'' (the entry on the state tape of the UTM) obtained on input $x$ with error syndrome $\gamma$ when simulating $M$ for $t$ steps. If we only allow a single error on the $i$th entry of the description tape the \emph{influence function of weight one} is (Section \ref{section:influence_func_utm})
\begin{equation}
g_i(x) = \Big| \big\{ \gamma \l \mathcal{U}(x,\gamma) \neq y(x) \big\} \Big|\,.
\end{equation}
That is, $g_i \in \mathbb{N}^I$ assigns to each input $x$ the number of those error syndromes, with one error localised to the $i$th entry of the description tape, which cause an output error. Up to second-order the main structural information there is to know about the germ $([M], L)$ is whether the singularity is nondegenerate, that is, whether its Hessian is invertible. One can show that this geometry of $L$ is the same as the geometry of a polynomial function $H$ and our results are formulated about this polynomial (Section \ref{section:geometry}).

The Hessian of $H$ can be described in terms of influence functions (Corollary \ref{corollary:partial_formula_explicit_hessian})
\begin{equation}\label{eq:intro_hessian_H}
\frac{1}{2} \frac{\partial^2}{ \partial w_i \partial w_j} H \Bigr|_{w = [M]} = \mathbb{E}_x\Big[ g_i(x) g_j(x) \Big]
\end{equation}
where $\mathbb{E}_x[-]$ denotes expectation with respect to $q(x)$ and $w_1,\ldots,w_d$ are natural coordinates on the space of noisy TM codes $W$ at $[M]$. It follows that $([M], L)$ is nondegenerate if and only if the set of influence functions is linearly independent, or what is the same, if \emph{every bit of the code of $M$ influences the computation in an independent way}. This provides the first hint of a nontrivial relation between the geometry of $([M], L)$ and the structure of the Turing machine (see Section \ref{section:meaning_nondeg}).

More generally for a nonzero multi-index $\bold{k} = (k_1,\ldots,k_d)$ we have (Lemma \ref{lemma:partial_derivatives_H})
\begin{equation}
\frac{\partial^{|\bold{k}|}}{\partial w_1^{k_1} \cdots \partial w_d^{k_d}} H \Bigr|_{w = [M]} = \sum_{\substack{\bold{i},\bold{j} \neq 0 \\ \bold{i} + \bold{j} = \bold{k}}} C(\bold{i}, \bold{j})\, \mathbb{E}_x\Big[ g_{\bold{i}}(x) g_{\bold{j}}(x) \Big]
\end{equation}
where the $C$ are some combinatorial factors and $g_{\bold{i}}(x)$ is a general influence function, which can be expressed in terms of counts of error syndromes subject to the condition that the number of errors in each entry of the description tape is at most the integer specified in $\bold{i}$ (Theorem \ref{theorem:partial_formula_explicit}). In this precise sense, the combinatorics of error syndromes controls the geometry of $H$ and thus of $L$.

\subsubsection{Programs as Singularities}\label{section:pas}

We claim that typical Turing machine codes are highly degenerate (see Section \ref{section:full_geometry} for an example). That is, the effects of different entries on the description tape are highly correlated in terms of the effect that errors in these bits have on the output. We claim further that the pattern of these correlations reflects internal structure of Turing machines. Since there is no general formal definition of ``internal structure'' of programs to compare to, evidence for this claim takes the form of the development (in this paper and its sequels) of examples of program structure, together with a study of how this structure is reflected in correlations of influence functions (Figure \ref{figure:triumvirate}).
\\

In this paper we give three examples:
\begin{itemize}
\item If a Turing machine $M$ has run-time error correction then the Taylor series expansion of $H$ at $[M]$ has many vanishing lower order terms; this leads to upper bounds on the learning coefficient of $L$ (Section \ref{section:error_correction}).
\item It is apparent from \eqref{eq:intro_hessian_H} that Turing machines which process subsets of the input using independent control paths should have that structure reflected in their geometry (see Section \ref{section:control_flow}).
\item In Section \ref{section:full_geometry} we study a Turing machine $\detectA^{(0)}$ which recognises the language consisting of all strings in $\{ A, B \}$ which contain an $A$. We compute explicitly the polynomial $H$ to second order and diagonalise its Hessian in local coordinates to exhibit how the geometry relates to the structure of the machine. 
\end{itemize}

\begin{figure}[htpb]
\begin{center}
\quad\xymatrix@R+3pc{
& *++[F-,]{\text{Combinatorics of error syndromes}} \ar[dr(0.8)]\\
*++[F-,]{\text{Program structure}} \ar@{=>}[rr]\ar[ur(0.8)] & & *++[F-,]{\text{Geometric structure}}
}
\end{center}
\caption{The relationship between program structure and geometry we establish works by relating both to the combinatorics of error syndromes, which are patterns of flips in bits on the description tape of a UTM which affect the output of the simulated machine.}
\label{figure:triumvirate}
\end{figure}
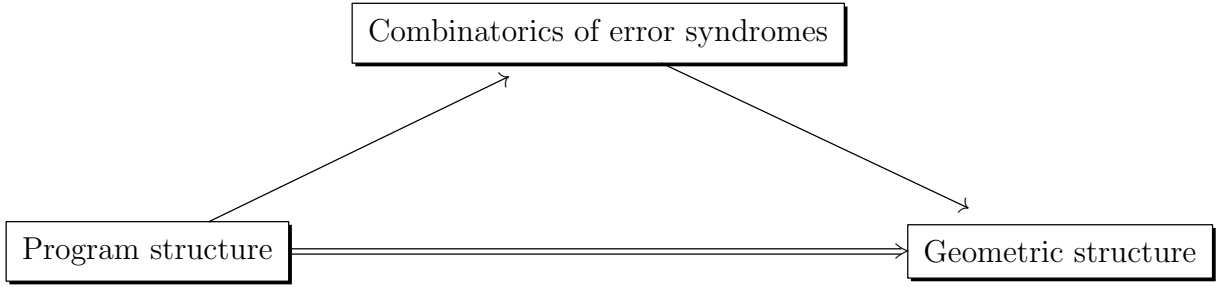

\begin{table}[hptb]
\centering
\begin{tabular}{|c|p{9cm}|c|}
\hline
\textbf{Symbol} & \textbf{Description} & \textbf{Reference} \\
\hline
$\gamma$ & Error syndrome & Definition \ref{defn:error_syndrome} \\
\hline
$\mathcal{U}$ & Staged pseudo-universal Turing machine & Section \ref{appendix:stagedpseudoutm} \\
\hline
$W$ & Parameter space of noisy Turing machines & Section \ref{section:background} \\
\hline
$d$ & Dimension of $W$ & Equation \eqref{eq:d_dim_W}\\
\hline
$W^{\texttt{code}}$ & Set of (deterministic) Turing machine codes & Section \ref{section:background} \\
\hline
$[M]$ & The code of a Turing machine $M$ & Definition \ref{def:code_of_T} \\
\hline
$\Delta \textrm{step}^t$ & Execution of a noisy Turing machine for $t$ steps & Equation \eqref{eq:smooth_relax} \\
\hline
$L(w)$ & Average negative log-likelihood function & Section \ref{section:details_model} \\
\hline
$K(w)$ & KL divergence & Section \ref{section:details_model} \\
\hline
$H(w)$ & Polynomial function comparable to $K$ & Definition \ref{definition:polynomial_H} \\
\hline
$g_{\bold{i}}(x)$ & Influence function measuring the effect of errors of pattern $\bold{i}$ on input $x$ & Definition \ref{defn:influence} \\
\hline
$A^{\bold{s}}(x)$ & Count of weight $\bold{s}$ error syndromes that cause output errors on input $x$ & Definition \ref{defn:Asfunctions} \\
\hline
$\mathbb{E}_x[-]$ & Expectation with respect to $q(x)$ & Section \ref{section:geometry} \\
\hline
$\lambda([M],q)$ & Local learning coefficient of a Turing machine $M$ & Definition \ref{defn:local_learning_coefficient} \\
\hline
$\psi$ & Plain proof in intuitionistic linear logic & Definition \ref{def:plain} \\
\hline
$\pi$ & Proof (usually the linear part of $\psi$) & Definition \ref{def:plain} \\
\hline
$f^\tau_\psi$ & Polynomial computing the denotation of $\psi$ & Definition \ref{defn:f_tau} \\
\hline
$\den{\cdot}$ & Denotational semantics operator & Section \ref{section:linear_logic_setup} \\
\hline
$\mathcal{U}(x,\gamma)$ & Evaluation of $\mathcal{U}$ input $x$ and error syndrome $\gamma$ & Definition \ref{defn:uxgamma} \\
\hline
\end{tabular}
\caption{Key symbols used throughout the paper and their descriptions.}
\label{table:symbols}
\end{table}

\section{Background}\label{section:background}

\subsection{From Turing Machine to Singularity}\label{section:background_iid}

We recall the treatment of inductive inference in \cite{clift2018derivatives, clift2021geometry}, assuming familiarity with Turing machines at the level of \cite{arora2009computational}. Our conventions for Turing machines are as in \cite[\S 6.1]{clift2018derivatives}. We denote by $\Sigma$ the set of symbols and $Q$ the set of states. The set $\Sigma$ is assumed to contain some designated blank symbol $\Box$ which is the only symbol that is allowed to occur infinitely often on the tape. Often one also designates a starting state.

The \emph{configuration} of a Turing machine $M$ is an element $\big( (\sigma_u)_{u \in \mathbb{Z}}, q )$ of $\Sigma^{\mathbb{Z}} \times Q$ where $q$ is the current state and the symbol in the square in position $u$ relative to the head is $\sigma_u$ (so $\sigma_0$ is the symbol currently under the tape head, $\sigma_{-1}$ is the symbol immediately to the left of the tape head). Observe that the configuration of the tape actually lies in the smaller set of functions which are finitely supported, in the following sense:

\begin{definition}\label{defn:finsupstate} We write
\[
\Sigma^{\mathbb{Z}, \Box} = \{ f: \mathbb{Z} \lto \Sigma \l f(u) = \Box \text{ except for finitely many $u$} \}\,.
\]
\end{definition}

Given a finite alphabet $\Sigma$ and set of states $Q$ a Turing machine $M$ for this alphabet and set of states is identified with its \emph{transition function}
\begin{equation}
\delta: \Sigma \times Q \lto \Sigma \times Q \times \{ \text{L}, \text{S}, \text{R} \}\,.
\end{equation}
The interpretation is that when the machine reads the symbol $\sigma$ and is in state $q$ it writes $\sigma'$, transitions to state $q'$ and moves in direction $d$ (Left, Stay or Right) where $\delta(\sigma,q) = (\sigma', q', d)$. We sometimes write $\delta_0(\sigma,q), \delta_1(\sigma,q), \delta_2(\sigma,q)$ for $\sigma',q',d$ respectively.

The set of Turing machines with this alphabet and these states is
\begin{equation}
    W^{\texttt{code}} = \prod_{\sigma, q}\Sigma \times Q \times \{\text{L}, \text{S}, \text{R}\}
\end{equation}
where unless specified otherwise, $\sigma$ always ranges over $\Sigma$ and $q$ over $Q$. We think of $W^{\texttt{code}}$ as a set of \emph{codes} since the tuples $(\sigma, q, \sigma', q', d)$ are, when encoded on the description tape of a Universal Turing Machine (UTM) for all $\sigma,q$, a code for the Turing machine with the given transition function.

\begin{definition}
\label{def:code_of_T}
    If $M$ is a Turing machine with alphabet $\Sigma$ and set of states $Q$, then the point in $W^{\texttt{code}}$ corresponding to $M$ is the \emph{code of $M$} and is denoted $[M]$.
\end{definition}

In this paper $\Sigma, Q$ are viewed as fixed and all Turing machines have the same set of symbols and states. Thus all Turing machines have codes in $W^{\texttt{code}}$. See Appendix \ref{section:faq_realutm} for some comments on this restriction. 

\begin{definition}
\label{def:simplex}
    Given a set $Z$ we denote by $\Delta Z$ the \emph{simplex of $Z$} which is
    \begin{equation}
        \Delta Z := \big\{f: Z \lto [0,1] \mid \sum_{z \in Z}f(z) = 1 \big\}.
    \end{equation}
\end{definition}
A \emph{noisy Turing machine} is given by a function
\begin{equation}
\delta: \Sigma \times Q \lto \Delta \Sigma \times \Delta Q \times \Delta \{ \text{L}, \text{S}, \text{R} \}
\end{equation}
so the set of all noisy Turing machines is
\begin{equation}
\label{eq:noisy}
        W := \prod_{\sigma, q}\Delta \Sigma \times \Delta Q \times \Delta\{\text{L}, \text{S},\text{R}\}\,.
\end{equation}
This is a manifold with corners of dimension
\be\label{eq:d_dim_W}
d := \dim W = \sum_{\sigma, q}( |\Sigma| - 1 + |Q| - 1 + 2 ) = |\Sigma||Q|\big(|\Sigma| + |Q|\big)\,.
\ee
Again, we do not distinguish between a machine and its code, which are the points of $W$. The interpretation of a noisy Turing machine code is that when the UTM attempts to simulate the machine on an input $x$, it has some probability of \emph{read errors} when it accesses the description tape. For instance if the second component of $\delta(\sigma,q)$ is $\bold{q}' \in \Delta Q$ then when the UTM attempts to determine what state to transition the simulated machine to when it reads $\sigma$ and is in state $q$, it receives from the description tape a sample from the distribution $\bold{q}'$. That is, communication between the description tape and other parts of the UTM are a \emph{noisy channel} (see Appendix \ref{section:think_noisy} for some elaboration).\footnote{There is an existing literature on noisy Turing machines, with some variations in the definition. Unlike \cite{asarin2005noisy} we do not consider random errors in every part of the configuration of the Turing machine. For example, we do not introduce random errors on the tape squares. Our noise is restricted to determining the tuples in the description to execute.}

Different ways of thinking about the operation of the UTM in the presence of this noise channel lead to different \emph{smooth relaxations} \cite{clift2018derivatives}, by which we mean in particular the specification of smooth functions $\Delta \textrm{step}^t$ for all integers $t \ge 1$ which make the diagram
\begin{equation}\label{eq:smooth_relax}
    \xymatrix@C+3pc{
    \Sigma^* \times W \ar[r]^-{\Delta \textrm{step}^t} & \Delta Q\\
    \Sigma^* \times W^{\texttt{code}} \ar[u]\ar[r]_-{\textrm{step}^t} & Q \ar[u]
    }
\end{equation}
commute, where the bottom row is the usual execution of the UTM for $t$ steps of the simulated machine, and the vertical maps are the canonical inclusions. We read
\begin{equation}\label{eq:prop_uncer}
    \Delta \textrm{step}^t(x,w) = \sum_y \Delta \textrm{step}^t(x,w)_y \cdot y \in \Delta Q
\end{equation}
as saying that when the UTM executes the noisy code $w$ on a (deterministic) input $x \in \Sigma^*$ for $t$ steps, the distribution over the state that results from propagating this uncertainty about the code to uncertainty about the final state of the simulated machine is $\{ \Delta \textrm{step}^t(x,w)_y \}_{y \in Q}$. If there is no uncertainty in the code, there is no uncertainty in the state. The precise smooth relaxation we have in mind is recalled in Section \ref{section:encoding_utm_step}.
\\

This becomes a problem of inductive inference \cite[\S 7]{clift2018derivatives} if we observe the final state $y$ of the execution of some unknown machine when run on some inputs $x \in \Sigma^*$, and we wish to know \emph{which Turing machine} explains these observations. If our model of the propagation of uncertainty is \eqref{eq:prop_uncer} then the conditional distribution of the state $y \in Q$ given $x,w$ is\footnote{For expository purposes we give here a simplified form of the model, see Section \ref{section:details_model} for details.}
\begin{equation}
\label{eq:posterior}
p(y|x,w) = \Delta \textrm{step}^t(x,w)_y
\end{equation}
where we assume $t$ is fixed; for a more general class of synthesis problem see \cite{waring2021geometric} and also Appendix \ref{section:faq_timesteps}. According to Bayes' rule if we observe a set of inputs and final states $D_n = \{ (x_i, y_i) \}_{i=1}^n$ the probability $p(w|D_n)$ of a given noisy code $w$ is
\begin{equation}
p(w|D_n) = \frac{p(D_n|w) \varphi(w)}{p(D_n)} = \frac{1}{Z_n} \varphi(w) \prod_{i=1}^n p(y_i | x_i, w)
\end{equation}
where $\varphi(w)$ is our prior belief in codes, and $Z_n := p(D_n)$ is a normalising constant which is independent of $w$. Thus the degree of our belief in $w$ is controlled by the quantity
\begin{equation}
L_n(w) = - \sum_{i=1}^n \frac{1}{n} \log p(y_i|x_i,w)
\end{equation}
which is called the \emph{negative log-likelihood}, via the relationship
\begin{align}
p(w|D_n) &= \frac{1}{Z_n} \exp(-n L_n(w)) \varphi(w)\,, \label{eq:posterior_defn}\\
Z_n &= \int \exp(-n L_n(w) ) \varphi(w) dw\,. \label{eq:zn_defn}
\end{align}
Our philosophical position here is that ultimately the explanations we wish to consider are the (honest) Turing machine codes $w \in W^{\texttt{code}}$, but in executing these codes with the UTM we cannot avoid some small (perhaps extremely small) possibility of this execution being corrupted by read errors of the kind modelled by small perturbations within $W$.

Given Turing machines $M_1, M_2$ we compare codes $[M_1], [M_2] \in W^{\texttt{code}}$ by comparing the conditional probability that, given $D_n$, the explanation lies in a small neighbourhood $\mathcal{W}_i$ of $[M_i]$. These probabilities are
\begin{equation}\label{eq:pmathcalw}
p(\mathcal{W}_i \mid D_n) = \frac{1}{Z_n} \int_{\mathcal{W}_i} \exp(-n L_n(w) ) \varphi(w) dw\,.
\end{equation}
It turns out that, asymptotically as $n$ becomes large, this comparison is determined by the geometry at these points of the \emph{average negative log-likelihood}
\[
L = - \sum_{x \in I} q(x) \sum_{y \in Q} q(y|x) \log p(y|x,w)\,.
\]
We explain how this comparison works in detail in Section \ref{section:free_energy}. For now we take this as motivation to introduce the mathematical concept which captures the behaviour of a function at a point:

\begin{definition}\label{defn:germ}
A \emph{germ} of continuous (resp. analytic) real-valued functions on a space $X$ is an equivalence class of pairs $(U, f)$ where $U$ is an open neighbourhood of $x$, $f: U \lto \mathbb{R}$ is continuous (resp. analytic) and we say two such pairs are equivalent $(U, f) \sim (V, g)$ if there exists an open neighbourhood $Z$ of $x$ with $Z \subseteq U \cap V$ and $f(z) = g(z)$ for all $z \in Z$.
\end{definition}

This provides the basic setting where we associate to a Turing machine $M$ the germ $([M], L)$ of an analytic function, and some property of this germ is relevant to the question of whether or not we should prefer $M$ as an explanation of the observed data, relative to some other candidate Turing machine.
\\

We only consider \emph{deterministic} synthesis problems for which there is a \emph{true function} $y = y(x)$ from $I$ to $Q$ and the true conditional distribution $q(y|x) = \delta_{y = y(x)}$ puts all probability mass on this correct output. We suppose that two states $\{ 0, 1\} \subseteq Q$ have been identified (which we could read as for example $\reject, \accept$) and that $y(x) \in \{0,1\}$ for all $x \in I$. For $z \in \{0,1\}$ we sometimes write $\overline{z} = z + 1$ modulo $2$. 

\begin{definition}
Given a Turing machine $M$ we let
\[
M(x) := \textrm{step}^t(x,[M]) \in Q
\]
denote the state after executing $M$ for $t$ steps on input $x$.\footnote{We may also assume there is some common initialisation state $\underline{\text{init}} \in Q$.}
\end{definition}

\begin{definition} A Turing machine $M$ is a \emph{classical solution} if $M(x) = y(x)$ for all $x \in I$.
\end{definition}

\begin{remark} By commutativity of \eqref{eq:smooth_relax} to be a classical solution is equivalent to
\[
\Delta\textrm{step}^t\big(x,[M]\big)_{y(x)} = 1\,, \qquad \forall x \in I\,.
\]
\end{remark}

\begin{definition}\label{definition:polynomial_H} We define
\begin{equation}
H(w) = \sum_{x} q(x) \, p\big( y \neq y(x) \mid x, w\big)^2
\end{equation}
where $y = y(x)$ is the true output.
\end{definition}

To be explicit
\be\label{eq:model_from_stept}
p\big( y \neq y(x) \mid x, w\big) = \sum_{z \in Q, z \neq y(x)} \Delta\textrm{step}^t(x,w)_{z}\,.
\ee
For the smooth relaxation $\Delta \textrm{step}^t$ we use in this paper, the function $H$ is a polynomial function of $w$. By construction the zeros of $H$ among the classical codes $W^{\texttt{code}}$ are precisely the classical solutions. One can show \cite[Lemma B.2]{clift2021geometry} that the zeros of $H$ all lie on the boundary of the probability simplex, but we will not need that result here.

When things are set up correctly (see Section \ref{section:details_model} below) the average log-likelihood $L(w)$ of the model is comparable to $H(w)$ up to an additive constant and so it suffices for many purposes to study the geometry of the polynomial function $H(w)$.

\begin{definition} The \emph{algebraic singularity} associated to a classical solution $M$ is the function germ $([M], H)$. This singularity depends on the UTM, the smooth relaxation \eqref{eq:smooth_relax} and the time cutoff $t$.
\end{definition}

\subsection{Details of the Model}\label{section:details_model}

For somewhat technical reasons the model we actually use is a slight variation on \eqref{eq:posterior}. The reader is advised to skip this section on a first reading. Before giving the complete definition of the model, here are the two issues with the naive definition in \eqref{eq:posterior}:

\begin{itemize}
    \item \textbf{Correct vs Incorrect.} It is convenient to read the prediction from a single tape square of the UTM. Since we focus in this paper on the case of recognising languages, this should be a state of the simulated machine, as in \eqref{eq:posterior}. However, it is unnecessary to restrict the machine to two states \emph{during operation}. So we simply measure the probability assigned by the code to the \emph{correct} final state and everything else.
    \item \textbf{Moving off the Boundary.} It is convenient to analyse a polynomial $H(w)$ that is comparable to the (analytic, but not polynomial) KL divergence $K(w)$. This is possible (see Appendix \ref{sec:KL_Squared_Distance}) but requires that we move the true distribution $q$ and the binary model $p^{\bb{Z}_2}$ ``off the boundary'' of the probability simplex $\Delta \bb{Z}_2$. From a Bayesian point of view this can be thought of as insisting that zero probabilities are not realistic, so we introduce some (very small) chance of error.
\end{itemize}

Given the setting of the previous section, we define
\[
p^{\bb{Z}_2}( z | x, w ) = \begin{cases} p( z | x, w) & z = y(x)\,\\
\sum_{y \in Q, y \neq y(x)} p(y|x,w) & z \neq y(x)
\end{cases}
\]
where $p(y|x,w)$ is as in \eqref{eq:posterior}. That is, $p^{\bb{Z}_2}(z | x, w)$ is a distribution over $z \in \{0,1\}$ and
\[
p^{\bb{Z}_2}( y(x) | x, w ) = p(y(x) | x, w)\,.
\]
We regard the deterministic conditional distribution $q(y|x)$ as a distribution on $\bb{Z}_2$. For $0 < \mu < 1$ let $\varepsilon_\mu: \Delta \bb{Z}_2 \lto \Delta \bb{Z}_2$ move a probability distribution $\bold{x}$ a fraction $\mu$ towards the barycenter $\bold{b} = (\tfrac{1}{2},\tfrac{1}{2})$, i.e. the maximum entropy distribution (see Appendix \ref{section:boundary_away})
\[
\varepsilon_\mu(\bold{x}) = (1-\mu) \bold{x} + \mu \bold{b}\,.
\]
We fix some small value of $\mu$, perhaps so small that with high probability the distinction from zero is never noticed in our experiments.  The true distribution and statistical model we use in this paper are
\[
q_\mu(y|x) = \varepsilon_\mu q(y|x)\,, \qquad p^{\bb{Z}_2}_\mu(y|x,w) = \varepsilon_\mu p^{\bb{Z}_2}(y|x,w)\,.
\]
We let $L_\mu$ be the average log-likelihood defined with $q_\mu, p_\mu$.

\begin{definition}\label{defn:comparable} We say that two analytic functions $f(w),g(w)$ on an open subset $U \subseteq \mathbb{R}^d$ are \emph{comparable} (or \emph{comparable up to constants}) if there exist $c, d > 0$ such that
\[
c \cdot g(w) \le f(w) \le d \cdot g(w)\,, \qquad \forall w \in U\,.
\]
It is easy to check that this is an equivalence relation on the set of analytic functions on $U$, and we write $f \sim g$ when $f,g$ are comparable.
\end{definition}

Recall that we say the true distribution $q_\mu$ is \emph{realisable} by the model $p^{\bb{Z}_2}_\mu$ if there exists $w \in W$ with $p^{\bb{Z}_2}_\mu(y|x,w) = q_\mu(y|x)$ for all $x \in I, y \in \{0,1\}$. But since $\varepsilon_\mu$ is injective that is just to say that $p^{\bb{Z}_2}(y|x,w) = q(y|x)$ for all $x,y$ which is to say that $p(y(x)|x,w) = 1$ for all $x \in I$. Since this is independent of $\mu$ we will simply say that the \emph{true function is realisable}. In particular this is true if any classical solution exists.

\begin{lemma}\label{lemma:comparable_H_K_L} If the true function is realisable by $w^*$ then the KL divergence
\be
K_\mu(w) = \dkl( q_\mu \, || \, p^{\bb{Z}_2}_\mu(w) )
\ee
is comparable to the polynomial $\tfrac{1}{2} H(w)$ in some open neighbourhood of $w^*$. In particular, the average log-likelihood $L_\mu(w)$ is comparable to $\tfrac{1}{2} H(w) + L^\mu_0$ where $L^\mu_0$ is a constant, the entropy of the true distribution $q_\mu$.
\end{lemma}
\begin{proof}
The smoothed polynomial probabilities remain positive and normalised on a sufficiently small open neighbourhood of $w^*$ in the ambient parameter coordinates. By Lemma \ref{lemma:H_comparable_K} applied to $q_\mu, p_\mu$ we find that $K_\mu(w)$ is comparable to $\tfrac{1}{2} H_\mu(w)$ on this neighbourhood, where $H_\mu(w)$ is defined by
\[
H_\mu(w) = \sum_{x \in I} q(x) \sum_{z \in \bb{Z}_2} \Big( p^{\bb{Z}_2}_\mu( z | x, w) - q_\mu(z|x) \Big)^2\,.
\]
By Appendix \ref{section:boundary_away}, $\tfrac{1}{2}H_\mu(w)$ is equal to
\begin{align*}
\tfrac{1}{2} (1-\mu)^2 &\sum_{x \in I} q(x)\Big[ \Big( p^{\bb{Z}_2}( y(x) | x, w) - 1 \Big)^2 + p^{\bb{Z}_2}\big( \overline{y(x)} \mid x, w\big)^2 \Big]\\
&= (1-\mu)^2 \sum_{x \in I} q(x) \, p^{\bb{Z}_2}\big( \overline{y(x)} \mid x, w\big)^2
\end{align*}
which is $(1-\mu)^2 H(w)$, proving the claim.
\end{proof}

Note that the zeros of $K_\mu(w)$ are by injectivity of $\varepsilon_\mu$ the same as those of $\dkl( q || p(w) )$. In the rest of this document we write $K, L$ for $K_\mu, L_\mu$.

\begin{definition} The \emph{analytic singularities} associated to a classical solution $M$ are the function germs $([M], K)$, $([M], L)$. These singularities depend on the UTM, the smooth relaxation \eqref{eq:smooth_relax}, the time cutoff $t$ and $\mu > 0$.
\end{definition}

Note that the germ is independent of $\mu$ up to the comparability relation. Since $K, L$ differ by a constant there is no essential difference and we use both in the sequel.

\subsection{Learning Coefficients}\label{section:background_llc}

Given a compact semi-analytic subset $\Omega \subseteq \mathbb{R}^d$, an analytic real-valued function $f$ on $\Omega$ and a smooth positive function $\varphi: \Omega \lto \mathbb{R}$ we consider the zeta function
\be\label{eq:zeta_function}
\zeta(z) = \int_{\Omega} |f(w)|^z \varphi(w) dw\,.
\ee
Assume that $\Omega$ is the closure of its interior and that $f$ is not identically zero on any connected component of $\Omega$. This defines a holomorphic function on the half-plane $\operatorname{Re}(z) > 0$ in $\mathbb{C}$ and provided $f$ vanishes somewhere on $\Omega$ the zeta function can be analytically continued to a meromorphic function on all of $\mathbb{C}$ with poles on the negative real line which are rational numbers and have a greatest element \cite[\S 4.5]{watanabeAlgebraicGeometryStatistical2009}, \cite[Corollary 3.10]{lin2011algebraic}. If $-\lambda$ is the largest pole then $\lambda$ is called the \emph{real log canonical threshold} of the pair $f, \varphi$ and is denoted $\operatorname{RLCT}_{\Omega}(f ; \varphi)$.

In the setting of singular learning theory \cite{watanabeAlgebraicGeometryStatistical2009} we specify a triple $(p, q, \varphi)$ consisting of a model $p(y|x,w)$, a true distribution $q(y|x), q(x)$ and a prior distribution $\varphi(w)$ over some semi-analytic subset $W \subseteq \mathbb{R}^d$ which we refer to as parameter space. In this subsection and Section \ref{section:free_energy}, we write $p(y|x,w)$ and $q(y|x)$ for the smoothed distributions $p^{\bb{Z}_2}_\mu(y|x,w)$ and $q_\mu(y|x)$ of Section \ref{section:details_model}. The input distribution $q(x)$ still has finite support $I \subseteq \Sigma^*$, and the underlying \emph{deterministic synthesis problem} \cite{clift2021geometry} has target function $y = y(x)$. The prior $\varphi(w)$ does not play a significant role in this paper and may be chosen to be a (normalised) Lebesgue measure.

Let $K: W \lto \mathbb{R}$ denote the KL divergence associated to this choice of triple
\be
K(w) = \sum_{x \in I} q(x) \sum_{y \in \bb{Z}_2} q(y|x) \log\frac{q(y|x)}{p(y|x,w)}\,.
\ee
Note that with
\begin{align*}
L(w) = -\sum_{x \in I} q(x) \sum_{y \in \bb{Z}_2} q(y|x) \log p(y|x,w)
\end{align*}
we have $K(w) = - H(q) + L(w)$ where $H(q)$ is the conditional entropy of the output given the input.

Taking $f = K$ in \eqref{eq:zeta_function} we obtain the real log canonical threshold, which is also called the \emph{learning coefficient} \cite[\S 7]{watanabeAlgebraicGeometryStatistical2009} and denoted
\[
\lambda = \operatorname{RLCT}_W(K ; \varphi)\,.
\]
By \cite[Proposition 3.9]{lin2011algebraic} this \emph{global} learning coefficient is the minimum
\[
\lambda = \min_{x \in \mathcal{V}} \operatorname{RLCT}_{W_x}( K ; \varphi )
\]
over the set $\mathcal{V} = \{ w \in W \l K(w) = 0 \}$. Here $W_x$ is a compact semi-analytic neighbourhood of $x \in \mathcal{V}$, which by the proof of \cite[Proposition 3.9]{lin2011algebraic} may be chosen sufficiently small that the RLCT is independent of the choice of neighbourhood. When $x = [M]$ is the code of a Turing machine which is a classical solution we call this the \emph{local learning coefficient} of $M$ associated to the synthesis problem encoded by $q$.

\begin{definition}\label{defn:local_learning_coefficient}
The \emph{local learning coefficient} of $M$ for the synthesis problem $q$ is 
\be
\lambda([M], q) := \operatorname{RLCT}_{W_{[M]}}( K ; \varphi )
\ee
where $W_{[M]}$ is a sufficiently small compact semi-analytic neighbourhood of $[M]$.
\end{definition}

\begin{lemma} Let $H$ be the polynomial function of Definition \ref{definition:polynomial_H}. Then
\[
\operatorname{RLCT}_{W_{[M]}}( K ; \varphi ) = \operatorname{RLCT}_{W_{[M]}}( H ; \varphi)\,.
\]
\end{lemma}
\begin{proof}
Follows from Lemma \ref{lemma:comparable_H_K_L} and \cite[\S 7]{watanabeAlgebraicGeometryStatistical2009}.
\end{proof}

\begin{definition} We define the \emph{Newton polyhedron} of $M$ to be
\be
\mathcal{P}(M) = \operatorname{conv}\Big\{ \bold{k} \in \mathbb{N}^d \,\Big|\, \frac{\partial^{|\bold{k}|}}{ \partial w_1^{k_1} \cdots \partial w_d^{k_d}} H \Bigr|_{w = [M]} \neq 0 \Big\} + \mathbb{R}_{\ge 0}^d.
\ee
The Newton polyhedron of $M$ is the convex hull of the set of monomials $w^{\bold{k}}$ represented as points $\bold{k}$ in $\mathbb{R}^d$ with nonzero coefficients in the Taylor series expansion of $H$ at $[M]$, plus the nonnegative orthant $\mathbb{R}_{\ge 0}^d$. The \emph{distance} $l = l(M)$ of the Newton polyhedron is
\be
l(M) = \inf\{ s \ge 0 \l (s,s,\ldots,s) \in \mathcal{P}(M) \}\,.
\ee
\end{definition}

Note that by Theorem \ref{theorem:partial_formula_explicit} the Newton polyhedron encodes some information about the error syndromes with weight $\le \bold{k}$ which cause output errors in the execution of $M$.

\begin{proposition}\label{prop:upper_bound_RLCT} $\lambda([M], q) \le 1/l(M)$.
\end{proposition}
\begin{proof}
Let $I = \langle H \rangle$ be the ideal generated by $H$ in the ring of germs of analytic functions at $[M]$ (let us change coordinates so this is at the origin). Then using the notation of \cite[Proposition 4.3]{lin2011algebraic} it follows from \cite[Proposition 4.8]{lin2011algebraic} that
\[
\lambda([M], q) = \operatorname{RLCT}_0(I ; 1) \le \operatorname{RLCT}_0(\operatorname{mon}(I) ; 1)\,.
\]
Now $\operatorname{mon}(I)$ is generated by $w^{\bold{k}}$ with $\frac{\partial^{|\bold{k}|}}{ \partial w_1^{k_1} \cdots \partial w_d^{k_d}} H \bigr|_{w = [M]} \neq 0$ and hence the claim follows from \cite[Theorem 4.18]{lin2011algebraic}.    
\end{proof}

\subsection{Free Energy and Comparison of Classical Solutions}\label{section:free_energy}

In this paper we establish relations between internal structure of a Turing machine $M$ and the geometry of the associated singularities $([M], H), ([M], L), ([M], K)$. By further linking the geometry of these singularities, quantified by the local learning coefficient $\lambda([M], q)$, to the Bayesian learning process for the synthesis problem, we link the internal structure of $M$ with the behaviour of this process.

The central role played by the learning coefficient $\lambda$ in controlling the asymptotic form of the Bayesian posterior is captured by the \emph{free energy}, which is the negative logarithm of the marginal likelihood \cite{watanabeAlgebraicGeometryStatistical2009}. Let $D_n = \{ (x_i, y_i) \}_{i=1}^n$ consist of $n$ independent observations from the smoothed law $q(x,y) = q(x)q(y|x)$.\footnote{In this section $n$ means the number of samples, elsewhere in the paper it is $n = |\Sigma|$.} Here $L_n(w) = -n^{-1}\sum_{i=1}^n \log p(y_i|x_i,w)$ uses the smoothed model $p$ under the same notation convention. Given a prior $\varphi(w)$ over codes $w \in W$, the \emph{posterior} is defined by \eqref{eq:posterior_defn} with normalising constant $Z_n$. The free energy is defined to be
\begin{equation}
    F_n := -\log Z_n = -\log \int_W \exp(-n L_n(w)) \varphi(w) dw.
\end{equation}
Under some technical conditions the asymptotic behaviour of $F_n$ as $n \to \infty$ is governed by Watanabe's free energy formula \cite[\S 6.3]{watanabe2018}, which asserts that
\begin{equation}\label{eq:free_energy_formula_programs}
    F_n = n L_n(w^*) + \lambda \log n + O_p(\log \log n),
\end{equation}
where $w^* \in W$ minimises the average negative log-likelihood $L(w)$ and $\lambda$ is the learning coefficient. In this paper we assume the true distribution is realisable, so that the relevant conditions are satisfied.

This provides a natural asymptotic comparison between competing candidate explanations in the synthesis problem: given two classical solutions $M_1, M_2$ with respective codes $[M_1], [M_2] \in W^{\texttt{code}}$, the posterior probability mass contained in a small neighbourhood of each solution can be approximated in terms of the asymptotic behaviour of the free energies associated to those neighbourhoods. Our treatment follows \cite[\S 7.6]{watanabeAlgebraicGeometryStatistical2009}, \cite{watanabephasetalk}, \cite{chen2023tms1} and we refer to these references for full details. In the context of program synthesis these ideas were previously discussed in \cite{clift2021geometry, waring2021geometric}.
\\

Let $\mathcal{W}_1, \mathcal{W}_2 \subseteq W$ be sufficiently small compact semi-analytic neighbourhoods of $[M_1], [M_2]$ respectively, and let
\[
F_n(\mathcal{W}_i) := -\log \int_{\mathcal{W}_i} \exp(-n L_n(w)) \varphi(w) dw.
\]
Then the posterior mass in $\mathcal{W}_i$ is proportional to $\exp(-F_n(\mathcal{W}_i))$ and
\[
F_n(\mathcal{W}_i) = n L_n([M_i]) + \lambda([M_i], q) \log n + O_p(\log \log n)\,.
\]
Since both $M_1, M_2$ are classical solutions, their smoothed output distributions agree, so $L_n([M_1]) = L_n([M_2])$ for every dataset. Thus, in the large-sample limit, the ratio of posterior masses in $\mathcal{W}_1$ and $\mathcal{W}_2$ increasingly favours the neighbourhood of the solution with the smaller local learning coefficient. This explains why we are interested in understanding the kinds of internal structures of Turing machines $M$ which lead to lower learning coefficients $\lambda([M], q)$: the posterior will tend to prefer such machines among all other classical solutions for the same synthesis problem.

\begin{remark}
There can be nontrivial competition between the order $n$, $\log n$, $\log \log n$ and constant order terms in the asymptotic expansion of the free energy. Particularly for program synthesis we expect that some of these lower order terms may play a significant role, and the emphasis here on the learning coefficient should be viewed as an appropriate starting point for the theory but not the complete story.
\end{remark}

\begin{remark}
While the discussion above compares classical solutions, the Bayesian posterior assigns nonzero mass to any neighbourhood in $W$ with positive prior measure. In particular, a machine $M'$ which is \emph{not} a classical solution but which fits the data well on average and has a small learning coefficient $\lambda([M'], q)$ may dominate a classical solution $M$ in the posterior for small to moderate $n$, despite being wrong in the limit. 

This phenomenon is due to the competition between the loss term $n L_n([M'])$ and the complexity term $\lambda([M']) \log n$ in the free energy. Such competition can induce \emph{Bayesian transitions} in the posterior, where the dominant contribution to the marginal likelihood shifts from one region of parameter space to another as $n$ increases \cite[\S 7.6]{watanabeAlgebraicGeometryStatistical2009}, \cite{chen2023tms1}.
\end{remark}

\section{Geometry}\label{section:geometry}

Given a synthesis problem we have associated in Section \ref{section:background_iid} to a classical solution $M$ the germ $([M], L)$ of an analytic function. The claim made in the introduction is that there is some relation between the ``structure'' of $M$ and the geometry of the germ. In this section we explain in a self-contained way what we mean by this geometry, and sketch the first layer of the relation to the structure of $M$.

The germ $([M], L)$ is equivalent to the information of all the Taylor series coefficients
\begin{equation}
\frac{\partial^{|\alpha|} L}{ \partial^{\alpha_1} w_1 \cdots \partial^{\alpha_d} w_d}\Bigr|_{w = [M]}
\end{equation}
where $w_1,\ldots,w_d$ are the coordinates of the parameter space $W$ and for $\alpha = (\alpha_1,\ldots,\alpha_d) \in \mathbb{N}^d$ we write $|\alpha| = \sum_{i=1}^d \alpha_i$. Thus, roughly speaking, the geometric information comes in \emph{layers} with the $r$th layer corresponding to the information in the Taylor series coefficients with $|\alpha| \le r$. In this section we analyse the first two layers. For the higher layers we need more algebraic geometry. As explained in the previous section, we find it easier to analyse the geometric information in the Taylor series expansion of the polynomial function $H$, and deduce from this geometric information about $L, K$ (see Section \ref{section:details_model}).

\begin{remark}\label{remark:coord_arb} In this section the choice of coordinates $w_1,\ldots,w_d$ is arbitrary. Note that the $w_i$ are \emph{not} probabilities themselves, since these are not all independent. Later we will describe explicit coordinates in a neighbourhood of a Turing machine code.
\end{remark}

We assume that there is a given true function $y = y(x)$ from $I$ to $\{ 0,1 \} \subseteq Q$ and that $M$ is a classical solution. Thus $H([M]) = 0$ and $\nabla H([M]) = 0$. The second-order information in the germ is contained in the Hessian matrix
\begin{equation}
\operatorname{Hess}(H, [M]) = \Bigg( \frac{\partial^2 H}{\partial w_i \partial w_j}\Bigr|_{w = [M]} \Bigg)_{1 \le i, j \le d}\,.
\end{equation}
The rank of the Hessian does not depend on the choice of local coordinates. In the chosen coordinates, its matrix is real symmetric and therefore diagonalisable, and we study its eigenvalues and eigenvectors.

\begin{definition}\label{defn:influence} Given $\bold{i} \in \mathbb{N}^d$ the \emph{influence function}
\[
g_{\bold{i}}: \Sigma^* \lto \mathbb{R}
\]
is defined by
\begin{equation}
g_{\bold{i}}(x) = \frac{\partial^{|\bold{i}|}}{\partial w_1^{i_1} \cdots \partial w_d^{i_d}}\, p\big(y \neq y(x) \mid x, w \big)\Bigr|_{w = [M]}\,.
\end{equation}
When $\bold{i} = e_i$ is the $i$th standard basis vector we write $g_i(x)$ for $g_{\bold{i}}(x)$.
\end{definition}

In the coordinates that we will choose later at the point $[M]$ in $W$, each index $1 \le i \le d$ will correspond to a particular direction to vary a particular bit in the code $[M]$ away from its given value, and the derivative $g_i(x)$ measures the rate of change of the probability of the incorrect outputs with respect to this variation.

By the assumption that $M$ is a classical solution, we have $g_{\bold{0}}(x) = 0$ where $\bold{0} = (0,\ldots,0)$ means no derivatives.

\begin{lemma}\label{lemma:partial_derivatives_H} Let $|\bold{k}| > 0$. Then
\begin{equation}
\frac{\partial^{|\bold{k}|}}{ \partial w_1^{k_1} \cdots \partial w_d^{k_d}} H \Bigr|_{w = [M]} =\sum_{\substack{\bold{i},\bold{j} \neq 0 \\ \bold{i} + \bold{j} = \bold{k}}} C(\bold{i}, \bold{j})\, \mathbb{E}_x\Big[ g_{\bold{i}}(x) g_{\bold{j}}(x) \Big]
\end{equation}
where $C(\bold{i}, \bold{j}) = \prod_{s=1}^d \binom{i_s + j_s}{i_s}$, $\mathbb{E}_x[-]$ denotes expectation with respect to $q(x)$ and $\bold{i}, \bold{j} \in \mathbb{N}^d$.
\end{lemma}
\begin{proof}
For the duration of this proof we write
\[
g(x,w) = p\big( y \neq y(x) \mid x, w \big)\,.
\]
Recall that given functions $p(w), q(w)$ we have
\[
\frac{\partial^{|\bold{k}|}}{ \partial w_1^{k_1} \cdots \partial w_d^{k_d}}\Big[ p(w) q(w) \Big] = \sum_{0 \le \bold{i} \le \bold{k}} \Big\{ \prod_{s=1}^d \binom{k_s}{i_s} \Big\} \frac{\partial^{|\bold{i}|}}{ \partial w_1^{i_1} \cdots \partial w_d^{i_d}} p(w) \frac{\partial^{|\bold{j}|}}{ \partial w_1^{j_1} \cdots \partial w_d^{j_d}} q(w)
\]
where $\bold{i} \le \bold{k}$ means $i_s \le k_s$ for $1 \le s \le d$ and $\bold{j} = \bold{k} - \bold{i}$. Recall that $\binom{0}{0} = 1$ by convention.

Substituting $p(w) = q(w) = g(x,w)$ we have
\begin{align*}
\frac{\partial^{|\bold{k}|}}{ \partial w_1^{k_1} \cdots \partial w_d^{k_d}} H(w) &= \sum_x q(x) \sum_{\substack{\bold{i},\bold{j} \ge 0 \\ \bold{i} + \bold{j} = \bold{k}}} C(\bold{i},\bold{j}) \frac{\partial^{|\bold{i}|}}{ \partial w_1^{i_1} \cdots \partial w_d^{i_d}} g(x,w) \frac{\partial^{|\bold{j}|}}{ \partial w_1^{j_1} \cdots \partial w_d^{j_d}} g(x,w)\,.
\end{align*}
Evaluating at $w = [M]$ we obtain
\begin{align*}
\frac{\partial^{|\bold{k}|}}{ \partial w_1^{k_1} \cdots \partial w_d^{k_d}} H(w)\Bigr|_{w = [M]} &= \sum_x q(x) \sum_{\substack{\bold{i},\bold{j} \ge 0 \\ \bold{i} + \bold{j} = \bold{k}}} C(\bold{i},\bold{j}) g_{\bold{i}}(x) g_{\bold{j}}(x)\,.
\end{align*}
Since $g_{\bold{0}}(x) = 0$ we obtain the result.
\end{proof}

In particular for $1 \le i \le d$
\[
\frac{\partial}{\partial w_i} H\Bigr|_{w = [M]} = 0\,.
\]
That is, $[M]$ is a critical point of $H$, which is why we call the germ $([M], H)$ a \emph{singularity}. Another important special case is when $\bold{k} = e_i + e_j$ for $1 \le i, j \le d$. Then when $i \neq j$ we have two summands indexed by $(\bold{i}, \bold{j})$ equal to $(e_i, e_j), (e_j, e_i)$ and when $i = j$ just one summand with $\bold{i} = \bold{j} = e_i$, so in either case
\begin{equation}
\frac{\partial^2}{ \partial w_i \partial w_j} H \Bigr|_{w = [M]} = 2\, \mathbb{E}_x\Big[ g_i(x) g_j(x) \Big]\,.
\end{equation}
Let $I \subseteq \Sigma^*$ be the set of allowed inputs. If we write $P$ for the $|I| \times d$ matrix
\[
P = \big( g_i(x) \big)_{x, i}
\]
then we have shown
\begin{equation}\label{eq:hessian_formula_P}
\operatorname{Hess}(H, [M]) = 2 P^T Q P = 2 (Q^{1/2} P)^T (Q^{1/2} P)
\end{equation}
where $Q$ is the diagonal matrix with entries $q(x)$. The nonzero eigenvalues of the Hessian are twice the squares of the nonzero singular values of $Q^{1/2} P$. In particular the Hessian and $Q^{1/2} P$ have the same rank, which is the rank of $P$.

\begin{lemma}\label{lemma:slt25} Let $V$ be the space of functions in $\mathbb{R}^I$ spanned by $g_1,\ldots,g_d$. Then the rank of $\operatorname{Hess}(H, [M])$ is equal to $\dim(V)$. In particular, the Hessian is nondegenerate if and only if the $g_i$ are linearly independent.
\end{lemma}

We have seen in Section \ref{section:background} how inductive inference on Turing machines, viewed from the perspective of embedding ordinary codes into noisy codes, boils down to the comparison of integrals of the posterior over regions of parameter space \eqref{eq:pmathcalw} and how the asymptotics of those comparisons are sensitive to the geometry of the analytic function $L: W \lto \mathbb{R}$ around a code $[M]$. That is, to the geometry of the germ $([M], L)$. Studying the geometry of this germ has now lead us to the following question: what does it mean \emph{computationally} for there to be a linear dependence relation
\begin{equation}
\sum_{i=1}^d a_i g_i(x) = 0\,, \qquad \forall x \in I\,.
\end{equation}
In the rest of the paper we explore this question, which we answer in Section \ref{section:meaning_nondeg}.

\section{Derivatives and Errors for Plain Proofs}\label{section:derivatives_errors}

\subsection{Linear Logic}\label{section:linear_logic_setup}

In this paper \emph{linear logic} means \emph{first-order intuitionistic linear logic} with connectives $\otimes, \&, \multimap, !$ along with the corresponding introduction rules and cut-elimination transformations from \cite{mellies2009categorical, Benten}. Our use of linear logic follows \cite{clift2020cofree, clift2020encodings, clift2018derivatives}. For an introduction we recommend \cite{girard1987linear, mellies2009categorical, murfet2014logic}.

In this paper the role of linear logic is that by encoding the step function of a UTM into this logic, and then making use of the semantics of linear logic in vector spaces, we obtain a natural smooth relaxation of the execution of the UTM in the sense of \eqref{eq:smooth_relax}. In particular, this provides a definition of the probability $p(y \neq y(x)|x,w)$ of an output error on some input $x$ when $w$ is a noisy Turing machine. The reader who is willing to take this probabilistic model at face value is encouraged to skip this section on a first reading.
\\

Fix an algebraically closed field $\mathbbm{k}$.\footnote{We take $\mathbbm{k} = \mathbb{C}$ but all coefficients are taken to be real.} Given formulas (or synonymously, \emph{types}) $A,B$ the definition of the \emph{Sweedler semantics} \cite{murfet2014logic, murfet2015sweedler, CliftMastersThesis} of linear logic in the category of vector spaces over $\mathbbm{k}$, which we denote by $\llbracket - \rrbracket$, associates a vector space $\llbracket X \rrbracket$ to each atomic formula $X$, and to $A,B$, the vector spaces determined by the following rules:
\begin{align*}
    \llbracket A \otimes B \rrbracket &= \llbracket A \rrbracket \otimes \llbracket B \rrbracket\\
    \llbracket A \, \& \, B\rrbracket &= \llbracket A \rrbracket \oplus \llbracket B \rrbracket\\
    \llbracket A \multimap B \rrbracket &= \operatorname{Hom}_{\mathbbm{k}}(\llbracket A \rrbracket, \llbracket B \rrbracket)\\
    \llbracket !A \rrbracket &= !\llbracket A \rrbracket
\end{align*}
where $!V$ denotes the universal cocomutative coassociative coalgebra of $V$ (referred to here as simply the \emph{cofree coalgebra}). The universal morphism is usually denoted $d_V: !V \lto V$ or just $d$. An explicit description of the cofree coalgebra is \cite{murfet2015sweedler}
\begin{equation}
    !V = \bigoplus_{P \in V}\operatorname{Sym}_P(V)
\end{equation}
where $\operatorname{Sym}_P(V) = \operatorname{Sym}(V)$ is the symmetric algebra with its canonical structure as a coalgebra. If $V$ is of dimension $n$ and $e_1, \ldots, e_n$ is a basis for $V$ then as a vector space there is an isomorphism $\operatorname{Sym}(V) \cong \mathbbm{k}[e_1, \ldots, e_n]$. The denotations of proofs are defined such that if $\rho$ is a proof of $A$, written $\rho : A$, then $\den{\rho} \in \den{A}$. In particular $\varphi: A \vdash B$ then $\den{\varphi}: \den{A} \lto \den{B}$ is a linear transformation. We sometimes write $\varphi(\rho) : B$ for the cut of $\varphi, \rho$ and think in this way of $\varphi$ as defining a function from cut-elimination equivalence classes of proofs of $A$ to such equivalence classes of proofs of $B$. If
\[
\psi: {!} A_1, \ldots, {!} A_r \vdash B
\]
is a proof, and $\rho_i : A_i$, then we also sometimes write
\be\label{eq:eval_proof}
\psi(\rho_1,\ldots,\rho_r) : B
\ee
for the proof obtained by cutting $\psi$ against the promotions of the $\rho_i$ to proofs of $\vdash {!} A_i$.

These denotations of proofs are constructed using the structural maps of the functors $\otimes, \&, \multimap$ as well as ${!}$. We briefly recall the latter. Given $v_1, \ldots, v_s \in V$, the corresponding vector in $\operatorname{Sym}_P(V)$ is written using a ket
\begin{equation}
    \ket{v_1, \ldots, v_s}_P := v_1 \otimes \ldots \otimes v_s \in \operatorname{Sym}_P(V).
\end{equation}
The identity element $1 \in \operatorname{Sym}_P(V)$ is denoted by $\ket{\emptyset}_P$. For a subset $I = \{i_1, \ldots, i_p\}$ we denote by $v_I$ the sequence $v_{i_1}, \ldots, v_{i_p}$ and $I^c$ is the complement of $I$. With this notation, the universal map $d$ is defined by
\begin{equation}
    d\ket{\emptyset}_P = P,\quad d\ket{v}_P = v,\quad d\ket{v_1, \ldots, v_s}_P = 0\quad s>1
\end{equation}
and the comultiplication on ${!}V$ is defined by
\begin{equation}
    \Delta \ket{v_1, \ldots, v_s}_P = \sum_{I \subseteq \{1, \ldots, s\}} \ket{v_I}_P \otimes \ket{v_{I^{c}}}_P
\end{equation}
where $I$ ranges over all subsets including the empty set. In particular
\begin{equation}
    \Delta\ket{\emptyset}_P = \ket{\emptyset}_P \otimes \ket{\emptyset}_P.
\end{equation}
We write $A^n = A \,\& \, \ldots \, \& \, A$ where there are $n$ copies of $A$. %

\subsubsection{Proofs}

A \emph{proof} in linear logic is a decorated tree that represents the application of deduction rules to instances of axioms; see \cite{girard1987linear,mellies2009categorical}. Following the Brouwer-Heyting-Kolmogorov interpretation \cite{troelstra2011history} of intuitionistic proof, or the Curry-Howard correspondence \cite{howard1980formulae}, we can think of proofs in intuitionistic linear logic as algorithms which take inputs in the form of proofs of their hypotheses and produce an output which is the proof of their conclusion.

In general a proof in linear logic, viewed as an algorithm, will use a different number of copies of its inputs depending on other inputs. For example the multiplication of Church numerals iterates one of its inputs a number of times specified by another of its inputs. Such a proof is \emph{not} plain, in the sense of the following definition. The defining characteristic of a plain proof is that the number of copies of each input required to compute the output is a \emph{property of the proof itself} and does not depend on the inputs. %

Throughout in the hypotheses $n\, A$ means $\overset{n \text{ times}}{\overbrace{A,\ldots,A}}$.

The following definition is adapted from \cite[Definition 3.1]{clift2020encodings}.

\begin{definition}
\label{def:plain}
    A proof of a sequent ${!}A_1, \ldots, {!} A_r \vdash B$ for $r \geq 0$ is \emph{plain} if it is equivalent under cut-elimination to
    \begin{center}
\AxiomC{$\proofvdots{\pi}$}
\noLine\UnaryInfC{$n_1\, A_1,\ldots,n_r \, A_r \vdash B$}
\RightLabel{\scriptsize der}
\doubleLine\UnaryInfC{$n_1\, {!}A_1, \ldots, n_r \, {!} A_r \vdash B$}
\RightLabel{\scriptsize ctr/wk}
\doubleLine\UnaryInfC{${!}A_1,\ldots,{!}A_r \vdash B$}
\DisplayProof
\end{center}
    for some proof $\pi$ and non-negative integers $n_1, \ldots, n_r$, where for $n_i > 1$ in the final step there is a corresponding contraction, and if $n_i = 0$ the final step involves a weakening. We refer to $n_i$ as the \emph{degree} of $A_i$.
\end{definition}

Given such a presentation of a plain proof $\psi$, we call $\pi$ its \emph{linear part}.

\begin{remark}
Plain proofs may have degree zero in hypotheses introduced by weakening. The later formulas involving $\Delta^{n-1}$ are written for $n\geq1$; when $n=0$, the copying and dereliction maps are replaced by the coalgebra counit, which interprets weakening.
\end{remark}

A proof is \emph{component-wise plain} if it is equivalent under cut-elimination to one obtained by tensoring the promotions of plain proofs with common exponential hypotheses and contracting the repeated inputs; see \cite[Definition 3.5]{clift2020encodings}.

\begin{example}
    Many simple data types such as booleans, integers and binary integers have a natural encoding as types in linear logic. Some of these are reviewed in \cite[\S 3]{CliftMastersThesis}, among which are the Church numerals $\underline{m}$, $m \geq 0$, (for some given formula $A$) of type
    \[
    \underline{m}: \textbf{int}_A = {!}(A \multimap A) \multimap (A \multimap A)
    \]
    and the integers $\underline{i}$, $0 \le i < n$ for fixed $n >0$ of generalised Boolean type
\[
    \underline{i}: {}_n\textbf{bool}_A = A^n \multimap A.
\]
These proofs are all plain. Generally we write types with bold letters, and denote proofs with underlines. The proof $\underline{i}$ encodes a projection onto one of $n$ factors.
\end{example}

\subsubsection{Denotations of Plain Proofs}

The denotation of a plain proof $\psi$ is defined by a set of polynomials $f^\tau_\psi$ which will play an important role in the sequel and which we now define. Here $\psi: !A_1, \ldots, !A_r \vdash B$ denotes a plain proof with linear part $\pi: n_1 A_1, \ldots, n_r A_r \vdash B$ as in Definition \ref{def:plain}. Suppose given finite sets of proofs $\mathcal{P}_i$ of $A_i$ and $\mathcal{Q}$ of $B$. Assume that $\{\llbracket \nu \rrbracket\}_{\nu \in \mathcal{Q}}$ is linearly independent in $\llbracket B \rrbracket$ and that
\begin{equation}
    \{\pi(X_1, \ldots, X_r) \mid X_i \in \mathcal{P}_i^{n_i}\} \subseteq \mathcal{Q}.
\end{equation}
Our polynomials have variables $\{ x^i_\rho \}_{1 \le i \le r, \rho \in \mathcal{P}_i}$.

The way denotations work is that
\begin{align*}
\den{\psi}\Big( \ket{\emptyset}_{\den{\rho_1}} \otimes \cdots \otimes \ket{\emptyset}_{\den{\rho_r}} \Big) &= \den{ \psi( \rho_1, \ldots, \rho_r ) }\\
&= \den{ \pi\big( \overset{n_1 \text{ times}}{\overbrace{\rho_1, \ldots, \rho_1}}, \ldots, \overset{n_r \text{ times}}{\overbrace{\rho_r, \ldots, \rho_r}} \big) }
\end{align*}
where $\psi( \rho_1, \ldots, \rho_r )$ recall is defined by cutting $\psi$ against the promotions of the $\rho_i: A_i$ as in \eqref{eq:eval_proof}. However there is more information in the denotation $\den{\psi}$ than is in this formula, and it is this information that is surfaced when we take derivatives.

\begin{definition}\label{defn:f_tau}
    Given $\tau \in \mathcal{Q}$ we define
    \be
    f_\psi^\tau = \sum_{\gamma_1, \ldots, \gamma_r}\delta_{\tau = \operatorname{eval}(\pi, \gamma)}\prod_{i = 1}^r \prod_{j = 1}^{n_i}x_{\gamma_i(j)}^{i}
    \ee
    where $\gamma_i$ ranges over all functions $\{1, \ldots, n_i\} \lto \mathcal{P}_i$, and $\operatorname{eval}(\pi, \gamma)$ means
\begin{equation}
    \operatorname{eval}(\pi, \gamma) = \pi(\gamma_1(1),\ldots, \gamma_1(n_1), \ldots, \gamma_r(1), \ldots, \gamma_r(n_r)),
\end{equation}
\end{definition}

The next lemma says that $\den{\psi}$ is computed by these polynomials $f^\tau_\psi$. Given a set $\mathcal{A}$ we denote by $\mathbbm{k} \mathcal{A}$ the free vector space on the set, that is, the vector space with the elements of $\mathcal{A}$ as a basis.

\begin{lemma}
\label{lem:F_function}
There exists a unique function $F_\psi$ making the diagram
\[\begin{tikzcd}
	{!\llbracket A_1 \rrbracket \otimes \ldots \otimes !\llbracket A_r \rrbracket} & {\llbracket B \rrbracket} \\
	{\mathbbm{k}\mathcal{P}_1 \times \ldots \times \mathbbm{k}\mathcal{P}_r} & {\mathbbm{k}\mathcal{Q}}
	\arrow["{\llbracket \psi \rrbracket}", from=1-1, to=1-2]
	\arrow["\iota", from=2-1, to=1-1]
	\arrow["{F_\psi}", from=2-1, to=2-2]
	\arrow["{\llbracket - \rrbracket}"', from=2-2, to=1-2]
\end{tikzcd}\]
    commute, where $\iota$ maps an element
    \begin{equation}
    \label{eq:random_element_F}
        \Big(\sum_{\rho \in \mathcal{P}_1}\lambda_{\rho}^1\rho, \ldots, \sum_{\rho \in \mathcal{P}_r}\lambda_{\rho}^r\rho\Big) \in \mathbbm{k}\mathcal{P}_1 \times \ldots \times \mathbbm{k}\mathcal{P}_r
    \end{equation}
    to the tensor
    \begin{equation}
        \ket{\emptyset}_{\sum_{\rho \in \mathcal{P}_1}\lambda_\rho^1\llbracket \rho \rrbracket} \otimes \ldots \otimes \ket{\emptyset}_{\sum_{\rho \in \mathcal{P}_r}\lambda_\rho^r\llbracket \rho \rrbracket}\,.
    \end{equation}
    The function $F_\psi$ maps $\eqref{eq:random_element_F}$ to
    \begin{equation}
        \sum_{\tau \in \mathcal{Q}}f^\tau_\psi\Bigr\vert_{x_\rho^i = \lambda_\rho^i}\cdot \tau\,.
    \end{equation}
\end{lemma}
\begin{proof}
    See \cite[Proposition 3.8]{clift2020encodings}.
\end{proof}

The papers \cite{clift2018derivatives}, \cite{clift2020cofree}, \cite{clift2020encodings} give an interpretation of the polynomials $f_\psi^\tau$ using a probabilistic semantics. Recall from Definition \ref{def:simplex} that associated to every finite set $Z$ there is a simplex $\Delta Z$. The following is \cite[Proposition 5.3]{clift2018derivatives}.

\begin{proposition}\label{prop:delta_psi}
    There is a unique function
    \begin{equation}
        \Delta \psi: \Delta \mathcal{P}_1 \times \ldots \times \Delta \mathcal{P}_r \lto \Delta \mathcal{Q}
    \end{equation}
    which makes the diagram
    \begin{equation}\label{eq:delta_psi_square}
        \begin{tikzcd}
            !\llbracket A_1 \rrbracket \otimes \ldots \otimes !\llbracket A_r \rrbracket\arrow[r,"{\llbracket \psi \rrbracket}"] & \llbracket B \rrbracket\\
            \Delta \mathcal{P}_1 \times \ldots \times \Delta \mathcal{P}_r\arrow[r,"{\Delta \psi}"]\arrow[u,"{\inc}"] & \Delta \mathcal{Q}\arrow[u, swap, "{\llbracket - \rrbracket}"]
        \end{tikzcd}
    \end{equation}
    commute.
\end{proposition}
That is, probability distributions are sent to probability distributions by $\llbracket \psi \rrbracket$. By uniqueness of such a function $\Delta \psi$ has the same equation as $F_\psi$. Note that $\Delta \cat{P} \subseteq \mathbbm{k} \cat{P}$ and so the left hand vertical arrow in \eqref{eq:delta_psi_square} is a restriction of the $\iota$ in Lemma \ref{lem:F_function}.

In the next section we explain how to think about this function $\Delta \psi$ and what it means to propagate uncertainty through $\psi$ (which we mean evaluate $\Delta \psi$ on input distributions, representing uncertainty about proofs in $\mathcal{P}_i$). See also \cite[Section 5.1]{clift2018derivatives}.

\subsection{Evaluating Programs in Vector Spaces}\label{section:eval_prog_vec}

From the point of view of linear logic the inputs to a Turing machine (encoded by its step function) are represented as proofs, and these inputs are ``distributed'' through the steps of the Turing machine by contraction rules, to eventually be ``unpackaged'' by dereliction rules and consumed by linear parts of the proof. Semantically, the contraction or copying step is represented by comultiplication $\Delta$ on the coalgebra ${!} \den{A}$ and the dereliction step by the counit $d: {!} \den{A} \lto \den{A}$. For background on the cofree coalgebra see \cite{murfet2014logic, murfet2015sweedler, clift2020cofree}.

For example if $\rho : A$ is a proof representing such an input (say some bit in the code of a Turing machine on the description tape of a UTM) then a pair of copies are made via
\be
\Delta \vacu_{\den{\rho}} = \vacu_{\den{\rho}} \otimes \vacu_{\den{\rho}}\,.
\ee
At the point where these copies are to be used, they are ``unpackaged'' by derelictions
\be\label{eq:d_otimes_2}
(d \otimes d) \Delta \vacu_{\den{\rho}} = \den{\rho} \otimes \den{\rho}\,.
\ee
Since the comultiplication is co-associative, we have $(1 \otimes \Delta)\Delta = (\Delta \otimes 1)\Delta$ and following standard practice we denote both maps by $\Delta^2$. Similarly, $\Delta^n$ denotes all possible ways of arranging $n$ comultiplications, one after another (which are the same linear map). Thus to make and consume $n$ copies we are computing
\be
d^{\otimes n} \Delta^{n-1} \vacu_{\den{\rho}} = d^{\otimes n} \Big( \vacu_{\den{\rho}} \otimes \cdots \otimes \vacu_{\den{\rho}} \Big) = \den{\rho} \otimes \cdots \otimes \den{\rho}\,.
\ee
A plain proof $\psi: {!} \vdash B$ is, by definition, a linear part $\pi$ below which are derelictions and contractions (and possibly weakenings). Thus its denotation is, roughly speaking and accounting for only one input of type $A$, of the form $\den{\pi} d^{\otimes n} \Delta^{n-1}$. Thus we may compute $\den{\psi}\vacu_{\den{\rho}}$ by applying the denotation $\den{\pi}$ of the linear part to the above:
\be\label{eq:normal_evaluation}
\den{\psi} \vacu_{\den{\rho}} = \den{\pi} d^{\otimes n} \Delta^{n-1} \vacu_{\den{\rho}} = \den{\pi(\rho,\ldots,\rho)}\,.
\ee
To compute derivatives of denotations of proofs is to provide inputs other than vacuums $\vacu_{\den{\rho}}$ \cite[Corollary 4.5]{clift2018derivatives}. This has an easy to understand effect on the set of copies that are provided throughout the proof:
\begin{align}
d^{\otimes n} \Delta^{n-1} \ket{\den{\nu}}_{\den{\rho}} &= d^{\otimes n} \Big( \sum_{i=1}^n \vacu_{\den{\rho}} \otimes \cdots \otimes \underbrace{\ket{\den{\nu}}_{\den{\rho}}}_{i\text{th}} \otimes \cdots \otimes \vacu_{\den{\rho}} \Big) \nonumber\\
&= \sum_{i=1}^n \den{\rho} \otimes \cdots \otimes \underbrace{\den{\nu}}_{i\text{th}} \otimes \cdots \otimes \den{\rho}\,.
\end{align}
That is, among all the $n$ copies of $\rho$ that are produced to be unpackaged and consumed in the proof, we substitute exactly one copy with $\nu$. This is the central idea in Ehrhard-Regnier's notion of derivative \cite{ehrhard2003differential}. If we now apply a proof $\pi$ to this
\be
\den{\pi} d^{\otimes n} \Delta^{n-1} \ket{\den{\nu}}_{\den{\rho}} 
= \sum_{i=1}^n \den{\pi(\rho,\ldots,\nu,\ldots,\rho)}\,.
\ee
Hence
\be
\den{\pi} d^{\otimes n} \Delta^{n-1} \ket{\den{\nu} - \den{\rho}}_{\den{\rho}} 
= \sum_{i=1}^n \big\{ \den{\pi(\rho,\ldots,\nu,\ldots,\rho)} - \den{\pi(\rho,\ldots,\rho,\ldots,\rho)} \big\} \,.
\ee
Now suppose that we have a linear combination $\sum_\rho x_\rho \rho$. Then
\begin{align*}
d^{\otimes n} \Delta^{n-1} \vacu_{\sum_\rho x_\rho \den{\rho}} &= d^{\otimes n} \Big( \vacu_{\sum_\rho x_\rho \den{\rho}} \otimes \cdots \otimes \vacu_{\sum_\rho x_\rho \den{\rho}} \Big)\\
&= \sum_{\rho_1,\ldots,\rho_n} x_{\rho_1} \cdots x_{\rho_n} \den{\rho_1} \otimes \cdots \otimes \den{\rho_n}\,.
\end{align*}
If these copies are to be consumed in the linear part $\pi$ of the overall component-wise plain proof $\psi$ (which is roughly $\pi d^{\otimes n} \Delta^{n-1}$) then we end up with an output
\begin{align}
\den{\psi} \vacu_{\sum_\rho x_\rho \den{\rho}} &= \den{\pi} d^{\otimes n} \Delta^{n-1} \vacu_{\sum_\rho x_\rho \den{\rho}} \nonumber\\
&= \sum_{\rho_1,\ldots,\rho_n} x_{\rho_1} \cdots x_{\rho_n} \den{\pi(\rho_1,\ldots,\rho_n)} \nonumber\\
&= \sum_\tau \Big\{ \sum_{\rho_1,\ldots,\rho_n} \delta_{\tau = \pi(\rho_1,\ldots,\rho_n)} x_{\rho_1} \cdots x_{\rho_n} \Big\} \den{\tau} \nonumber\\
&= \sum_\tau f^\tau_\psi \den{\tau} \label{eq:compute_psi_general}
\end{align}
where $f^\tau_\psi$ is the polynomial of Definition \ref{defn:f_tau} in the special case where there is a single input and $\rho$ runs over the proofs chosen in $\cat{P}$. Note that if $\sum_\rho x_\rho \rho$ is a probability distribution, then we can define $p(\tau) := f^\tau_\psi$ in which case
\be\label{eq:comp_path_discover}
p(\tau) = \sum_\gamma p(\tau|\gamma) p(\gamma)
\ee
where $p(\tau|\gamma) := \delta_{\tau = \pi(\rho_1,\ldots,\rho_n)}$ and $p(\gamma) := x_{\rho_1} \cdots x_{\rho_n}$. Here $\gamma$ ranges over sequences of proofs $\rho_1,\ldots,\rho_n$. Note that the linear part $\pi$ requires $n$ inputs, and when we evaluate $\psi$ on a linear combination (as opposed to the denotation of a proof) we can end up passing the linear part sequences other than $\rho,\ldots,\rho$ as in the normal evaluation \eqref{eq:normal_evaluation}.

Thus the polynomials $f^\tau_\psi$, which are the central algebraic objects, can be thought of as presenting the probabilities of output states $\tau$ by marginalising out some random variable $\gamma$ which represents a sequence of $n$ choices of proofs of type $A$. This random variable $\gamma$ is the central object of the present paper, and we explain in the next section how to think of it as an \emph{error syndrome}.

\subsection{Error Syndromes for Plain Proofs}\label{section:error_syndromes}

\subsubsection{Special Case}

We begin with the case $\psi: {!} A \vdash B$ where $\psi$ has a single input with linear part $\pi: n A \vdash B$ and we restrict to the case where our set of allowed proofs is $\cat{P} = \{ 0, 1 \}$. We consider a situation in which $0$ is the ``correct'' input and $1$ is an error. So the correct output is
\[
\psi(0) = \pi(0, \ldots, 0)
\]
and an error syndrome $\gamma$ is by definition an assignment to every computation path through the proof of either a $0$ (meaning that this path experiences no error) or $1$ (meaning that for this path there is an error). This explains the term \emph{error syndrome} which is the standard terminology for such a pattern in the theory of error correcting codes.

\begin{theorem}\label{theorem:error_count} We have
\be
\den{\psi} \ket{ \den{1} - \den{0} }_{\den{0}} = \sum_{\tau \neq \psi(0)} A_\tau \Big\{ \den{\tau} - \den{\psi(0)} \Big\}
\ee
where $A_\tau$ is a non-negative integer defined for $\tau \neq \psi(\underline{0})$ by
\[
A_\tau = \Big| \big\{ 1 \le i \le n \l \pi(0, \ldots, \overset{i}{1}, \ldots, 0) = \tau \big\} \Big|
\]
the count of the \emph{number of weight one error syndromes that produce $\tau$}.
\end{theorem}
\begin{proof}
    We consider a linear combination of these denotations in $\den{A}$ with coefficients $x_0,x_1$ and then differentiate the output probabilities obtained using \eqref{eq:compute_psi_general}:
\begin{align*}
    \Big\{ \frac{\partial}{\partial x_1} - \frac{\partial}{\partial x_0} \Big\} & \den{\psi} \vacu_{x_0 \den{0} + x_1 \den{1}}\Bigr|_{x_0=1,x_1=0}\\
    &\qquad = \sum_\tau \sum_{\rho_1,\ldots,\rho_n} \delta_{\tau = \pi(\rho_1,\ldots,\rho_n)} \Big\{ \frac{\partial}{\partial x_1} - \frac{\partial}{\partial x_0} \Big\}\Big( x_{\rho_1} \cdots x_{\rho_n} \Big) \Bigr|_{x_0=1,x_1=0} \den{\tau}\\
    &\qquad = \sum_{i=1}^n \sum_\tau \sum_{\rho_1,\ldots,\rho_n} \delta_{\tau = \pi(\rho_1,\ldots,\rho_n)} \Big\{ x_{\rho_1} \cdots \frac{\partial}{\partial x_1}( x_{\rho_i} ) \cdots x_{\rho_n} - \\
    &\qquad\qquad x_{\rho_1} \cdots \frac{\partial}{\partial x_0}( x_{\rho_i} ) \cdots x_{\rho_n} \Big) \Bigr|_{x_0=1,x_1=0} \den{\tau}\\
    &\quad = \sum_{i=1}^n \Big\{ \den{\pi(0,\ldots,\overset{i}{1},\ldots,0)} - \den{\psi(0)} \Big\}\\
    &\quad = \sum_{\tau \neq \psi(0)} A_\tau \Big\{ \den{\tau} - \den{\psi(0)} \Big\}\,.
\end{align*}
Note that under the prevailing hypotheses the set $\{ \den{\tau} \}_\tau$ is linearly independent, hence $\{ \den{\tau} - \den{\psi(0)} \}_{\tau \neq \psi(0)}$ is linearly independent. The conclusion follows from \cite[Corollary 4.5]{clift2018derivatives} which shows that the quantity calculated above is $\den{\psi} \ket{ \den{1} - \den{0} }_{\den{0}}$.
\end{proof}

In general we can express the derivatives of $f^\tau_\psi$ in terms of combinatorics of higher-weight error syndromes that produce $\tau$.

\subsubsection{General Case}\label{section:error_syndrome_general_case}

Let $\psi$ be a plain proof of a sequent ${!}A_1, \ldots, {!} A_r \vdash B$ as in Definition \ref{def:plain} with linear part $\pi: n_1 A_1, \ldots, n_r A_r \vdash B$ as in the following diagram:
\begin{center}
\AxiomC{$\proofvdots{\pi}$}
\noLine\UnaryInfC{$n_1\, A_1,\ldots,n_r \, A_r \vdash B$}
\RightLabel{\scriptsize der}
\doubleLine\UnaryInfC{$n_1\, {!}A_1, \ldots, n_r \, {!} A_r \vdash B$}
\RightLabel{\scriptsize ctr/wk}
\doubleLine\UnaryInfC{${!}A_1,\ldots,{!}A_r \vdash B$}
\DisplayProof
\end{center}
For $1 \leq i \leq r$ we assume given a finite set $\mathcal{P}_i$ of proofs of $A_i$ of size $p_i + 1 = |\mathcal{P}_i|$ and a set of proofs $\mathcal{Q}$ of $B$ satisfying the hypotheses of Lemma \ref{lem:F_function}. We further assume given an enumeration of these proofs via bijections
\be\label{eq:enumeration_Pi}
\{ 0, 1, \ldots, p_i \} \overset{\cong}{\longrightarrow} \mathcal{P}_i
\ee
for each $1 \le i \le r$ and
\be\label{eq:enumeration_Q}
\{ 0, 1, \ldots, |\mathcal{Q}| - 1 \} \overset{\cong}{\longrightarrow} \mathcal{Q}\,.
\ee
Even if $A_i = A_j$ for $i \neq j$ and $\mathcal{P}_i = \mathcal{P}_j$ as sets, we still allow the enumerations of these sets to be different for $i$ and $j$. By abuse of notation we treat these enumerations as equalities, and refer to an element of $\mathcal{P}_i$ by its corresponding integer $0 \le j \le p_i$. We further assume these enumerations are chosen so that $\psi(0,\ldots,0) = 0$. Following \cite{clift2018derivatives} we write
\[
\Big( \sum_{j=0}^{p_i} x^i_j \cdot j \Big)_{i=1}^r \in \prod_{i=1}^r \Delta \mathcal{P}_i
\]
for a generic input to the map $\Delta \psi$ of Proposition \ref{prop:delta_psi}. Note that $\Delta \cat{P}_i \subseteq \mathbb{R}^{p_i+1}$ is a manifold with corners, where we give the ambient space the coordinates $\{ x^i_j \}_{j=0}^{p_i}$. The distribution in $\Delta \cat{P}_i$ concentrated at $0$ we simply denote by $\bold{0}$, that is,
\[
\bold{0} = (1,0,\ldots,0) \in \Delta \cat{P}_i\,.
\]
We write $\bold{0}$ for the sequence $\big( \bold{0} \big)_{i=1}^r$ and $\bold{0} = (1,0,\ldots,0) \in \Delta \mathcal{Q}$. We are interested in
\[
T_{\bold{0}}(\Delta \psi): T_{\bold{0}}\Big( \prod_{i=1}^r \Delta \mathcal{P}_i \Big) \lto T_{\bold{0}}\big( \Delta \mathcal{Q} \big)
\]
and higher-order derivatives.

The linear part $\pi$ takes $n_i$ inputs of type $A_i$ for $1 \le i \le r$. These correspond to $n_i$ different \emph{computation paths} in which different copies of the $i$th input are used differently in the algorithm represented by the proof.\footnote{This will become more explicit in the special case of the encoding of the pseudo-UTM step function below, but the cut-elimination algorithm of \cite[Proposition 3.6]{clift2020encodings} provides a general notion of computation path when $\psi$ is the cut of multiple component-wise plain proofs.} On each of those paths we have $p_i$ distinct \emph{types} of error (the elements of $\mathcal{P}_i$ other than $0$). An error syndrome, as defined below, assigns to each computation path either a $0$ (no error) or a particular type of error. In what follows $|S|$ denotes the cardinality of a finite set $S$.

\begin{definition}\label{defn:error_syndrome} An \emph{error syndrome} $\gamma$ is a sequence of functions $\gamma_1,\ldots,\gamma_r$ where
\[
\gamma_i: \{ 1, \ldots, n_i \} \lto \{ 0, \ldots, p_i \} \qquad 1 \le i \le r\,.
\]
Hence $|\gamma^{-1}_i(j)|$ for $1 \le i \le r, 1 \le j \le p_i$ is the number of errors of type $j$ that occur in input $i$ in the syndrome. The total number of errors in the $i$th input is
\[
|\gamma_i| = \Big| \big\{ 1 \le c \le n_i \l \gamma_i(c) \neq 0 \big\} \Big|\,.
\]
The \emph{weight} of the error syndrome is
\[
\operatorname{wt}(\gamma) = (\bold{s}^1,\ldots,\bold{s}^r) \in \prod_{i=1}^r \mathbb{N}^{p_i}
\]
where
\[
\bold{s}^i = ( s^i_1, \ldots, s^i_{p_i} )\,, \qquad s^i_j = |\gamma_i^{-1}(j)|\,.
\]
\end{definition}

\begin{figure}[t]
    \centering
    \includegraphics
        [width=0.5\textwidth]
        {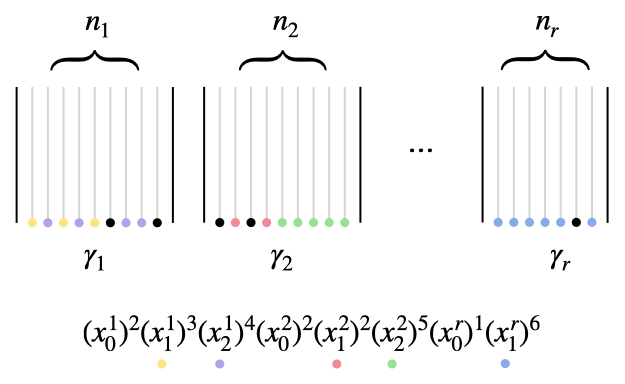}
    \caption{\label{fig:syndrome_diagram}%
        The plain proof $\psi$ has $r$ inputs, and $n_i$ computation paths leading from the $i$th input. In this depiction of an error syndrome $\gamma$, integers $1 \le j \le p_i$ are represented as colours with a different set of colours (proofs) for each $i$. Here for example $\gamma_1$ is the function corresponding to the integer sequence $1,2,1,2,1,0,2,2,0$. The corresponding factor in the monomial associated to $\gamma$ in $f^\tau_\psi$ is $(x^1_0)^2(x^1_1)^3(x^1_2)^4$ with $x^i_0$ factors being associated to ``no error'' (the integer $0$). Note that different orderings of the integer sequence represent distinct error syndromes but contribute the same monomial; this is the origin of coefficients in $f^\tau_\psi$ other than $1$.
    }
\end{figure}

\begin{definition}\label{defn:eval_error_syndrome}
Let $\gamma$ be an error syndrome. Then we define
\be\label{eq:eval_pi_gamma}
\operatorname{eval}(\pi, \gamma) = \pi\Big( \gamma_1(1), \ldots, \gamma_1(n_1),
\gamma_2(1), \ldots, \gamma_2(n_2), \ldots
\gamma_r(1), \ldots, \gamma_r(n_r)\Big)\,.
\ee
\end{definition}

\begin{theorem}\label{theorem:main_error_syndrome_general} Let $\bold{k}^i \in \mathbb{N}^{p_i}$ for $1 \le i \le r$ and set $\bold{k} = (\bold{k}^1,\ldots,\bold{k}^r)$. Then
\be\label{eq:main_comp1}
\dntn{\psi}\Big( \bigotimes_{i=1}^r \Big| \big(\den{1} - \den{0}\big)^{\otimes k^i_1}, \ldots, \big(\den{p_i} - \den{0}\big)^{\otimes k^i_{p_i}} \Big\rangle_{\den{0}} \Big) = \sum_{\tau \in \mathcal{Q}} \sum_{\bold{0} \le \bold{s} \le \bold{k}} S^{\bold{k}}_{\bold{s}} A^{\bold{s}}_\tau  \den{\tau}
\ee
where $A^{\bold{s}}_\tau$ is a non-negative integer defined by
\be\label{eq:defn_A_s_tau}
A^{\bold{s}}_\tau = \Big| \big\{ \gamma \l \operatorname{wt}(\gamma) = \bold{s} \text{ and } \operatorname{eval}(\pi, \gamma) = \tau \big\} \Big|
\ee
the count of the \emph{number of weight $\bold{s}$ error syndromes that produce $\tau$}. Here $\bold{s} \le \bold{k}$ means $s^i_j \le k^i_j$ for all $1 \le i \le r$, $1 \le j \le p_i$ and
\be\label{eq:s_coeff_theorem}
S^{\bold{k}}_{\bold{s}} = \prod_{i=1}^r \Big\{ \delta_{n_i \ge |\bold{k}^i|} (-1)^{|\bold{k}^i| - |\bold{s}^i|} \frac{(n_i-|\bold{s}^i|)!}{(n_i-|\bold{k}^i|)!} \prod_{j=1}^{p_i} \frac{k^i_j!}{(k^i_j - s^i_j)!} \Big\}
\ee
where $|\bold{s}^i| = \sum_j s^i_j$, $|\bold{k}^i| = \sum_j k^i_j$. Throughout this subsection, a term containing delta factors is understood to be zero whenever one of their conditions fails; factorial quotients are evaluated only when all those conditions hold.
\end{theorem}
\begin{proof}
In a neighbourhood of $\bold{0} \in \Delta \cat{P}_i$ we can parametrise the probability simplex via a smooth immersion of manifolds with corners
\begin{gather}
\iota_i: [0,\varepsilon)^{p_i} \lto \Delta \cat{P}_i\\
(w^i_1,\ldots,w^i_{p_i}) \mapsto (1-\sum_{j=1}^{p_i} w^i_j, w^i_1,\ldots,w^i_{p_i}) \label{eq:w_coord_simplex}
\end{gather}
with
\[
T_{\bold{0}}(\iota)\Big( \frac{\partial}{\partial w^i_j} \Big) = \frac{\partial}{\partial x^i_j} - \frac{\partial}{\partial x^i_0} \in T_{\bold{0}}\big( \Delta \mathcal{P}_i \big)
\]
for $1 \le j \le p_i$. The composite
\[
\xymatrix@C+2pc{
\prod_{i=1}^r [0, \varepsilon)^{p_i} \ar[r]^{\prod_i \iota_i} & \prod_{i=1}^r \Delta \mathcal{P}_i \ar[r]^{\Delta \psi} & \Delta \mathcal{Q}
}
\]
has derivative, writing $\iota = \prod_i \iota_i$ and letting $\bold{k}^i \in \mathbb{N}^{p_i}$
\begin{align}
\prod_{i=1}^r \frac{\partial^{|\bold{k^i}|}}{\partial (w^i_1)^{k^i_1} \cdots \partial (w^i_{p_i})^{k^i_{p_i}}} ( \Delta \psi \circ \iota )\Bigr|_{w = \bold{0}} &= \sum_{\tau \in \mathcal{Q}} \prod_{i=1}^r \frac{\partial^{|\bold{k^i}|}}{\partial (w^i_1)^{k^i_1} \cdots \partial (w^i_{p_i})^{k^i_{p_i}}} f^\tau_{\psi} \Bigr|_{w = \bold{0}} \cdot \tau \label{eq:theorem_eq_1}
\end{align}
Abbreviating $x^1_0 = 1, x^1_1 = 0, \ldots, x^2_0 = 1, \ldots$ to $\bold{x}_0=1,\bold{x}_{\neq 0}=0$ we can write the coefficient of $\tau$ in this expression as
\begin{align}
\prod_{i=1}^r \prod_{j=1}^{p_i} \Big[ \frac{\partial}{\partial x^i_j} &- \frac{\partial}{\partial x^i_0} \Big]^{k^i_j} f^{\tau}_\psi \Bigr|_{\bold{x}_0=1,\bold{x}_{\neq 0}=0}\\
&= \prod_{i=1}^r \prod_{j=1}^{p_i} \sum_{s = 0}^{k^i_j} \binom{k^i_j}{s} (-1)^{k^i_j-s}\Big[ \frac{\partial}{\partial x^i_j} \Big]^s \Big[ \frac{\partial}{\partial x^i_0} \Big]^{k^i_j-s} f^{\tau}_\psi \Bigr|_{\bold{x}_0=1,\bold{x}_{\neq 0}=0}
\end{align}
Hence by \cite[Corollary 4.5]{clift2018derivatives}, \eqref{eq:theorem_eq_1} is equal to the left hand side of \eqref{eq:main_comp1}, noting that the $\ket{-}$ symbols are multi-linear in their entries. On the other hand, we can by \cite[Proposition 4.3]{clift2018derivatives} compute \eqref{eq:theorem_eq_1} by expanding $f^\tau_\psi$:
\begin{align}
&\sum_{\gamma} \delta_{\tau = \operatorname{eval}(\pi, \gamma)} \Big\{ \prod_{i=1}^r \prod_{j=1}^{p_i} \frac{\partial^{k^i_j}}{\partial (w^i_j)^{k^i_j}} \Big\}\Big( x^1_{\gamma_1(1)} \cdots x^1_{\gamma_1(n_1)} \cdots x^r_{\gamma_r(1)} \cdots x^r_{\gamma_r(n_r)} \Big)\Bigr|_{w = \bold{0}} \nonumber\\
&= \sum_{\gamma} \delta_{\tau = \operatorname{eval}(\pi, \gamma)} \Big\{ \prod_{i=1}^r \prod_{j=1}^{p_i} \frac{\partial^{k^i_j}}{\partial (w^i_j)^{k^i_j}} \Big\}\Big( \prod_{i=1}^r (1-\sum_{j=1}^{p_i} w^i_j)^{n_i-|\gamma_i|} (w_1^i)^{|\gamma^{-1}_i(1)|} \cdots (w_{p_i}^i)^{|\gamma^{-1}_i(p_i)|} \Big)\Bigr|_{w = \bold{0}} \nonumber\\
&= \sum_{\gamma} \delta_{\tau = \operatorname{eval}(\pi, \gamma)} \prod_{i=1}^r \Big\{ \prod_{j=1}^{p_i} \frac{\partial^{k^i_j}}{\partial (w^i_j)^{k^i_j}} \Big( (1-\sum_{j=1}^{p_i} w^i_j)^{n_i-|\gamma_i|} (w_1^i)^{|\gamma^{-1}_i(1)|} \cdots (w_{p_i}^i)^{|\gamma^{-1}_i(p_i)|} \Big) \Big\}  \Bigr|_{w = \bold{0}}\,. \label{eq:main_comp1_gammapol}
\end{align}
Next we apply Lemma \ref{comin_deriv_pair} below to each $i$ separately. We find that \eqref{eq:main_comp1_gammapol} is equal to a sum over $\gamma$ of $\delta_{\tau = \operatorname{eval}(\pi, \gamma)}$ times the product from $i = 1$ to $r$ of the delta functions
\[
\delta_{|\bold{k}^i| \le n_i} \delta_{|\gamma_i^{-1}(1)| \le k^i_1} \cdots \delta_{|\gamma_i^{-1}(p_i)| \le k^i_{p_i}}
\]
multiplied by the coefficient
\[
(-1)^{|\bold{k}^i| - |\gamma_i|} \frac{(n_i - |\gamma_i|)!}{(n_i - |\bold{k}^i|)!} \prod_{j=1}^{p_i} \frac{k^i_j!}{(k^i_j-|\gamma^{-1}_i(j)|)!}\,.
\]
Hence the coefficient of $\tau$ in \eqref{eq:theorem_eq_1} is
\begin{align*}
\sum_{\gamma} \delta_{\tau = \operatorname{eval}(\pi, \gamma)} \prod_{i=1}^r \Big\{ \delta_{n_i \ge |\bold{k}^i|} (-1)^{|\bold{k}^i| - |\gamma_i|} \frac{(n_i-|\gamma_i|)!}{(n_i-|\bold{k}^i|)!} \prod_{j=1}^{p_i} \frac{\delta_{k^i_j \ge |\gamma_i^{-1}(j)|}k^i_j!}{(k^i_j - |\gamma^{-1}_i(j)|)!} \Big\}\,.
\end{align*}
The combinatorial factor depends only on the weight of $\gamma$ so we can write \eqref{eq:theorem_eq_1} as
\begin{align*}
\prod_{i=1}^r \frac{\partial^{|\bold{k^i}|}}{\partial (w^i_1)^{k^i_1} \cdots \partial (w^i_{p_i})^{k^i_{p_i}}} ( \Delta \psi \circ \iota )\Bigr|_{w = \bold{0}} &= \sum_{\tau} \sum_{\gamma} \delta_{\tau = \operatorname{eval}(\pi, \gamma)} \delta_{\bold{k} \ge \operatorname{wt}(\gamma)} S^{\bold{k}}_{\operatorname{wt}(\gamma)} \tau\\
&= \sum_\tau \sum_{\bold{s} \le \bold{k}} S^{\bold{k}}_{\bold{s}} A^{\bold{s}}_\tau \tau
\end{align*}
using the notation of \eqref{eq:s_coeff_theorem}, as claimed.
\end{proof}

\begin{lemma}\label{comin_deriv_pair} Given variables $y_1,\ldots,y_p$ and integers $a_1, \ldots, a_p, b, c_1, \ldots, c_p \ge 0$ we have
\begin{align*}
\frac{\partial^{a_1}}{\partial y_1^{a_1}} \cdots \frac{\partial^{a_p}}{\partial y_p^{a_p}} &\Big\{ \big( 1 - \sum_{j=1}^p y_j \big)^b y_1^{c_1} \cdots y_p^{c_p} \Big\}\Bigr|_{\bold{y} = 0}\\
&= \delta_{|\bold{a}| - |\bold{c}| \le b} \delta_{c_1 \le a_1} \cdots \delta_{c_p \le a_p} (-1)^{|\bold{a}| - |\bold{c}|} \frac{b!}{(b - (|\bold{a}| - |\bold{c}|))!} \prod_{j=1}^p \frac{a_j!}{(a_j-c_j)!}
\end{align*}
where $|\bold{a}| = \sum_{j=1}^p a_j$, $|\bold{c}| = \sum_{j=1}^p c_j$.
\end{lemma}
\begin{proof}
The left hand side is clearly zero unless $c_j \le a_j$ for all $j$, in which case we get for each $1 \le j \le p$ a factor of $\binom{a_j}{c_j}$ from all the ways of assigning $c_j$ copies of $\frac{\partial}{\partial y_j}$ to $y_j^{c_j}$ and then another factor of $c_j!$ from the $c_j$ derivatives of $y_j^{c_j}$ yielding the factor $\frac{a_j!}{(a_j-c_j)!}$. Then
\[
\frac{\partial^{a_1-c_1}}{\partial y_1^{a_1-c_1}} \cdots \frac{\partial^{a_p-c_p}}{\partial y_p^{a_p-c_p}}\Big( 1 - \sum_{j=1}^p y_j \Big)^b\Bigr|_{\bold{y} = \bold{0}}
\]
contributes the remaining factors.
\end{proof}

\begin{remark}\label{remark:S_factor_zero}
The sum in \eqref{eq:main_comp1} includes $\bold{s} = \bold{0} = (0,\ldots,0)$. However there is only one error syndrome $\gamma$ of weight $\bold{0}$, the constant function assigning $0$ to every computation path. Moreover by hypothesis if $\gamma = \bold{0}$ denotes this syndrome then
\[
\operatorname{eval}(\pi, \gamma) = \psi(\bold{0}) = 0
\]
that is, $A^{\bold{0}}_\tau = \delta_{\tau=0}$. We note that when $\bold{s} = 0$, the combinatorial factor simplifies
\[
S^{\bold{k}}_{\bold{0}} = \prod_{i=1}^r \Big\{ \delta_{n_i \ge |\bold{k}^i|} (-1)^{|\bold{k}^i|} \frac{n_i!}{(n_i-|\bold{k}^i|)!} \Big\}\,.
\]
\end{remark}

\begin{remark}\label{remark:skk}
Note that provided $|\bold{k}^i| \le n_i$ for $1 \le i \le r$
\[
S^{\bold{k}}_{\bold{k}} = \prod_{i=1}^r \prod_{j=1}^{p_i} k^i_j!\,.
\]
\end{remark}

\section{Derivatives of a Universal Turing Machine}\label{section:diff_utm}

Recall that to define our probabilistic model $p(y|x,w)$ of final states $y$ of our UTM simulating a noisy Turing machine $w$ for $t$ steps on input $x$, we have in \eqref{eq:smooth_relax} of Section \ref{section:background_iid} assumed that we were given a smooth relaxation of the function
\be\label{eq:type_function_stept}
\xymatrix@C+2pc{
\Sigma^* \times W^{\texttt{code}} \ar[r]^-{\textrm{step}^t} & Q
}\,.
\ee
This is constructed as follows: first we encode $\Sigma^*$, $W^{\texttt{code}}$, $Q$ as types and their elements as proofs, then we define a plain proof $\psi$ which computes $\textrm{step}^t$ and then from this we obtain a naive probabilistic extension $\Delta \psi$ which by Proposition \ref{prop:delta_psi} makes the required diagram \eqref{eq:smooth_relax} commute. This $\Delta \psi$ is our smooth relaxation $\Delta \textrm{step}^t$. Our model $p(y|x,w)$ has already been defined given this smooth relaxation (see \eqref{eq:model_from_stept} and Section \ref{section:details_model}) and the derivatives of $H$ are related to derivatives of $\Delta \psi$.

In this section we first define $\mathcal{U}$ (Section \ref{appendix:stagedpseudoutm}), then we define $\psi$ (Section \ref{section:encoding_utm_step}) and then we specialise the discussion earlier of error syndromes for plain proofs to abstractly describe error syndromes for $\psi$ (Section \ref{section:error_syndromes_U}).

\subsection{Staged Pseudo-UTM}
\label{appendix:stagedpseudoutm}

Recall that in a Turing machine with multiple tapes, at each timestep the heads for each tape read the current symbol and, as a function of that sequence of symbols and the current state, the machine transitions to a new state, writes to all of the tapes simultaneously and moves each head \cite{arora2009computational}. We assume our UTM has a description tape (which contains the specification of the TM to simulate), a staging tape (which is used to avoid overloading the states of the UTM with the description tape information), a state tape (which contains the current state of the simulated machine) and a working tape (which is the state of the tape of the simulated machine). 

An example of a pseudo-UTM $\mathcal{U}$ with these features is presented in this section. Note that this is not a universal machine because we assume the symbols and states of the simulated machine can be encoded on single tape squares. However, for the purposes of comparing any finite set of Turing machines this is sufficient (see Appendix \ref{section:faq_realutm}).
\\

We now introduce the staged pseudo-UTM $\mathcal{U}$ following \cite{clift2021geometry}.\footnote{We make the following amendment: we introduce a new symbol $e$, let the symbols $a,c,d,e$ stand for generic symbols which are not $X$, and let $b$ stand for a generic symbol (which may be $X$). In \cite{clift2021geometry}, there is only $a,b,c,d$ and all of these are required to not be $X$. Other than this, our presentation is identical to \cite{clift2021geometry}.} Simulating a Turing machine $M$ with tape alphabet $\Sigma_M$ and set of states $Q_M$ on a UTM generally requires the specification of an encoding of $\Sigma_M$ and $Q_M$ into strings in the tape alphabet $\Sigma_{\text{UTM}}$ of the UTM. From the point of view of this paper it is not clear that this additional complexity is interesting, and it is a significant obstacle to doing calculations by hand. As such given $\Sigma, Q$ we consider a \emph{staged pseudo-UTM} whose alphabet is
\[
\Sigma_{\text{UTM}} = \Sigma \cup Q \cup \{L,R,S\} \cup \{X, \square\}
\]
where the union is disjoint and $\square$ is the blank symbol (which is distinct from the blank symbol of the simulated machines). We emphasise that there is no \emph{theoretical} obstacle to considering a staged UTM of the same design but which makes use of encoded alphabets and states, or even to using a general UTM; we only restrict to the staged pseudo-UTM to make it tractable to reason about the resulting geometry. 

Such a machine is capable of simulating any machine with tape alphabet $\Sigma$ and set of states $Q$ but cannot simulate arbitrary machines and is not a UTM in the standard sense. The adjective \emph{staged} refers to the design of the UTM. The set of states is
\begin{align*}
Q_{\text{UTM}} = \{\,
  & \text{compSymbol, compState, copySymbol, copyState, copyDir,} \\
  & \text{$\neg$compState, $\neg$copySymbol, $\neg$copyState, $\neg$copyDir,}\\
  & \text{updateSymbol, updateState, updateDir, resetDescr} \,\}.
\end{align*}
The UTM has four tapes numbered from 0 to 3, which we refer to as the \textit{description tape}, the \textit{staging tape}, the \textit{state tape} and the \textit{working tape} respectively. We write
\be
n = |\Sigma|\,, \quad m = |Q|\,, \quad N = nm\,.
\ee
Then initially the description tape contains a string of the form
\[
X\sigma_1q_1\sigma_1'q_1'd_1\sigma_2q_2\sigma_2'q_2'd_2 \dots \sigma_Nq_N\sigma_N'q_N'd_NX,
\]
corresponding to the tuples which define $M$, with the tape head initially on $\sigma_1$. The staging tape is initially a string $XXX$ with the tape head over the second $X$. The state tape has a single square containing some distribution in $\Delta Q$, corresponding to the initial state of the simulated machine $M$, with the tape head over that square. Each square on the working tape is some distribution in $\Delta \Sigma$ with only finitely many distributions different from $\Box$. The UTM is initialized in state compSymbol.

The specification of $\mathcal{U}$ is given in Figure~\ref{figure:staged_pseudo_utm}. It consists of two phases; the \textit{scan phase} (middle and right path), and the \textit{update phase} (left path). During the scan phase, the description tape is scanned from left to right, and the first two squares of each tuple are compared to the contents of the working tape and state tape respectively. If both agree, then the last three symbols of the tuple are written to the staging tape (middle path), otherwise the tuple is ignored (right path). Once the $X$ at the end of the description tape is reached, the UTM begins the update phase, wherein the three symbols on the staging tape are then used to print the new symbol on the working tape, to update the simulated state on the state tape, and to move the working tape head in the appropriate direction. The tape head on the description tape is then returned to its initial position over $\sigma_1$, immediately to the right of the initial $X$.

We define the \textit{period} of the UTM to be the smallest nonzero time interval taken for the tape head on the description tape to return to its initial position over $\sigma_1$, and the machine to reenter the state compSymbol. If the number of tuples on the description tape is $N$, then the period of the UTM is $P = 10N+5$. Moreover, other than the working tape, the position of the tape heads are $P$-periodic.

\begin{figure}[htpb]
    \centering
\[
\begin{tikzcd}[row sep = large]
	&& {\text{compSymbol}} \\
	\\
	{\text{updateSymbol}} && {\text{compState}} && {\neg \text{compState}} \\
	\\
	{\text{updateState}} && {\text{copySymbol}} && {\neg\text{copySymbol}} \\
	\\
	{\text{updateDir}} && {\text{copyState}} && {\neg\text{copyState}} \\
	\\
	{\text{resetDescr}} && {\text{copyDir}} && {\neg\text{copyDir}} \\
	\\
	&& {\textsl{go to compSymbol}}
	\arrow["{X,b,c,d}"'{pos=0.3}, from=1-3, to=3-1]
	\arrow["LLSS"'{pos=0.7}, from=1-3, to=3-1]
	\arrow["{a,b,c,a}"{pos=0.3}, from=1-3, to=3-3]
	\arrow["RLSS"{pos=0.7}, from=1-3, to=3-3]
	\arrow["{\text{else}}"{pos=0.3}, from=1-3, to=3-5]
	\arrow["RLSS"{pos=0.7}, from=1-3, to=3-5]
	\arrow["{a,X,c,d}"'{pos=0.3}, shift right=5, from=3-1, to=5-1]
	\arrow["SRSS"'{pos=0.7}, shift right=5, from=3-1, to=5-1]
	\arrow["{a,e,c,d}"{pos=0.3}, shift left=5, from=3-1, to=5-1]
	\arrow["{\genfrac{}{}{0pt}{}{\text{write }a,X,c,e}{SRSS}}"{pos=0.7}, shift left=5, from=3-1, to=5-1]
	\arrow["{a,b,a,d}"{pos=0.3}, from=3-3, to=5-3]
	\arrow["RSSS"{pos=0.7}, from=3-3, to=5-3]
	\arrow["{\text{else}}"{pos=0.3}, from=3-3, to=5-5]
	\arrow["RSSS"{pos=0.7}, from=3-3, to=5-5]
	\arrow["{a,b,c,d}"{pos=0.3}, from=3-5, to=5-5]
	\arrow["RSSS"{pos=0.7}, from=3-5, to=5-5]
	\arrow["{a,X,c,d}"'{pos=0.3}, shift right=5, from=5-1, to=7-1]
	\arrow["SRSS"'{pos=0.7}, shift right=5, from=5-1, to=7-1]
	\arrow["{a,e,c,d}"{pos=0.3}, shift left=5, from=5-1, to=7-1]
	\arrow["{\genfrac{}{}{0pt}{}{\text{write }a,X,e,d}{SRSS}}"{pos=0.7}, shift left=5, from=5-1, to=7-1]
	\arrow["{a,b,c,d}"{pos=0.3}, from=5-3, to=7-3]
	\arrow["{\genfrac{}{}{0pt}{}{\text{write }a,a,c,d}{RRSS}}"{pos=0.7}, from=5-3, to=7-3]
	\arrow["{a,b,c,d}"{pos=0.3}, from=5-5, to=7-5]
	\arrow["RRSS"{pos=0.7}, from=5-5, to=7-5]
	\arrow["{a,X,c,d}"'{pos=0.3}, shift right=5, from=7-1, to=9-1]
	\arrow["SLSS"'{pos=0.7}, shift right=5, from=7-1, to=9-1]
	\arrow["{a,e,c,d}"{pos=0.3}, shift left=5, from=7-1, to=9-1]
	\arrow["{\genfrac{}{}{0pt}{}{\text{write }a,X,c,d}{SLSe}}"{pos=0.7}, shift left=5, from=7-1, to=9-1]
	\arrow["{a,b,c,d}"{pos=0.3}, from=7-3, to=9-3]
	\arrow["{\genfrac{}{}{0pt}{}{\text{write }a,a,c,d}{RRSS}}"{pos=0.7}, from=7-3, to=9-3]
	\arrow["{a,b,c,d}"{pos=0.3}, from=7-5, to=9-5]
	\arrow["RRSS"{pos=0.7}, from=7-5, to=9-5]
	\arrow["{a,b,c,d}"'{pos=0.3}, from=9-1, to=9-1, loop, in=265, out=185, distance=15mm]
	\arrow["LSSS"{pos=0.3}, no head, from=9-1, to=9-1, loop, in=185, out=265, distance=15mm]
	\arrow["{X,b,c,d}"'{pos=0.3}, from=9-1, to=11-3]
	\arrow["RSSS"'{pos=0.7}, from=9-1, to=11-3]
	\arrow["{a,b,c,d}"{pos=0.3}, from=9-3, to=11-3]
	\arrow["{\genfrac{}{}{0pt}{}{\text{write }a,a,c,d}{RLSS}}"{pos=0.7}, from=9-3, to=11-3]
	\arrow["{a,b,c,d}"{pos=0.3}, from=9-5, to=11-3]
	\arrow["RLSS"{pos=0.7}, from=9-5, to=11-3]
\end{tikzcd}
\]
\caption{\label{figure:staged_pseudo_utm}\textbf{The staged pseudo-UTM $\mathcal{U}$.} The nodes, except for \emph{go to compSymbol}, are states of $\mathcal{U}$ and an arrow $q \to q'$ has the following interpretation: if the UTM is in state $q$ and sees the tape symbols (on the four tapes) as indicated by the source of the arrow, then the UTM transitions to state $q'$, writes the indicated symbols (or if there is no write instruction, simply rewrites the same symbols back onto the tapes), and performs the indicated movements of each of the tape heads with $R, S, L$ standing for Right, Stay, Left. The symbols $a,c,d,e$ stand for generic symbols which are not $X$, and $b$ stands for a generic symbol (which may be $X$).}
\end{figure}
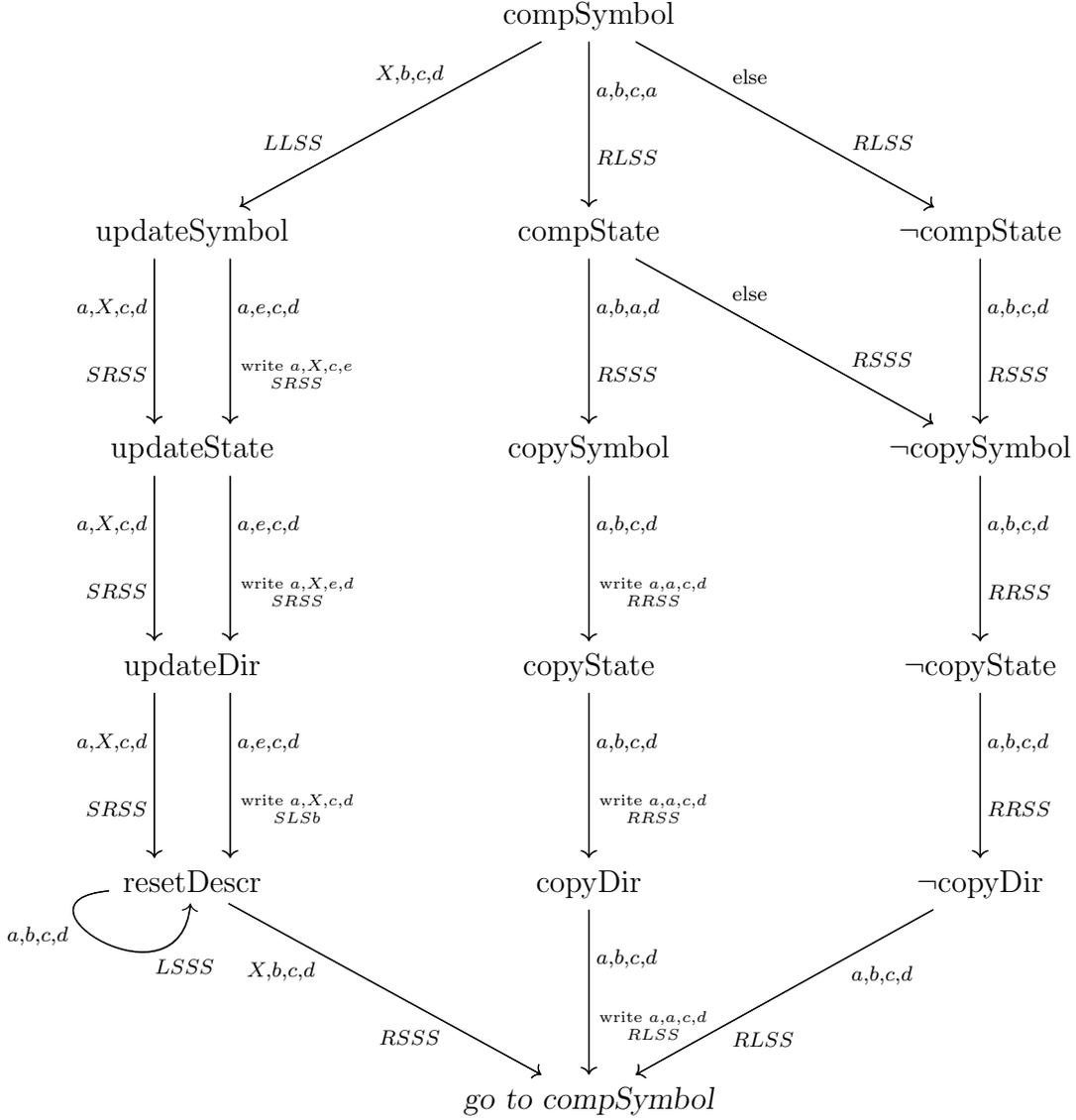

\subsection{Encoding of the Step Function}\label{section:encoding_utm_step}

We encode the function \eqref{eq:type_function_stept} as a plain proof which can be described schematically as
\be\label{eq:schematic_Psi}
\psi: \big\{ \text{type encoding } \Sigma^* \big\}, \big\{ \text{type encoding } W^{\texttt{code}} \big\} \vdash \big\{ \text{type encoding } Q \big\}.
\ee
This is an encoding in the sense that cutting $\psi$ against a proof encoding a string $x \in \Sigma^*$ on the work tape and a proof encoding $w \in W^{\texttt{code}}$ yields a proof equivalent under cut-elimination to $\textrm{step}^t(x,w)$, the state after $\mathcal{U}$ has simulated $w$ on input $x$ for $t$ steps.

In Appendix \ref{section:multi_tape} we define a component-wise plain proof ${}_h\underline{\text{$\mathcal{U}$relstep}}^t$ which encodes the UTM $\mathcal{U}$ running for $t$ cycles, that is, simulating $t$ steps of a Turing machine (see Definition \ref{defn:urelstep_t}). The component of this component-wise plain proof that produces the state of the simulated machine after $t$ steps has type
\be\label{defn:utm_state_comp_type}
\textbf{dscr} \otimes \textbf{stg} \otimes {!}{}_m\textbf{bool} \otimes {!}{}_n\textbf{bool}^{\otimes 2h+1}\otimes {!}{}_{13}\textbf{bool} \vdash {}_m\textbf{bool}\,.
\ee
The notation is defined in full in the appendix, but in summary:
\begin{itemize}
    \item Proofs of type $\textbf{dscr}$ encode configurations of the description tape, where
    \begin{equation}
    \textbf{dscr} = \bigotimes_{i = 1}^{N}\Big( \underset{\sigma'_i}{\underbrace{{!}_n\textbf{bool}}} \otimes \underset{q'_i}{\underbrace{{!}_m\textbf{bool}}} \otimes \underset{d_i}{\underbrace{{!}_3\textbf{bool}}} \Big)\,.
\end{equation}
    \item Proofs of $\textbf{stg}$ encode configurations of the staging tape.
    \item Proofs of ${!}{}_m\textbf{bool}$ encode initial states of the simulated machine (i.e. elements of $Q$ which has size $m$).
    \item Proofs of ${!}{}_n\textbf{bool}^{\otimes 2h + 1}$ encode configurations of the work tape (that is, $2h + 1$ elements of $\Sigma$ where we note that the work tape is encoded relative to the head, so this means there are $h$ entries to the right of the head and $h$ entries to the left).
    \item Proofs of ${!}{}_{13}\textbf{bool}$ encode the state of the UTM itself (i.e. an element of $Q_{\text{UTM}}$).
\end{itemize}

\begin{definition}
We define $\psi$ to be the plain proof obtained by cutting the component of ${}_h\underline{\text{$\mathcal{U}$relstep}}^t$ with type \eqref{defn:utm_state_comp_type} against the following inputs
\[
\underline{XXX}: \textbf{stg}\,, \quad \underline{\text{init}}: {!}{}_m\textbf{bool}\,, \quad \underline{\text{compSymbol}} : {!}{}_{13}\textbf{bool}
\]
which respectively encode putting $XXX$ on the staging tape, some initialisation state on the state tape and starting the UTM in the initial state $\text{compSymbol}$. Here the underlines refer to proofs encoding a given piece of data. This proof has type
\[
\psi: \underset{w}{\underbrace{\textbf{dscr}}} \otimes \underset{x}{\underbrace{{!}{}_n\textbf{bool}^{\otimes 2h+1}}} \vdash {}_m\textbf{bool}
\]
realising the schematic in \eqref{eq:schematic_Psi}.\footnote{On the left hand side of the turnstile tensors may be exchanged for commas, and freely permuted.}
\end{definition}

We work with fixed sets of denotations of proofs (see \cite[Appendix B]{clift2020encodings} for definitions):
\begin{equation}
    \mathcal{P}_i^{\text{dscr}} =
    \begin{cases}
        \{\llbracket \underline{j} \rrbracket\}_{j = 0,\ldots, n-1}\subseteq \llbracket {}_{n}\textbf{bool} \rrbracket, & i = 1\text{ mod 3},\\
        \{\llbracket \underline{j} \rrbracket \}_{j = 0,\ldots, m-1}\subseteq \llbracket {}_{m}\textbf{bool} \rrbracket, & i = 2\text{ mod 3},\\
        \{\llbracket \underline{j} \rrbracket\}_{j = 0,1,2}\subseteq \llbracket {}_3\textbf{bool} \rrbracket, & i = 0\text{ mod 3}.
    \end{cases}
\end{equation}

\begin{remark}\label{remark:enum_proofs}
These proofs are used to encode elements of $\Sigma, Q$ and directions to move. There is a subtlety: while all the squares of the sample ``type'' ($\sigma',q'$ or $d$) are encoded by the same formula and have the set of proofs, later it will be important that \emph{different squares on the description tape use different enumerations} in the sense of \eqref{eq:enumeration_Pi}. The enumerations are arbitrary, except that we always choose $0$ to label the symbol, state or direction which appears in this position in the code $[M]$ of the chosen Turing machine $M$.
\end{remark}

Since $\psi$ is plain \cite[Lemma 3.4]{clift2020encodings} we may pass to probability distributions (Proposition \ref{prop:delta_psi}) to obtain the function
\be
\Delta \psi: \Sigma^* \times \prod_{i=1}^N \Big( \Delta \Sigma \times \Delta Q \times \Delta \{ \text{L}, \text{S}, \text{R} \}\Big) \lto \Delta Q \,.
\ee
Since we do not allow uncertainty in the input $x$ we restrict the function to $\Sigma^* \subseteq \Delta \Sigma^*$.

\begin{definition} The smooth relaxation $\Delta \textrm{step}^t$ in \eqref{eq:smooth_relax} is $\Delta \psi$.
\end{definition}

\subsection{Error Syndromes for $\mathcal{U}$}\label{section:error_syndromes_U}

Error syndromes for $\mathcal{U}$ are determined by the linear part $\pi$ of the proof $\psi$ encoding $\textrm{step}^t$ defined in the previous section. Abstractly we know $\psi$ is a plain proof and thus there are uniquely determined integers $c,d$ (with various sub- and superscripts) which are the degrees of each input with
\be
\pi: \bigotimes_{i = 1}^{N}\Big( c_i^{\text{sym}}\, {}_n\textbf{bool} \otimes c_i^{\text{state}}\, {}_m\textbf{bool} \otimes c_i^{\text{dir}} \, {}_3\textbf{bool} \Big) \otimes \bigotimes_{i=-h}^h b_i \, {}_n\textbf{bool} \vdash {}_m\textbf{bool}\,.
\ee
However it is not necessarily trivial to compute these integers or enumerate the computation paths that they count. Later in Section \ref{section:relstep_UTM} we do this explicitly when $t = 2$. In any case, we can now formally define \emph{error syndromes} for $\mathcal{U}$ to be error syndromes in the sense of Definition \ref{defn:error_syndrome} associated to the pair consisting of the plain proof $\psi$ of the previous section with linear part $\pi$. 

\begin{remark}
Because of the way that the UTM is designed, the number of computation paths starting at a square on the description tape depends on the \emph{type} of the square but not its position on the tape. That is, there exist integers $c^{\text{sym}}, c^{\text{state}}, c^{\text{dir}}$ such that $c_i^{\text{sym}} = c^{\text{sym}}$, $c_i^{\text{state}} = c^{\text{state}}$ and $c_i^{\text{dir}} = c^{\text{dir}}$ for all $1 \le i \le N$.
\end{remark}

It is important to keep in mind that since we study geometry \emph{locally} at a code $[M]$ error syndromes are always considered relative to a particular Turing machine. To make this explicit, let us now:
\begin{center}
    \textbf{Fix a Turing machine $M$} so that ``errors'' are variations away from $[M]$.
\end{center}
That is, error syndromes describe variations away from this code as a string on the description tape of the UTM (see Remark \ref{remark:enum_proofs}). We also do not consider derivatives with respect to the tape squares encoding the input $x$ and thus do not consider them in the definition of error syndromes. So to be explicit:

\begin{definition}\label{defn:error_syndrome_utm} An \emph{error syndrome} $\gamma$ for $\mathcal{U}$ is a sequence of functions $\gamma_1,\ldots,\gamma_{3N}$ where
\[
\gamma_i: \{ 1, \ldots, n_i \} \lto \{ 0, \ldots, p_i \} \qquad 1 \le i \le 3N\,.
\]
and
\[
n_i = \begin{cases} c^{\text{sym}} & i = 1 \text{ mod } 3\\
c^{\text{state}} & i = 2 \text{ mod } 3\\
c^{\text{dir}} & i = 3 \text{ mod } 3
\end{cases}\,,
\qquad
p_i = \begin{cases} n - 1 & i = 1 \text{ mod } 3\\
m - 1 & i = 2 \text{ mod } 3\\
2 & i = 3 \text{ mod } 3
\end{cases}
\]
where $\gamma_i(j) = 0$ is always interpreted as \emph{no error} that is, whatever symbol, state or direction is in this position of the description of some fixed background machine $[M]$.
\end{definition}

The tuples $\bold{s}, \bold{k}$ of Section \ref{section:error_syndrome_general_case} belong to
\be
\prod_{i=1}^r \mathbb{N}^{p_i} = \prod_{a=1}^N\big[ \mathbb{N}^{|\Sigma|-1} \times \mathbb{N}^{|Q|-1} \times \mathbb{N}^2 \big]\,.
\ee
To make this more concrete, suppose $\bold{s} \in \prod_{i=1}^r \mathbb{N}^{p_i}$ is such a tuple and that $1 \le a \le N$ is the index of a tuple $\sigma q \sigma' q' d$ in $[M]$ and that
\be
\bold{s} = \Big( \cdots, \overset{a}{\overbrace{\underset{\mathbb{N}^{|\Sigma|-1}}{\bold{u}}, \underset{\mathbb{N}^{|Q|-1}}{\bold{v}}, \underset{\mathbb{N}^2}{\bold{w}}\;}}, \ldots \Big)\,.
\ee
If $\operatorname{wt}(\gamma) = \bold{s}$ for some error syndrome $\gamma$, then out of all the computation paths leading from the description square containing $\sigma'$ in $[M]$ the count of those assigned errors of each type given by elements of $\Sigma \setminus \{ \sigma' \}$ is given by $\bold{u}$, and similarly $\bold{v}, \bold{w}$ count assignments of errors of various types to computation paths out of $q', d$ respectively.

\begin{definition}\label{defn:uxgamma}
Given an input $x$ and error syndrome $\gamma$ we define
\begin{equation}\label{eq:mathcalU_x_gamma}
\mathcal{U}(x,\gamma) = \pi\big( \gamma, x \big)
\end{equation}
where as in \eqref{eq:eval_pi_gamma} the right hand side involves placing the outputs of all the $\gamma_i$ as inputs to $\pi$ in the correct positions. In particular if $\gamma = \bold{0}$ is the constant function then
\[
\mathcal{U}(x, \bold{0}) = M(x)\,.
\]
\end{definition}

\begin{remark}\label{remark:s_in_Nd}
Note that
\[
\sum_{i=1}^r p_i = |Q||\Sigma|\big( |\Sigma| + |Q| ) = d
\]
the dimension of the parameter space $W$. Hence we may also write $\bold{s}, \bold{k} \in \mathbb{N}^d$.
\end{remark}

With this preparation everything in Section \ref{section:error_syndromes} applies verbatim to the analysis of derivatives of the plain proof $\psi$. We denote by $A^{\bold{s}}(x)$ the error syndrome count from \eqref{eq:defn_A_s_tau} with $\tau$ the alternative output to the correct output $M(x)$, that is:

\begin{definition}\label{defn:Asfunctions} Given $\bold{s} \in \prod_{i=1}^r \mathbb{N}^{p_i}$ and $x \in \Sigma^*$
\be
A^{\bold{s}}(x) = \Big| \big\{ \gamma \l \operatorname{wt}(\gamma) = \bold{s} \text{ and } \mathcal{U}(x,\gamma) \neq M(x) \big\} \Big|
\ee
denotes the count of weight $\bold{s}$ error syndromes $\gamma$ which, when we evaluate the UTM with $[M]$ on the description tape, deviations from the code given by $\gamma$, and $x$ on the work tape, produce the incorrect output.
\end{definition}

\begin{remark}\label{remark:A_s_zero}
Note that by definition, $A^{\bold{0}}(x) = 0$ for all $x$.
\end{remark}

Since $\psi$ is defined by cutting $t$ component-wise plain proofs against one another, we can from the proof that the cut of component-wise plain proofs is component-wise plain \cite[Proposition 3.6]{clift2020encodings} derive the linear part $\pi$ of $\psi$ from the linear parts of each of these individual step functions. This gives us an explicit form of the computation paths through $\psi$ and thus the error syndromes. In the next section we present the computation paths in a special case.

\section{Error and Geometry}
\label{sec:GeometryofInductiveInference}

In Section \ref{section:background} we outlined a perspective on inductive inference over Turing machines, where we believe that a Turing machine is an explanation for our observations to the extent that the posterior concentrates around it in a continuous parameter space of ``noisy'' Turing machines. We sketched how our preferences for Turing machines are dominated, asymptotically in the number of samples $n$ from the generating process, by the geometry of the average log-likelihood $L$ at the code $[M]$ of our machine. 

This geometry was related in Section \ref{section:geometry} to the derivatives of the function that propagates uncertainty about Turing machine codes through our UTM to uncertainty about final states of those simulated machines. Section \ref{section:derivatives_errors} showed that these derivatives can be understood in terms of the combinatorics of error syndromes (Theorem \ref{theorem:error_count}). If we put this all together, we can make some progress into understanding the underlying geometry of inductive inference.
\\

As in Section \ref{section:error_syndromes_U}, we now fix in addition to the staged pseudo-UTM $\mathcal{U}$ a particular classical solution $M$. That is, a Turing machine with $M(x) = y(x)$ for all $x \in I$.

\subsection{Graphical Model}\label{section:graphical_model_utm}

We consider the evolution of the random variables involved in the computation of $\mathcal{U}$ as encoded in linear logic in the previous section. This was previously discussed in \cite[\S 6.2]{clift2018derivatives}.

To each pair consisting of a tape square of $\mathcal{U}$ and a time step $\mu$ within the cycle there is a corresponding random variable. There is also a random variable for each time step $t$ ranging over all possible states of $\mathcal{U}$. The semantics of the encoding of the step function of $\mathcal{U}$ give the random variables at timestep $\mu + 1$ as functions of those at timestep $\mu$. This defines a directed graphical model (DGM) for one cycle of the UTM (so one simulated step) as shown in Figure \ref{figure:symbol_write_error_diagram}. This has been simplified by the method in Section \ref{section:comp_path_utm}.

We indicate in the left-most column the entry $[M]_i$ in the code of $M$ which is currently under the UTM head. One particular square, of $\sigma'$ type, is highlighted and denoted $\theta$. Paths in the DGM that start from $\theta$ (or other random variables on the description tape) and end at the state $q$ in some future cycle are the computation paths of $\mathcal{U}$ which are enumerated carefully in the next section. Further details are given in Appendix \ref{section:comp_path_utm}.

\begin{figure}[htpb]
    \centering
\[\adjustbox{scale=0.65}{\begin{tikzcd}[column sep = tiny, row sep = small]
	& X & X & X & q & {\sigma_{-h-1}} & {\sigma_{-h}} & \ldots & {\sigma_{-1}} & {\sigma_{0}} & {\sigma_{1}} & \ldots & {\sigma_{h}} & {\sigma_{h+1}} & {\text{compSymbol}} & {10N+5} \\
	& X & X & X & q & {\sigma_{-h-1}} & {\sigma_{-h}} & \ldots & {\sigma_{-1}} & {\sigma_{0}} & {\sigma_{1}} & \ldots & {\sigma_{h}} & {\sigma_{h+1}} & {\text{resetDescr}} & {10N+4} \\
	& \vdots & \vdots & \vdots & \vdots & \vdots & \vdots & \vdots & \vdots & \vdots & \vdots & \vdots & \vdots & \vdots & \vdots & \vdots \\
	& X & X & X & q & {\sigma_{-h-1}} & {\sigma_{-h}} & \ldots & {\sigma_{-1}} & {\sigma_{0}} & {\sigma_{1}} & \ldots & {\sigma_{h}} & {\sigma_{h+1}} & {\text{resetDescr}} & {5N+4} \\
	& X & X & {s_{2}} & q && {\sigma_{-h}} & \ldots & {\sigma_{-1}} & {\sigma_{0}} & {\sigma_{1}} & \ldots & {\sigma_{h}} && {\text{updateDir}} & {5N+3} \\
	& X & {s_{1}} & {s_{2}} & q && {\sigma_{-h}} & \ldots & {\sigma_{-1}} & {\sigma_{0}} & {\sigma_{1}} & \ldots & {\sigma_{h}} && {\text{updateState}} & {5N+2} \\
	& {s_{0}} & {s_{1}} & {s_{2}} & q && {\sigma_{-h}} & \ldots & {\sigma_{-1}} & {\sigma_{0}} & {\sigma_{1}} & \ldots & {\sigma_{h}} && {\text{updateSymbol}} & {5N+1} \\
	& {s_{0}} & {s_{1}} & {s_{2}} & q && {\sigma_{-h}} & \ldots & {\sigma_{-1}} & {\sigma_{0}} & {\sigma_{1}} & \ldots & {\sigma_{h}} && {\text{compSymbol}} & 5N \\
	& \vdots & \vdots & \vdots & \vdots && \vdots & \vdots & \vdots & \vdots & \vdots & \vdots & \vdots && \vdots & \vdots \\
	& {s_{0}} & {s_{1}} & {s_{2}} & q && {\sigma_{-h}} & \ldots & {\sigma_{-1}} & {\sigma_{0}} & {\sigma_{1}} & \ldots & {\sigma_{h}} && {\text{compSymbol}} & {i+2} \\
	{[M]_{i+2}} & {s_{0}} & {s_{1}} & {s_{2}} & q && {\sigma_{-h}} & \ldots & {\sigma_{-1}} & {\sigma_{0}} & {\sigma_{1}} & \ldots & {\sigma_{h}} && \varphi & {i+1} \\
	{[M]_{i+1}} & {s_{0}} & {s_{1}} & {s_{2}} & q && {\sigma_{-h}} & \ldots & {\sigma_{-1}} & {\sigma_{0}} & {\sigma_{1}} & \ldots & {\sigma_{h}} && \varphi & i \\
	\theta & {s_0} & {s_{1}} & {s_{2}} & q && {\sigma_{-h}} & \ldots & {\sigma_{-1}} & {\sigma_{0}} & {\sigma_{1}} & \ldots & {\sigma_{h}} && \varphi & {i-1} \\
	& {s_0} & {s_{1}} & {s_{2}} & q && {\sigma_{-h}} & \ldots & {\sigma_{-1}} & {\sigma_{0}} & {\sigma_{1}} & \ldots & {\sigma_{h}} && \varphi & {i-2} \\
	& {s_0} & {s_{1}} & {s_{2}} & q && {\sigma_{-h}} & \ldots & {\sigma_{-1}} & {\sigma_{0}} & {\sigma_{1}} & \ldots & {\sigma_{h}} && {\text{compSymbol}} & {i-3} \\
	& \vdots & \vdots & \vdots & \vdots && \vdots & \vdots & \vdots & \vdots & \vdots & \vdots & \vdots && \vdots & \vdots \\
	& {s_0} & {s_{1}} & {s_{2}} & q && {\sigma_{-h}} & \ldots & {\sigma_{-1}} & {\sigma_0} & {\sigma_{1}} & \ldots & {\sigma_h} && {\text{compSymbol}} & 5 \\
	{[M]_5} & {s_0} & {s_{1}} & X & q && {\sigma_{-h}} & \ldots & {\sigma_{-1}} & {\sigma_0} & {\sigma_{1}} & \ldots & {\sigma_h} && \varphi & 4 \\
	{[M]_4} & {s_0} & X & X & q && {\sigma_{-h}} & \ldots & {\sigma_{-1}} & {\sigma_0} & {\sigma_{1}} & \ldots & {\sigma_h} && \varphi & 3 \\
	{[M]_{3}} & X & X & X & q && {\sigma_{-h}} & \ldots & {\sigma_{-1}} & {\sigma_{0}} & {\sigma_{1}} & \ldots & {\sigma_{h}} && \varphi & 2 \\
	& X & X & X & q && {\sigma_{-h}} & \ldots & {\sigma_{-1}} & {\sigma_{0}} & {\sigma_{1}} & \ldots & {\sigma_{h}} && \varphi & 1 \\
	& X & X & X & q && {\sigma_{-h}} & \ldots & {\sigma_{-1}} & {\sigma_0} & {\sigma_{1}} & \ldots & {\sigma_h} && {\text{compSymbol}} & 0
	\arrow[from=2-2, to=1-2]
	\arrow[from=2-3, to=1-3]
	\arrow[from=2-4, to=1-4]
	\arrow[from=2-5, to=1-5]
	\arrow[from=2-6, to=1-6]
	\arrow[from=2-7, to=1-7]
	\arrow[from=2-9, to=1-9]
	\arrow[from=2-10, to=1-10]
	\arrow[from=2-11, to=1-11]
	\arrow[from=2-13, to=1-13]
	\arrow[from=2-14, to=1-14]
	\arrow[from=3-2, to=2-2]
	\arrow[from=3-3, to=2-3]
	\arrow[from=3-4, to=2-4]
	\arrow[from=3-5, to=2-5]
	\arrow[from=3-6, to=2-6]
	\arrow[from=3-7, to=2-7]
	\arrow[from=3-9, to=2-9]
	\arrow[from=3-10, to=2-10]
	\arrow[from=3-11, to=2-11]
	\arrow[from=3-13, to=2-13]
	\arrow[from=3-14, to=2-14]
	\arrow[from=4-2, to=3-2]
	\arrow[from=4-3, to=3-3]
	\arrow[from=4-4, to=3-4]
	\arrow[from=4-5, to=3-5]
	\arrow[from=4-6, to=3-6]
	\arrow[from=4-7, to=3-7]
	\arrow[from=4-9, to=3-9]
	\arrow[from=4-10, to=3-10]
	\arrow[from=4-11, to=3-11]
	\arrow[from=4-13, to=3-13]
	\arrow[from=4-14, to=3-14]
	\arrow[from=5-2, to=4-2]
	\arrow[from=5-3, to=4-3]
	\arrow[from=5-4, to=4-6]
	\arrow[from=5-4, to=4-7]
	\arrow[from=5-4, to=4-9]
	\arrow[from=5-4, to=4-10]
	\arrow[from=5-4, to=4-11]
	\arrow[from=5-4, to=4-13]
	\arrow[from=5-4, to=4-14]
	\arrow[from=5-5, to=4-5]
	\arrow[from=5-7, to=4-6]
	\arrow[from=5-7, to=4-7]
	\arrow[from=5-9, to=4-9]
	\arrow[from=5-9, to=4-10]
	\arrow[from=5-10, to=4-9]
	\arrow[from=5-10, to=4-10]
	\arrow[from=5-10, to=4-11]
	\arrow[from=5-11, to=4-10]
	\arrow[from=5-11, to=4-11]
	\arrow[from=5-13, to=4-13]
	\arrow[from=5-13, to=4-14]
	\arrow[from=6-2, to=5-2]
	\arrow[from=6-3, to=5-5]
	\arrow[from=6-4, to=5-4]
	\arrow[from=6-5, to=5-5]
	\arrow[from=6-7, to=5-7]
	\arrow[from=6-9, to=5-9]
	\arrow[from=6-10, to=5-10]
	\arrow[from=6-11, to=5-11]
	\arrow[from=6-13, to=5-13]
	\arrow[from=7-2, to=6-10]
	\arrow[from=7-3, to=6-3]
	\arrow[from=7-4, to=6-4]
	\arrow[from=7-5, to=6-5]
	\arrow[from=7-7, to=6-7]
	\arrow[from=7-9, to=6-9]
	\arrow[from=7-10, to=6-10]
	\arrow[from=7-11, to=6-11]
	\arrow[from=7-13, to=6-13]
	\arrow[from=8-2, to=7-2]
	\arrow[from=8-3, to=7-3]
	\arrow[from=8-4, to=7-4]
	\arrow[from=8-5, to=7-5]
	\arrow[from=8-7, to=7-7]
	\arrow[from=8-9, to=7-9]
	\arrow[from=8-10, to=7-10]
	\arrow[from=8-11, to=7-11]
	\arrow[from=8-13, to=7-13]
	\arrow[from=9-2, to=8-2]
	\arrow[from=9-3, to=8-3]
	\arrow[from=9-4, to=8-4]
	\arrow[from=9-5, to=8-5]
	\arrow[from=9-7, to=8-7]
	\arrow[from=9-9, to=8-9]
	\arrow[from=9-10, to=8-10]
	\arrow[from=9-11, to=8-11]
	\arrow[from=9-13, to=8-13]
	\arrow[from=10-2, to=9-2]
	\arrow[from=10-3, to=9-3]
	\arrow[from=10-4, to=9-4]
	\arrow[from=10-5, to=9-5]
	\arrow[from=10-7, to=9-7]
	\arrow[from=10-9, to=9-9]
	\arrow[from=10-10, to=9-10]
	\arrow[from=10-11, to=9-11]
	\arrow[from=10-13, to=9-13]
	\arrow[from=11-1, to=10-4]
	\arrow[from=11-2, to=10-2]
	\arrow[from=11-3, to=10-3]
	\arrow[from=11-4, to=10-4]
	\arrow[from=11-5, to=10-5]
	\arrow[from=11-7, to=10-7]
	\arrow[from=11-9, to=10-9]
	\arrow[from=11-10, to=10-10]
	\arrow[from=11-11, to=10-11]
	\arrow[from=11-13, to=10-13]
	\arrow[from=11-15, to=10-4]
	\arrow[from=12-1, to=11-3]
	\arrow[from=12-2, to=11-2]
	\arrow[from=12-3, to=11-3]
	\arrow[from=12-4, to=11-4]
	\arrow[from=12-5, to=11-5]
	\arrow[from=12-7, to=11-7]
	\arrow[from=12-9, to=11-9]
	\arrow[from=12-10, to=11-10]
	\arrow[from=12-11, to=11-11]
	\arrow[from=12-13, to=11-13]
	\arrow[from=12-15, to=11-3]
	\arrow[from=12-15, to=11-15]
	\arrow[from=13-1, to=12-2]
	\arrow[from=13-2, to=12-2]
	\arrow[from=13-3, to=12-3]
	\arrow[from=13-4, to=12-4]
	\arrow[from=13-5, to=12-5]
	\arrow[from=13-7, to=12-7]
	\arrow[from=13-9, to=12-9]
	\arrow[from=13-10, to=12-10]
	\arrow[from=13-11, to=12-11]
	\arrow[from=13-13, to=12-13]
	\arrow[from=13-15, to=12-2]
	\arrow[from=13-15, to=12-15]
	\arrow[from=14-2, to=13-2]
	\arrow[from=14-3, to=13-3]
	\arrow[from=14-4, to=13-4]
	\arrow[from=14-5, to=13-5]
	\arrow[from=14-5, to=13-15]
	\arrow[from=14-7, to=13-7]
	\arrow[from=14-9, to=13-9]
	\arrow[from=14-10, to=13-10]
	\arrow[from=14-11, to=13-11]
	\arrow[from=14-13, to=13-13]
	\arrow[from=14-15, to=13-15]
	\arrow[from=15-2, to=14-2]
	\arrow[from=15-3, to=14-3]
	\arrow[from=15-4, to=14-4]
	\arrow[from=15-5, to=14-5]
	\arrow[from=15-7, to=14-7]
	\arrow[from=15-9, to=14-9]
	\arrow[from=15-10, to=14-10]
	\arrow[from=15-10, to=14-15]
	\arrow[from=15-11, to=14-11]
	\arrow[from=15-13, to=14-13]
	\arrow[from=16-2, to=15-2]
	\arrow[from=16-3, to=15-3]
	\arrow[from=16-4, to=15-4]
	\arrow[from=16-5, to=15-5]
	\arrow[from=16-7, to=15-7]
	\arrow[from=16-9, to=15-9]
	\arrow[from=16-10, to=15-10]
	\arrow[from=16-11, to=15-11]
	\arrow[from=16-13, to=15-13]
	\arrow[from=17-2, to=16-2]
	\arrow[from=17-3, to=16-3]
	\arrow[from=17-4, to=16-4]
	\arrow[from=17-5, to=16-5]
	\arrow[from=17-7, to=16-7]
	\arrow[from=17-9, to=16-9]
	\arrow[from=17-10, to=16-10]
	\arrow[from=17-11, to=16-11]
	\arrow[from=17-13, to=16-13]
	\arrow[from=18-1, to=17-4]
	\arrow[from=18-2, to=17-2]
	\arrow[from=18-3, to=17-3]
	\arrow[from=18-4, to=17-4]
	\arrow[from=18-5, to=17-5]
	\arrow[from=18-7, to=17-7]
	\arrow[from=18-9, to=17-9]
	\arrow[from=18-10, to=17-10]
	\arrow[from=18-11, to=17-11]
	\arrow[from=18-13, to=17-13]
	\arrow[from=18-15, to=17-4]
	\arrow[from=19-1, to=18-3]
	\arrow[from=19-2, to=18-2]
	\arrow[from=19-3, to=18-3]
	\arrow[from=19-4, to=18-4]
	\arrow[from=19-5, to=18-5]
	\arrow[from=19-7, to=18-7]
	\arrow[from=19-9, to=18-9]
	\arrow[from=19-10, to=18-10]
	\arrow[from=19-11, to=18-11]
	\arrow[from=19-13, to=18-13]
	\arrow[from=19-15, to=18-3]
	\arrow[from=19-15, to=18-15]
	\arrow[from=20-1, to=19-2]
	\arrow[from=20-2, to=19-2]
	\arrow[from=20-3, to=19-3]
	\arrow[from=20-4, to=19-4]
	\arrow[from=20-5, to=19-5]
	\arrow[from=20-7, to=19-7]
	\arrow[from=20-9, to=19-9]
	\arrow[from=20-10, to=19-10]
	\arrow[from=20-11, to=19-11]
	\arrow[from=20-13, to=19-13]
	\arrow[from=20-15, to=19-2]
	\arrow[from=20-15, to=19-15]
	\arrow[from=21-2, to=20-2]
	\arrow[from=21-3, to=20-3]
	\arrow[from=21-4, to=20-4]
	\arrow[from=21-5, to=20-5]
	\arrow[from=21-5, to=20-15]
	\arrow[from=21-7, to=20-7]
	\arrow[from=21-9, to=20-9]
	\arrow[from=21-10, to=20-10]
	\arrow[from=21-11, to=20-11]
	\arrow[from=21-13, to=20-13]
	\arrow[from=21-15, to=20-15]
	\arrow[from=22-2, to=21-2]
	\arrow[from=22-3, to=21-3]
	\arrow[from=22-4, to=21-4]
	\arrow[from=22-5, to=21-5]
	\arrow[from=22-7, to=21-7]
	\arrow[from=22-9, to=21-9]
	\arrow[from=22-10, to=21-10]
	\arrow[from=22-10, to=21-15]
	\arrow[from=22-11, to=21-11]
	\arrow[from=22-13, to=21-13]
\end{tikzcd}}\]
\caption{\label{figure:symbol_write_error_diagram} \textbf{Directed graphical model for a cycle of $\mathcal{U}$.} Shown is a complete cycle of the staged pseudo-UTM $\mathcal{U}$. Vertices represent random variables and arrows show the dependence relations. Columns are respectively random variables representing the squares on the description tape (column $1$, counting from the left), staging tape (columns $2-4$), state tape (column $5$), work tape squares, UTM state (second column from the right) and the timestep $\mu$ of the UTM within the cycle (rightmost column). Shown for reference is a random variable $\theta$ corresponding to a $\sigma'$ entry on the description tape.}
\end{figure}
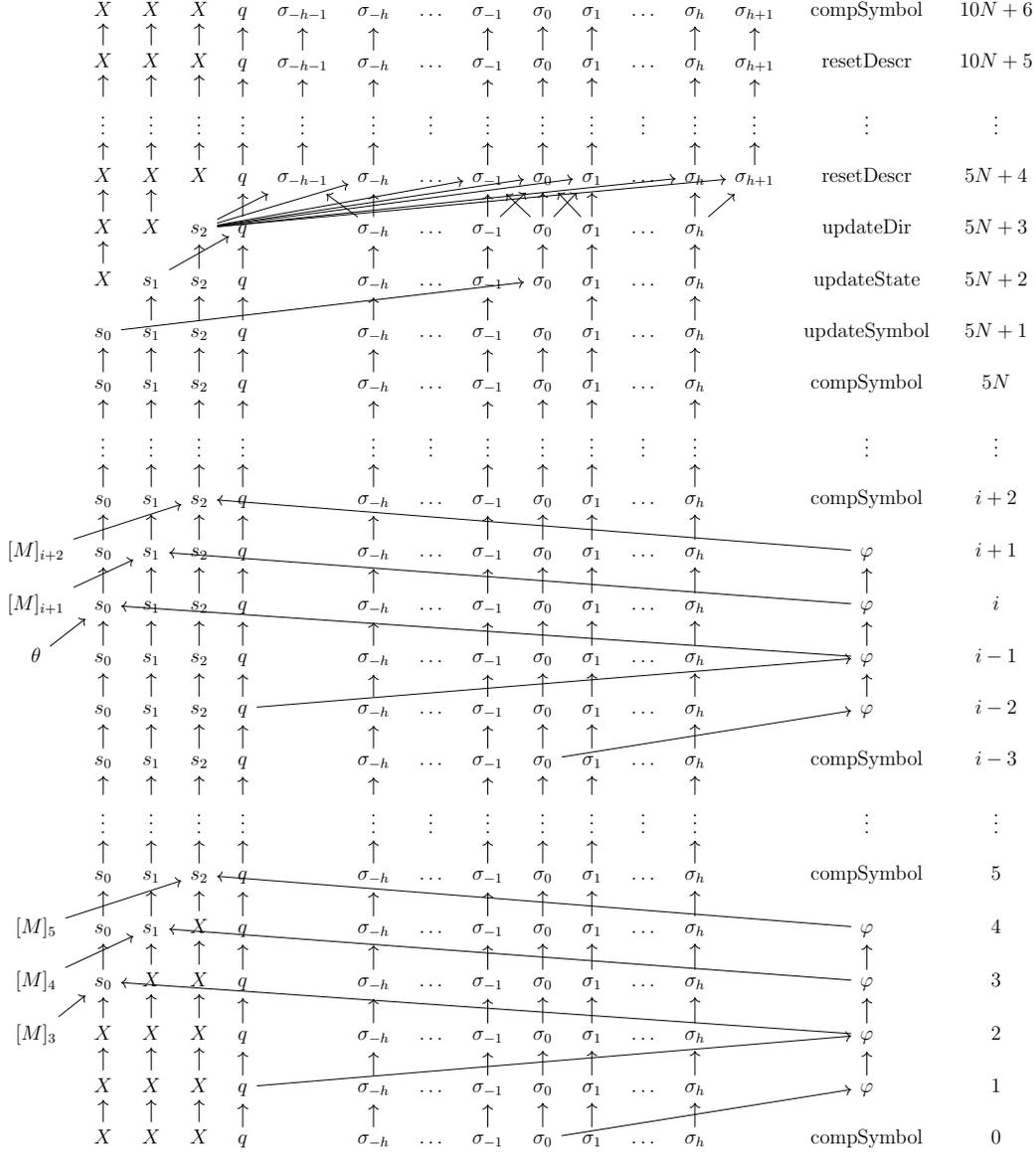

\subsection{Computation Paths for $\mathcal{U}$}\label{section:relstep_UTM}

In this section we enumerate the computation paths for the pseudo-UTM $\mathcal{U}$ from each of its $3N$ inputs (which recall represent squares on the description tape, which contains $N$ tuples) to the final state represented by the type $B$. Recall that we write description tape entries as $\sigma,q,\sigma',q',d$. For each index $1 \le j \le N$ we represent the corresponding tuple by $\sigma_j, q_j, \sigma'_j, q'_j, d_j$. Computation paths may be inferred from stacking two copies of Figure \ref{figure:symbol_write_error_diagram} to represent two cycles of the UTM.

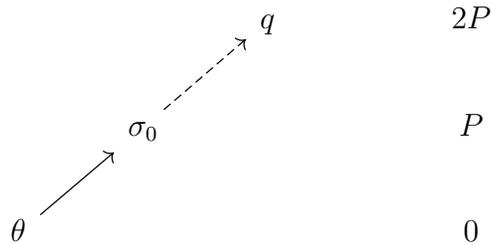
\begin{figure}[htpb]
    \centering
\[\begin{tikzcd}
	&& q && 2P \\
	& {\sigma_0} &&& P \\
	\theta &&&& 0
	\arrow[dashed, from=2-2, to=1-3]
	\arrow[from=3-1, to=2-2]
\end{tikzcd}\]
    \caption{\label{figure:write_error_two_steps_main} A collapsed version of two stacked diagrams as in Figure \ref{figure:symbol_write_error_diagram} (that is, two simulated timesteps of a TM, considering a possible error in either the symbol to write $\sigma'$ or the direction to move $d$) showing only paths that begin at $\theta$ and end at $q$. Dashed lines stand for $N$ computation paths (where $N$ is the number of tuples on the description tape) and $P$ for the number of UTM steps required to simulate a single step of $M$.
    }
\end{figure}

\begin{figure}[htpb]
    \centering
\[\begin{tikzcd}
	&& q && 2P \\
	\theta && q && P \\
	\theta &&&& 0
	\arrow[from=2-1, to=1-3]
	\arrow[shift left, dashed, from=2-3, to=1-3]
	\arrow[shift right, from=2-3, to=1-3]
	\arrow[from=3-1, to=2-3]
\end{tikzcd}\]
\caption{\label{figure:state_error_two_steps_main} A collapsed version of two stacked diagrams as in Figure \ref{figure:symbol_write_error_diagram} (that is, two simulated timesteps of a TM, considering a possible error in the state to transition $q'$ to on the description tape) showing only paths that begin at $\theta$ and end at $q$.}
\end{figure}
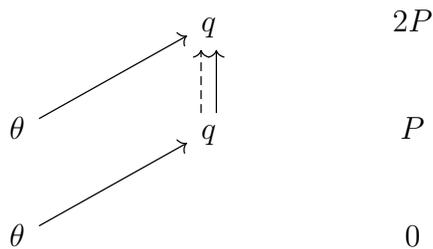

In Figure \ref{figure:write_error_two_steps_main} and Figure \ref{figure:state_error_two_steps_main} we present a graphical model of the random variables involved in the forward propagation of uncertainty through the UTM from the description tape. The computation paths through the plain proof encoding $2P$ steps of the staged pseudo-UTM (so two steps of the simulated machine) which begin at a description tape entry $\theta$ and end at $q$ depend on whether $\theta$ is a $\sigma',q'$ or $d$-type square:
\begin{itemize}
    \item $\sigma'$-\textbf{type}: there are computation paths from $\theta$ to $q$
    \be
    \Gamma_j \qquad 1 \le j \le N
    \ee
    as shown in Figure \ref{figure:write_error_two_steps_main}. This path goes through the $\sigma_0$ after one simulated timestep and then through the second entry on the staging tape to $q$. In the second cycle of the UTM, simulating the second step of the machine, the first two squares of each description tuple are read and compared to $\sigma_0, q$. One copy of $\sigma_0$ is used in each comparison, and the computation path $\Gamma_j$ represents the copy of $\sigma_0$ (which is produced using one copy of $\theta$) being used in the comparison with $\sigma_j$, the symbol on the $j$th tuple of the description tape. This comparison affects which state is written to the staging tape and eventually used to produce $q$.
    \item $q'$-\textbf{type}: there are computation paths from $\theta$ to $q$
    \be
    \Lambda_j \qquad 1 \le j \le N\,, \qquad \Omega\,, \qquad \Xi
    \ee
    as shown in Figure \ref{figure:state_error_two_steps_main}. The path $\Lambda_j$ goes through the $q$ after one simulated timestep and represents $\theta$ first being used to compute the state after one step, and then this output being used in the comparison with $q_j$ in order to update the staging tape. There is another computation path $\Omega$ where $\theta$ is used directly to produce the state in the second step. There is one additional computation path $\Xi$ where $\theta$ is used to compute the state after one simulated timestep, and then no match on the description tape is found, and so this state is carried through to the second simulated timestep.
    \item $d$-\textbf{type}: there are computation paths from $\theta$ to $q$
    \be
    \Theta_j \qquad 1 \le j \le N
    \ee
    as shown in Figure \ref{figure:write_error_two_steps_main}. As above this path goes through $\sigma_0$ and represents the copy of $\sigma_0$ (and thus $\theta$) used in the comparison with $\sigma_j$. The difference is that the original $\theta$ affects $\sigma_0$ in this case through the direction the tape head moves in the first step, rather than by writing a symbol.
\end{itemize}

\subsection{Influence Functions}\label{section:influence_func_utm}

We return to the question of understanding the computational content of the influence functions (Definition \ref{defn:influence}) for $\bold{i} \in \mathbb{N}^d$ (see Remark \ref{remark:s_in_Nd})
\begin{equation}
g_{\bold{i}}(x) = \frac{\partial^{|\bold{i}|}}{\partial w_1^{i_1} \cdots \partial w_d^{i_d}}\, p\big( y \neq y(x) \mid x, w \big) \Bigr|_{w = [M]}\,.
\end{equation}
However first we have to fix coordinates. In Section \ref{section:geometry} the local coordinates $w_1,\ldots,w_d$ at $[M]$ in $W$ were arbitrary (Remark \ref{remark:coord_arb}). However, now we have to fix the coordinates to be as specified in \eqref{eq:w_coord_simplex}. That is, $w$ parametrises deviations in the probability simplex from $[M]$ (which is $w = \bold{0}$) with one coordinate for each position on the description tape and each symbol, or state or direction other than the one given in that position in $[M]$.

\begin{corollary} For $\bold{0} \neq \bold{k} \in \mathbb{N}^d$
\be
g_{\bold{k}}(x) = \sum_{\bold{0} \neq \bold{s} \le \bold{k}} S^{\bold{k}}_{\bold{s}} A^{\bold{s}}(x)
\ee
where $S^{\bold{k}}_{\bold{s}}$ is as in \eqref{eq:s_coeff_theorem} and $A^{\bold{s}}(x)$ is as in Definition \ref{defn:Asfunctions}.
\end{corollary}
\begin{proof}
Immediate from \eqref{eq:theorem_eq_1} in the proof of Theorem \ref{theorem:main_error_syndrome_general} since $\Delta \psi = \Delta\textrm{step}^t(x,w)$ when we take local coordinates around $[M]$. We can assume $\bold{0} \neq \bold{s}$ by Remark \ref{remark:A_s_zero}.
\end{proof}

\begin{example}\label{example:weight_one_error}
Let $\bold{k} = e_k$ be the standard basis vector in $\mathbb{R}^d$ for some $1 \le k \le d$. Such an index corresponds to a tuple on the description tape, within that tuple one of the three squares $\sigma', q', d$ and within that square one of the possible values $z$ \emph{other} than what is specified in $[M]$. Then $w_k = 0$ means the same value appearing in $[M]$ and $w_k = 1$ means that we replace that entry by some other symbol, state or direction.

Then we write $g_k(x)$ for $g_{\bold{k}}(x)$ and this function measures the infinitesimal variation in the probability of an output error caused by an infinitesimal variation in the code in this particular square, away from its value in $[M]$ towards $z$. We have (see Remark \ref{remark:skk})
\begin{align*}
g_k(x) = S^{\bold{k}}_{\bold{k}} A^{\bold{k}}(x) = A^{\bold{k}}(x) = \Big| \big\{ \gamma \l \operatorname{wt}(\gamma) = e_k \text{ and } \mathcal{U}(x,\gamma) \neq M(x) \big\} \Big|\,.
\end{align*}
We also denote this quantity simply by $A^k(x)$. To say that $\operatorname{wt}(\gamma) = e_k$ is to say that $\gamma$ assigns no error to the computation paths from all other squares on the description tape, assigns to exactly one computation path from the square associated to $k$ an error, and that error is of type $z$. Note that since there are generally many computation paths, such an error syndrome is far from unique.
\end{example}

Combining this with Lemma \ref{lemma:partial_derivatives_H} we obtain

\begin{theorem}\label{theorem:partial_formula_explicit} For $|\bold{k}| > 0$
\be\label{eq:partial_formula_explicit}
\frac{\partial^{|\bold{k}|}}{ \partial w_1^{k_1} \cdots \partial w_d^{k_d}} H \Bigr|_{w = [M]} = \sum_{\substack{\bold{i},\bold{j} \neq 0 \\ \bold{i} + \bold{j} = \bold{k}}} \sum_{\substack{\bold{0} \neq \bold{s} \le \bold{i}\\ \bold{0} \neq \bold{t} \le \bold{j}}} C(\bold{i}, \bold{j})  S^{\bold{i}}_{\bold{s}} S^{\bold{j}}_{\bold{t}} \,\mathbb{E}_x\Big[ A^{\bold{s}}(x) A^{\bold{t}}(x) \Big]
\ee
where $C(\bold{i}, \bold{j}) = \prod_{s=1}^d \binom{i_s + j_s}{i_s}$, $\mathbb{E}_x[-]$ denotes expectation with respect to $q(x)$ and $\bold{i}, \bold{j} \in \mathbb{N}^d$.
\end{theorem}
\begin{proof}
We have by Lemma \ref{lemma:partial_derivatives_H}
\begin{align}
\frac{\partial^{|\bold{k}|}}{ \partial w_1^{k_1} \cdots \partial w_d^{k_d}} H \Bigr|_{w = [M]} &= \sum_{\substack{\bold{i},\bold{j} \neq 0 \\ \bold{i} + \bold{j} = \bold{k}}} C(\bold{i}, \bold{j})\, \mathbb{E}_x\Big[ g_{\bold{i}}(x) g_{\bold{j}}(x) \Big] \nonumber\\
&= \sum_{\substack{\bold{i},\bold{j} \neq 0 \\ \bold{i} + \bold{j} = \bold{k}}} C(\bold{i}, \bold{j})\,\mathbb{E}_x\Big[ \Big(\sum_{\bold{0} \neq \bold{s} \le \bold{i}} S^{\bold{i}}_{\bold{s}} A^{\bold{s}}(x) \Big) \Big(\sum_{\bold{0} \neq \bold{t} \le \bold{j}} S^{\bold{j}}_{\bold{t}} A^{\bold{t}}(x) \Big) \Big] \nonumber\\
&= \sum_{\substack{\bold{i},\bold{j} \neq 0 \\ \bold{i} + \bold{j} = \bold{k}}} \sum_{\substack{\bold{0} \neq \bold{s} \le \bold{i}\\ \bold{0} \neq \bold{t} \le \bold{j}}} C(\bold{i}, \bold{j})  S^{\bold{i}}_{\bold{s}} S^{\bold{j}}_{\bold{t}} \,\mathbb{E}_x\Big[ A^{\bold{s}}(x) A^{\bold{t}}(x) \Big]
\end{align}
as claimed.
\end{proof}

The only part of \eqref{eq:partial_formula_explicit} that depends on the machine $M$ are the correlation functions $\mathbb{E}_x\big[ A^{\bold{s}}(x) A^{\bold{t}}(x) \big]$. It is in this sense that the combinatorics and statistics of error syndromes determine the geometry of $H$ at $[M]$. An important special case is the Hessian of $H$:

\begin{corollary}\label{corollary:partial_formula_explicit_hessian} For $1 \le i, j \le d$ we have
\be\label{eq:partial_formula_explicit_hessian}
\frac{\partial^{2}}{ \partial w_i \partial w_j} H \Bigr|_{w = [M]} = 2\, \mathbb{E}_x\Big[ A^{i}(x) A^{j}(x) \Big]\,.
\ee
\end{corollary}

\begin{remark}\label{rmk:fourier_analysis}
In the analysis of Boolean functions \(f: \{0,1\}^n \lto \{0,1\}\) one often studies the \emph{influence} of the \(i\)-th variable on \(f\), denoted
\[
\mathrm{Inf}_i(f) \;=\; \Pr_{x}\bigl[f(x) \neq f(x^{\oplus i})\bigr],
\]
where \(x^{\oplus i}\) is obtained from \(x\) by flipping its \(i\)-th bit. That is, \(\mathrm{Inf}_i(f)\) is the probability that flipping the 
$i$-th bit in the input flips the value of the function. A rich theory has developed around these notions, often called the \emph{Fourier (or spectral) analysis of Boolean functions}, culminating in classical results such as the Poincar\'e inequality, KKL theorem, Friedgut's junta theorem, and so on \cite{o2014analysis}.

The influence functions $g_i(x)$ of this paper are conceptually similar, but distinct, in that they measure the infinitesimal rate at which perturbing the $i$th bit in the code of the Turing machine $M$ changes the probability of an output error on $x$. By Example \ref{example:weight_one_error} this is a count of weight one error syndromes which induce an output error on $x$.
\end{remark}

\subsection{Meaning of Nondegeneracy}\label{section:meaning_nondeg}

The germ $([M], H)$ is nondegenerate if and only if the functions $\{ g_i \}_{i=1}^d$ are linearly independent (Lemma \ref{lemma:slt25}). Since $K = L-L([M])$ is comparable to $H$ by Lemma \ref{lemma:comparable_H_K_L}, applying Lemma \ref{lemma:comparable_hessian} to $H,K$ shows that the same is true of $([M], L)$: subtracting a constant does not change the Hessian. We can now understand the computational content of this condition: it says that the pattern of counts of error syndromes in the $i$th square of the description tape which flip the output state cannot be computed from the pattern for the \emph{other} squares of the description tape. In short, nondegeneracy means that \emph{different parts of the code have different effects}.

More generally, the geometry of the singularity $([M], H)$ is determined by structure in the Taylor series expansion, which by Theorem \ref{theorem:partial_formula_explicit} is the same thing as structure in the combinatorics of error syndromes. That is, whenever we can show that the effect on the final state of errors in some bits \emph{interact} with errors in other bits, we can expect this to show up in the combinatorics and thus potentially in the geometry of $([M], H)$. In this way, we predict a general correspondence between geometric structure of the singularity and the internal structure of the Turing machine. Several particular cases are presented in the next section.

\section{Examples}
\label{sec:examples}

\subsection{Error-Correction}\label{section:error_correction}

The concepts of fault-tolerance and error-correction have been present in theoretical computer science since its earliest foundational works \cite{von1956probabilistic}. Von Neumann demonstrated that a bounded level of noise or faults occurring at the level of individual elements can be systematically mitigated. His method involved structuring computations redundantly by executing each computational step multiple times independently and using majority voting. In this way he showed that local errors could be corrected with high probability.

Other frameworks for fault-tolerance have been developed, including the formal notion of fault-tolerant Turing machines \cite{pippenger1985networks,spielman1996highly}. Such machines explicitly incorporate mechanisms for detecting and correcting faults that occur at the level of transitions, tape symbols, or internal states during computation. Under some conditions Çapuni and Gács \cite{ccapuni2012turing, çapuni2021reliableturingmachine} have constructed general machines of this kind.\footnote{We distinguish between noisy Turing machines and probabilistic Turing machines. While probabilistic Turing machines employ randomness deliberately as a computational resource, noisy or fault-tolerant Turing machines address a fundamentally different challenge: preserving computational correctness and universality despite the involuntary introduction of random faults.}

We must specify which errors are allowed. If every state-write variation is allowed and $|Q|>1$, an error in the final state write can change the output, with no subsequent step available to correct it. We therefore fix a set of allowed error coordinates. In the coordinates of \eqref{eq:w_coord_simplex}, centred at $[M]=\bold{0}$, choose $J \subseteq \{1,\ldots,d\}$ and put
\[
W_J = \{w \in W \mid w_k=0 \text{ for } k\notin J\},\qquad d_J=|J|,\qquad H_J=H|_{W_J}.
\]
Thus $W_J$ is the face on which only the error variations indexed by $J$ can occur. Here $J$ indexes individual error variations, rather than the $3N$ description-tape squares: a square may have several possible erroneous values. Allowing noise on selected squares is the special case in which $J$ contains all the error coordinates associated to those squares.

We equip $W_J$ with a smooth prior density $\varphi_J$ that is positive at $[M]$, relative to $d_J$-dimensional volume, and write
\[
\lambda_J([M],q) = \operatorname{RLCT}_{(W_J)_{[M]}}(H_J;\varphi_J),
\]
where $(W_J)_{[M]}$ is a sufficiently small compact semi-analytic neighbourhood of $[M]$ in $W_J$. By Lemma \ref{lemma:comparable_H_K_L}, this is also the local learning coefficient of the statistical model restricted to $W_J$. If $H_J$ vanishes identically near $[M]$, we set $\lambda_J([M],q)=0$, since the corresponding local learning integral has no polynomial decay.

\begin{definition}\label{defn:J_error_correction} An error syndrome $\gamma$ is \emph{$J$-admissible} if $\operatorname{supp}(\operatorname{wt}(\gamma))\subseteq J$. For any integer $C \ge 0$ we say that a Turing machine $M$ can \emph{correct $J$-admissible errors of weight $\le C$} on $I \subseteq \Sigma^*$ if for every $J$-admissible error syndrome $\gamma$ with $|\operatorname{wt}(\gamma)| \le C$
\be
\mathcal{U}(x, \gamma) = M(x) \qquad \forall x \in I\,.
\ee
\end{definition}

Hidden in the notation is that $M(x)$ is the state after running the machine for a \emph{fixed number $t$ of steps}. So informally this definition says: if we simulate $M$ on the UTM $\mathcal{U}$ for $t$ steps and, when this execution would make use of a sample from the distributions of symbols on the description tape, we make use of the outcomes provided by a $J$-admissible $\gamma$, the resulting state $\mathcal{U}(x,\gamma)$ is $M(x)$, that is, the outcome is the same as if there had been no errors (here \emph{errors} are any time the sample we obtain is different from the entry in that position in $[M]$). Syndromes that are not $J$-admissible have zero probability on $W_J$, because their probability contains a factor $w_k$ with $k\notin J$. We note again the unusual model of noisy execution of a UTM that we are employing (see Appendix \ref{section:think_noisy}).

\begin{corollary}\label{corollary:bound_llc_TM} If $M$ is a classical solution on $I$ and is able to correct $J$-admissible errors of weight $\le C$ on $I$, then
\be\label{eq:bound_llc_TM}
\lambda_J([M], q) \le \frac{d_J}{2(C+1)}\,.
\ee
\end{corollary}
\begin{proof}
If $H_J$ vanishes identically, the claim follows from our convention. Otherwise, let $\mathcal{P}_J(M)$ and $l_J(M)$ be the Newton polyhedron of $H_J$ in $\mathbb{R}^{d_J}$ and its distance. We consider only multi-indices $\bold{k}$ supported in $J$; the conditions $\bold{i}+\bold{j}=\bold{k}$, $\bold{s}\le\bold{i}$ and $\bold{t}\le\bold{j}$ then force $\bold{i},\bold{j},\bold{s},\bold{t}$ to be supported in $J$ as well.

If $\bold{i}, \bold{j} \neq 0$ and $\bold{s} \le \bold{i}, \bold{t} \le \bold{j}$, $\bold{i} + \bold{j} = \bold{k}$ then from $|\bold{s}| \ge C+1$, $|\bold{t}| \ge C+1$ we deduce
\[
|\bold{k}| = |\bold{i}| + |\bold{j}| \ge |\bold{s}| + |\bold{t}| \ge 2C+2\,.
\]
It follows from Theorem \ref{theorem:partial_formula_explicit} that we have
\[
\frac{\partial^{|\bold{k}|}}{ \partial w_1^{k_1} \cdots \partial w_d^{k_d}} H \Bigr|_{w = [M]} = 0
\]
unless $|\bold{k}| \ge 2(C + 1)$, since by assumption $A^{\bold{s}}(x) = 0$ for all $|\bold{s}| \le C$ supported in $J$. These derivatives determine the Taylor coefficients of $H_J$. Then the Newton polyhedron $\mathcal{P}_J(M)$ is the convex hull of the exponent vectors occurring in $H_J$, each of which has total degree at least $2(C+1)$, plus the nonnegative orthant. Since this half-space is convex and closed under addition of the nonnegative orthant, this means that
\[
\mathcal{P}_J(M) \subseteq \Big\{ \bold{u}\in\mathbb{R}_{\ge 0}^{d_J} \,\Big|\, \sum_{a=1}^{d_J}u_a \ge 2(C + 1) \Big\}
\]
and in particular $(s,s,\ldots,s) \in \mathcal{P}_J(M)$ implies $d_Js \ge 2(C + 1)$ and so $s \ge \frac{2(C+1)}{d_J}$. The distance $l_J(M)$ is the infimum of a subset of $[\frac{2(C+1)}{d_J}, \infty)$ and hence $l_J(M) \ge \frac{2(C+1)}{d_J}$. Thus $\frac{1}{l_J(M)} \le \frac{d_J}{2(C+1)}$ and the claim now follows from the argument of Proposition \ref{prop:upper_bound_RLCT} applied to $H_J$ on $W_J$.
\end{proof}

Note that any Turing machine can correct $C = 0$ errors, in which case the upper bound of \eqref{eq:bound_llc_TM} is $\frac{d_J}{2}$. Taking $J=\{1,\ldots,d\}$ recovers the universal bound $\lambda([M],q)\le d/2$ on the full parameter space \cite[\S 7]{watanabeAlgebraicGeometryStatistical2009}, \cite[\S 4.2.3]{lin2011algebraic}.

\begin{remark}
If $M$ does not make use of all the states in $Q$, it might be modified in such a way as to implement error-correction for a specified set of allowed errors. Using additional tuples can affect the learning coefficient, while Corollary \ref{corollary:bound_llc_TM} bounds the coefficient of the model restricted to the chosen face. Whether such a modification decreases the learning coefficient must therefore be considered with the noise model specified.

Simple majority voting schemes may be useful when the correction mechanism is protected. Allowing errors in that mechanism is a further problem, as in \cite{ccapuni2012turing}. In our fixed-time output model, the final state-write obstruction must also be addressed; the restricted bound alone makes no claim about the local learning coefficient on the full space $W$.
\end{remark}

\begin{remark} A relation between error sensitivity of neural networks in their forward pass and generalisation error has been noted by \cite{arora2018stronger}. It has also been suggested that evolution in the presence of noise tends to produce error-correcting codes and modularity \cite{mccourt2023noisydynamicalsystemsevolve}.
\end{remark}

\subsection{Control Flow}\label{section:control_flow}

Suppose that $I$ can be written as a disjoint union $I = I_1 \sqcup I_2$ and the Turing machine $M$ on input $x$ first classifies whether $x$ belongs to $I_1$ or $I_2$ and then executes separate control paths for each subset. We divide the tuples (and thus the parameters $w_1,\ldots,w_d$) into three blocks, where the first and second blocks correspond to tuples in the first and second control path and the third block performs the classification of the input.

Then for $x \in I_1$ we have $g_j(x) = 0$ for any $j$ in the second block, and for $x \in I_2$ we have $g_i(x) = 0$ for any $i$ in the first block, since no error syndrome in one block can affect the output for inputs handled by the other block. Hence for $i$ in the first block and $j$ in the second we have by Corollary \ref{corollary:partial_formula_explicit_hessian}
\begin{equation}
\frac{1}{2} \frac{\partial^2}{ \partial w_i \partial w_j} H \Bigr|_{w = [M]} = \mathbb{E}_x\Big[ g_i(x) g_j(x) \Big] = 0\,.
\end{equation}
This means that there is a block structure in the Hessian of $H$
\be
\operatorname{Hess}(H, [M]) = \begin{pmatrix}
H_1 & 0 & A_{13} \\
0 & H_2 & A_{23} \\
A_{31} & A_{32} & A_{33}
\end{pmatrix}\,.
\ee
This is the geometric signature of control flow in the Turing machine.

\begin{remark}
In general, increasing degeneracy of $H$ looks like vanishing of higher derivatives, or relations among these derivatives; in turn this by Theorem \ref{theorem:partial_formula_explicit} takes the form of certain functions $f(x)$ with the property that $\mathbb{E}_x[f(x)] = 0$. That is, as \emph{statistical patterns} on the sample space \cite{shawe2004kernel}. These combinatorial relations among error syndrome counts in return reflect the patterns by which parts of the Turing machine code coordinate to produce the output. %
\end{remark}

\subsection{Nondegenerate Direction for \detectA}\label{section:nondeg_detectA}

We use the example of the synthesis problem {\detectA} from \cite{clift2021geometry} to illustrate the correspondence between geometry and structure. Recall that
\begin{gather*}
    \Sigma=\{0,1,2\} =\{\Box, A, B\},\\
    Q=\{0,1\} = \{\reject,\accept\}\,.
\end{gather*}
Consider the function $f: \{A,B\}^* \lto \{\reject,\accept\}$ that returns $\accept$ if the string contains at least one $A$, and $\reject$ otherwise. In Figure \ref{fig:detectA} we give two Turing machines that implement this function when the string is initially placed on the work tape with the head over the first symbol (arranged to the right on the tape). The machines are initialised in the $\reject$ state. %

\begin{figure}[htbp]
    \centering
    \begin{subfigure}[b]{0.45\textwidth}
        \centering
        $$
\begin{array}{|c|c?c|c|c|}
\hline
\sigma & q & \sigma' & q' & d \\
\hline
\Box & \reject & \Box & \reject & \text{Stay} \\
\Box & \accept & \Box & \accept & \text{Stay} \\
A & \reject & A & \accept & \text{Stay} \\
A & \accept & A & \accept & \text{Stay} \\
B & \reject & B & \reject & \text{Right} \\
B & \accept & B & \accept & \text{Stay}\\
\hline
\end{array}
$$
        \caption{$\detectA^{(0)}$}
        \label{fig:detectA_left}
    \end{subfigure}
    \hfill
    \begin{subfigure}[b]{0.45\textwidth}
        \centering
        $$
\begin{array}{|c|c?c|c|c|}
\hline
\sigma & q & \sigma' & q' & d \\
\hline
\Box & \reject & B & \reject & \text{Right} \\
\Box & \accept & \Box & \accept & \text{Left} \\
A & \reject & \Box & \accept & \text{Left} \\
A & \accept & \Box & \accept & \text{Left} \\
B & \reject & \Box & \reject & \text{Right} \\
B & \accept & \Box & \accept & \text{Left}\\
\hline
\end{array}
$$
        \caption{$\detectA^{(1)}$}
        \label{fig:detectA_right}
    \end{subfigure}
    \caption{(a) A natural implementation of $\detectA$. Each row is a tuple of the form $(\sigma, q, \sigma', q', d)$. (b) A second, distinct machine implementing the same function, implicitly given in \cite[Example 3.5]{clift2021geometry} and discovered there by MCMC-based program synthesis.}
    \label{fig:detectA}
\end{figure}

We consider errors which perturb the state $q'$ transitioned to by both machines, when $\sigma = A$ and $q = \reject$. That is, the variation
\be\label{eq:example_detectA_bit}
(A, \reject, ?, {\accept \rightarrow \reject}, ?)\,.
\ee
For each machine, let $J_\theta$ consist of the single error coordinate in \eqref{eq:example_detectA_bit}. Throughout this calculation we restrict the noisy-code space to $W_{J_\theta}$, fixing every other description-tape entry at its value in that machine. The error syndromes below are $J_\theta$-admissible.

We let $\theta$ denote the random variable representing this $q'$-type square on the description tape (terminology of Section \ref{section:relstep_UTM}). Hence if we use the UTM to simulate the TMs for two steps, the computation paths from $\theta$ to the final state $q$ are
\[
\Lambda_{\Box, \reject}, \ldots, \Lambda_{B, \accept}, \Omega, \Xi
\]
where we note $N = 6$ and we label computation paths by the associated $\sigma,q$ for the tuple. Hence an error syndrome is a sequence $\gamma \in \{0,1\}^8$ where we recall $0$ is read as the symbol present in the code $[M]$ and $1$ as the alternative presented in the variation \eqref{eq:example_detectA_bit}. So for example $\gamma = (1,0,0,0,0,0,0,0)$ (we write $10000000$) means that we evaluate the UTM on a given input, using as our values of the entry $\theta$ the states
\[
\Lambda_{\Box, \reject} = \reject, \Lambda_{\Box, \accept} = \accept, \ldots, \Lambda_{B, \accept} = \accept, \Omega = \accept, \Xi = \accept
\]
where the fact that the first entry of $\gamma$ is $1$ is what causes the first state here to be $\reject$. What this means is that, while evaluating the UTM with $[\detectA^{(0)}]$ on its description tape and $x$ on its work tape, in the second cycle of $P$ steps of the UTM when we go to compare $\Box, \reject$ with $\sigma'_0, q'$ we use the state computed in the first step using the tuple
\[
A, \reject, A, {\color{red} \reject }, \text{Stay} \quad \text{rather than} \quad A, \reject, A, \accept, \text{Stay}\,.
\]
In all other computation paths we use the tuple with $\accept$, that is, without the error. For example on input $x = A$ with error syndrome $10000000$, we have (writing $\sigma_0 = A, q = \reject$ for the inital symbol and state, $\sigma'_0, q'$ for the symbol and state after one timestep of the simulated machine and $q''$ for the state after two timesteps)
\begin{align*}
q''(A) &= q''\big(\sigma'_0(A, \reject), q'(A, \reject, \overset{{\color{blue} \Box,\reject}}{\theta}), \ldots, q'(A, \reject, \overset{{\color{blue} B,\accept}}{\theta}), \overset{{\color{blue} 7}}{\theta}, \overset{{\color{blue} 8}}{\theta}\big)\\
&= q''\big( A, q'(A, \reject, 1), q'(A, \reject, 0), \ldots, q'(A, \reject, 0), 0, 0\big)\\
&= q''\big( A, \reject, \accept, \ldots, \accept, 0, 0 \big) = \accept\,.
\end{align*}
The syndrome $\gamma = 00100000$ is more interesting. This corresponds to an error occurring in the copy of $\theta$ used to compute $q'$ in the comparison of the current symbol and state with the first two squares $A, \reject$ of the third tuple of $\detectA^{(0)}$. Thus
\begin{align*}
q''(A) &= q''\big(\sigma'_0(A, \reject), \ldots, q'(A, \reject, \overset{{\color{blue} A,\reject}}{\theta}), q'(A, \reject, \overset{{\color{blue} A,\accept}}{\theta}), \ldots, \overset{{\color{blue} 7}}{\theta}, \overset{{\color{blue} 8}}{\theta}\big)\\
&= q''\big( A, \ldots, q'(A, \reject, 1),q'(A, \reject, 0), \ldots, 0, 0\big)\\
&= q''\big( A, \ldots, \reject, \accept, \ldots, 0, 0 \big) = \accept\,.
\end{align*}
While scanning the description tape the UTM tries to compare $A, \reject$ with the current symbol and state, gets a match on the first and then (mistakenly) a match on the second (due to the error in computing $q'$, which should be $\accept$ after the first step, but due to the error in $\theta$ is instead $\reject$). Hence the UTM writes $\accept$ to the staging tape. It then moves on to the comparison for the tuple beginning $A, \accept$ and \emph{receives a second match}, writing $\accept$ to the staging tape again.

We see that none of the syndromes of weight one can cause an error (that is, flip the output from $\accept$ to $\reject$) on the input $x = A$. However with syndrome $\gamma = 00000010$ and $x = BA$ there is an output error, as
\begin{align*}
q''(BA) &= q''\big(\sigma'_0(B, \reject), q'(B, \reject, \overset{{\color{blue} \Box,\reject}}{\theta}), \ldots, q'(B, \reject, \overset{{\color{blue} B,\accept}}{\theta}), \overset{{\color{blue} 7}}{\theta}, \overset{{\color{blue} 8}}{\theta}\big)\\
&= q''\big( A, \reject, \ldots, \reject, 1, 0 \big) = {\color{red} \reject}\,.
\end{align*}
Continuing in this way we fill out the tables in Figure \ref{fig:detectA_syndromes}.

As we have explained in Section \ref{section:geometry}, to second order the geometry of the singularity germ $([\detectA^{(j)}], L)$ for $j \in \{0,1\}$ is closely related to the geometry of the configuration of the functions $\{ g^{(j)}_i \}_{i=1}^d$ in the vector space of functions. From Figure \ref{fig:detectA_syndromes} we compute, viewing the functions $g^{(j)}$ as vectors of coefficients in the basis of functions on $A, B, AB, BA, AA, BB$ (in that order)
\begin{align}
    g^{(0)} &= \begin{pmatrix} 0 & 0 & 0 & 1 & 0 & 0 \end{pmatrix}^T\,,\\
    g^{(1)} &= \begin{pmatrix} 0 & 0 & 0 & 1 & 0 & 0 \end{pmatrix}^T
\end{align}
by counting the number of red entries for each $x$. 

If we assume all the inputs have equal probability $\tfrac{1}{6}$, and computing the Hessian with respect to just the variable $w$ (which has value $w = 0$ when the state $q'$ in \eqref{eq:example_detectA_bit} is $\accept$ and the value $w = 1$ when it is $\reject$) then by \eqref{eq:hessian_formula_P} we have for $j \in \{0,1\}$
\be
\operatorname{Hess}(H|_{W_{J_\theta}}, [\detectA^{(j)}]) = \frac{1}{3} (g^{(j)})^T g^{(j)} = \frac{1}{3}\,.
\ee
Thus, on the one-dimensional face $W_{J_\theta}$,
\[
H_{J_\theta}(w)=\tfrac16 w^2+O(w^3),\qquad \lambda_{J_\theta}([\detectA^{(j)}],q)=\tfrac12
\]
for a smooth positive prior on this face. Even though these two Turing machines are distinct, their differences are irrelevant to errors in this particular tuple, so they have the same geometry in the $w$ direction. This restriction does not correct all admissible errors of weight one on the six inputs: the syndrome on $\Omega$ causes an output error on $BA$. In Section \ref{section:full_geometry} we exhibit a different, 24-dimensional face on which all admissible errors of weight one are corrected.

\begin{figure}[htbp]
    \centering
    \begin{subfigure}[b]{0.45\textwidth}
        \centering
        $$
\begin{array}{|c|c|c|}
\hline
\text{Syndrome} & x & \\
\hline
*00 & A & \accept \\
\bold{0}10 & A & \accept \\
\bold{0}01 & A & \accept \\
*00 & B & \reject \\
\bold{0}10 & B & \reject \\
\bold{0}01 & B & \reject \\
*00 & AB & \accept \\
\bold{0}10 & AB & \accept \\
\bold{0}01 & AB & \accept \\
*00 & BA & \accept \\
\bold{0}10 & BA & {\color{red} \reject} \\
\bold{0}01 & BA & \accept \\
*00 & AA & \accept \\
\bold{0}10 & AA & \accept \\
\bold{0}01 & AA & \accept \\
*00 & BB & \reject \\
\bold{0}10 & BB & \reject \\
\bold{0}01 & BB & \reject \\
\hline
\end{array}
$$
        \caption{$\detectA^{(0)}$}
        \label{fig:detectA_left_syndromes}
    \end{subfigure}
    \hfill
    \begin{subfigure}[b]{0.45\textwidth}
        \centering
        $$
\begin{array}{|c|c|c|}
\hline
\text{Syndrome} & x & \\
\hline
*00 & A & \accept \\
\bold{0}10 & A & \accept \\
\bold{0}01 & A & \accept \\
*00 & B & \reject \\
\bold{0}10 & B & \reject \\
\bold{0}01 & B & \reject \\
*00 & AB & \accept \\
\bold{0}10 & AB & \accept \\
\bold{0}01 & AB & \accept \\
*00 & BA & \accept \\
\bold{0}10 & BA & {\color{red} \reject} \\
\bold{0}01 & BA & \accept \\
*00 & AA & \accept \\
\bold{0}10 & AA & \accept \\
\bold{0}01 & AA & \accept \\
*00 & BB & \reject \\
\bold{0}10 & BB & \reject \\
\bold{0}01 & BB & \reject \\
\hline
\end{array}
$$
        \caption{$\detectA^{(1)}$}
        \label{fig:detectA_right_syndromes}
    \end{subfigure}
    \caption{Influence functions for weight one error syndromes for the two Turing machines implementing $\detectA$. The bits are ordered by $\Lambda_{\Box,\reject},\ldots,\Lambda_{B,\accept},\Omega,\Xi$. Here $*$ ranges over the six-bit strings of weight one and $\bold{0} = 000000$.}
    \label{fig:detectA_syndromes}
\end{figure}

\subsection{Full geometry of \detectA}\label{section:full_geometry}

In this section we complete the analysis by studying the geometry at $\detectA^{(0)}$ to second-order in all coordinate directions. By Corollary \ref{corollary:partial_formula_explicit_hessian} to compute the second-order derivatives of $H$ at $w = [M]$ we need to count error syndromes of weight one which change the output from $M(x) = y(x)$ (that is, cause an output error). As explained in Remark \ref{example:weight_one_error} we use $1 \le k \le d$ to index triples consisting of
\begin{itemize}
    \item One among the $N = |\Sigma||Q|$ tuples on the description tape
    \item Within that tuple, one of the three positions $\sigma',q',d$
    \item For that position, a variation away from the entry of $[M]$ in that position.
\end{itemize}
For $\detectA^{(0)}$ we have $d = 30$ and the enumeration is given in Figure \ref{fig:detectA_left_coords}. From now on indices $1 \le k \le d$ are associated with triples as indicated by this Figure.

For each $1 \le k \le d$ and input $x$, $A^k(x)$ is the number of computation paths leading from the description tape square associated to $k$ which, when the type of error indexed by $k$ is introduced, lead to an output error. We compute $A^k(x)$ for all $k$ and $x$ in
\be\label{eq:input_set_X}
X = \{ A, B, AB, BA, AA, BB \}
\ee
when we run the simulated machine for two steps $t = 2$. We use the notation for computation paths
\[
\Gamma_j, \Lambda_j, \Omega, \Xi, \Theta_j \qquad 1 \le j \le N
\]
introduced in Section \ref{section:relstep_UTM}.

\begin{lemma}\label{lemma:irrelevant_entries_channel}
$A^k(x) = 0$ for all $k \notin \{3, 13, 18, 23, 24, 25\}$ and $x \in X$.    
\end{lemma}
\begin{proof}
We explain a few cases in detail and leave the remaining cases to the reader. For $k = 1$ the relevant computation paths are of type $\Gamma_j$ (here the $j$ indexes tuples on the description tape and is unrelated to $k$). If in computing the symbol under the head $\sigma_0$ after one simulated timestep we use the variation $\Box \rightarrow A$ for the first tuple in Figure \ref{fig:detectA_left_coords}, this does not affect $\sigma_0$ because the head never starts over a $\Box$. Thus no output error can occur as a result of this variation, and so $A^k(x) = 0$ for $k \in \{1,2,6,7\}$ and $x \in X$. The same logic applies to $k \in \{4,5,9,10\}$ with $\Gamma_j$ replaced by $\Theta_j$.

Since the machine starts in state $\reject$ the variations $k \in \{16, 17,19,20,26,27,29,30\}$ similarly cannot affect $\sigma_0$ after one timestep along the computation paths $\Gamma_j, \Theta_j$ and so $A^k(x) = 0$ also in these cases for all $x \in X$.

When $k = 8$ we have to consider computation paths $\Lambda_j, \Omega, \Xi$. Along the former, again since the head never starts over $\Box$, there can be no induced output error. The other computation path $\Omega$ is more subtle: note that $\detectA^{(0)}$ does not move the head once it transitions into state $\accept$, so under normal operation it \emph{cannot reach} the pair $(\Box, \accept)$. Hence an error in $\Omega$ cannot cause an output error. For $\Xi$, we check each $x$ and observe that a match on the description tape is always found and so this computation path is never taken. Thus, $A^{8}(x) = 0$ for all $x$.\footnote{Note how if we were to consider error syndromes of higher weight we would have to consider the possibility that the machine reaches $\Box, \accept$ due to another error, thus making $\Omega$ for $k = 8$ relevant.}

For $k \in \{11, 12\}$ if $x$ begins with a $B$ this tuple is irrelevant. If it starts with $A$ then the machine transitions into state $\accept$ and it is irrelevant what is written to the tape as $\sigma_0$ (because the machine can never exit state $\accept$). Hence errors of type $k$ on computation paths $\Gamma_j$ cannot cause an output error.\footnote{Notice the possibility for an output error from higher weight syndromes implicit in this reasoning.} For similar reasons $A^k(x) = 0$ for all $k \in \{14, 15\}$ and $x \in X$.

For $k \in \{21, 22\}$ if $x$ begins with $A$ this tuple is irrelevant. If $x$ begins with $B$ then the machine stays in state $\reject$ and moves right, so $\sigma_0$ (the symbol under the head after one simulated timestep) is independent of what is written to the tape in the first step. Hence $A^k(x) = 0$ for all $x$.\footnote{Note that if we had an error of type $22$ \emph{and} an error of type $24$ then the machine would write an $A$ and have the head stay over it, to be read in the second step. This suggests some third-order derivative may be nonzero, in a way that reflects the structure of the machine.}

Finally when $k = 28$ we have to consider computation paths $\Lambda_j, \Omega, \Xi$. Since we start in state $\reject$ no error on path $\Lambda_j$ can cause an output error. For $\Omega$, there is no way for the machine under normal operation to reach state $\accept$ with the head over $B$, so again this cannot contribute an output error. Lastly, a match on the description tape is always found, so $\Xi$ cannot contribute an output error. Hence $A^k(x) = 0$ for all $x$ also in this case.
\end{proof}

Set
\[
J=\{1,\ldots,30\}\setminus\{3,13,18,23,24,25\}.
\]
These are the error coordinates marked in blue in Figure \ref{fig:detectA_left_coords}. Lemma \ref{lemma:irrelevant_entries_channel} says that $\detectA^{(0)}$ corrects every $J$-admissible error of weight one on $X$ when run for two steps. The corresponding face $W_J$ has dimension $d_J=24$, so Corollary \ref{corollary:bound_llc_TM} gives
\[
\lambda_J([\detectA^{(0)}],q)\le \frac{24}{2(1+1)}=6
\]
for a true input distribution supported on $X$ and a smooth positive prior density on $W_J$. This is a bound for the restricted model. The restriction is nontrivial: on input $A$, an error of type $12$ on $\Gamma_6$ together with an error of type $28$ on $\Omega$ causes the final $(B,\accept)$ tuple to write $\reject$. Both errors are $J$-admissible, so errors of weight two need not be corrected.

The reader is encouraged to work through the proof of Lemma \ref{lemma:irrelevant_entries_channel} in detail, since this example illustrates the way in which \emph{concretely} the behaviour of a Turing machine is reflected in combinatorics of error syndromes. It is quite intuitive why the variations that remain to be considered in Figure \ref{fig:detectA_left_coords} (that is, variations other than the ones marked in blue which are ruled out by the Lemma) should be the most critical parts of the code.

\begin{figure}[htbp]
    \centering
    $$
\begin{array}{|c|c?c|c|c|}
\hline
\sigma & q & \sigma' & q' & d \\
\hline
\Box & \reject & {\color{blue} \overset{1}{\Box \rightarrow A}}, {\color{blue} \overset{2}{\Box \rightarrow B}} & \overset{3}{\reject \rightarrow \accept} & {\color{blue} \overset{4}{\text{Stay} \rightarrow \text{Left}}}, {\color{blue} \overset{5}{\text{Stay} \rightarrow \text{Right}}} \\
\Box & \accept & {\color{blue} \overset{6}{\Box \rightarrow A}}, {\color{blue} \overset{7}{\Box \rightarrow B}} & {\color{blue} \overset{8}{\accept \rightarrow \reject}} & {\color{blue} \overset{9}{\text{Stay} \rightarrow \text{Left}}}, {\color{blue}\overset{10}{\text{Stay} \rightarrow \text{Right}}} \\
A & \reject & {\color{blue} \overset{11}{A \rightarrow \Box}}, {\color{blue} \overset{12}{A \rightarrow B}} & \overset{13}{\accept \rightarrow \reject} & {\color{blue} \overset{14}{\text{Stay} \rightarrow \text{Left}}}, {\color{blue} \overset{15}{\text{Stay} \rightarrow \text{Right}}}\\
A & \accept & {\color{blue} \overset{16}{A \rightarrow \Box}}, {\color{blue} \overset{17}{A \rightarrow B}} & \overset{18}{\accept \rightarrow \reject} & {\color{blue} \overset{19}{\text{Stay} \rightarrow \text{Left}}}, {\color{blue} \overset{20}{\text{Stay} \rightarrow \text{Right}}} \\
B & \reject & {\color{blue} \overset{21}{B \rightarrow \Box}}, {\color{blue} \overset{22}{B \rightarrow A}} & \overset{23}{\reject \rightarrow \accept} & \overset{24}{\text{Right} \rightarrow \text{Stay}}, \overset{25}{\text{Right} \rightarrow \text{Left}} \\
B & \accept & {\color{blue} \overset{26}{B \rightarrow \Box}}, {\color{blue} \overset{27}{B \rightarrow A}} & {\color{blue} \overset{28}{\accept \rightarrow \reject}} & {\color{blue} \overset{29}{\text{Stay} \rightarrow \text{Left}}}, {\color{blue} \overset{30}{\text{Stay} \rightarrow \text{Right}}}\\
\hline
\end{array}
$$ 
\caption{Local coordinates at $[\detectA^{(0)}]$ in the space of noisy Turing machines. The rows are indexed by $\sigma, q$ pairs (first two columns) and then for each of the three associated squares $\sigma', q, d$ on the description tape we enumerate the valid variations in the entry in that square away from its value in $\detectA^{(0)}$. These variations are exactly the valid values taken by an error syndrome on that square. The variations also provide the local coordinates at $[\detectA^{(0)}]$ in the space of noisy TM codes. Entries marked in blue do not contribute to error syndromes of weight one which can flip the output.}
\label{fig:detectA_left_coords}
\end{figure}

In what follows, we enumerate the tuples on the description tape following the rows of Figure \ref{fig:detectA_left_coords}. We sometimes view functions $X \rightarrow \mathbb{R}$ as vectors in the natural way, using the ordering in \eqref{eq:input_set_X}.

\begin{lemma}\label{lemma:detectA_k3} For $k = 3$ the relevant computation paths are $\Lambda_j$ for $1 \le j \le N$ and $\Omega$ and so the error syndromes with weight one can be described by tuples
\begin{align}
e_j &= \Big(\overset{\Lambda_1}{0}, \ldots, \overset{\Lambda_j}{1}, \ldots, \overset{\Lambda_N}{0}, \overset{\Omega}{0}, \overset{\Xi}{0}\Big)\,\qquad 1 \le j \le N \label{eq:lemma_detectA_k3ej}\\
e_{N+1} &= \Big(\overset{\Lambda_1}{0}, \ldots, \overset{\Lambda_j}{0}, \ldots, \overset{\Lambda_N}{0}, \overset{\Omega}{1}, \overset{\Xi}{0}\Big)\,\\
e_{N+2} &= \Big(\overset{\Lambda_1}{0}, \ldots, \overset{\Lambda_j}{0}, \ldots, \overset{\Lambda_N}{0}, \overset{\Omega}{0}, \overset{\Xi}{1}\Big)\,\label{eq:lemma_detectA_k3enplusone}
\end{align}
The values of $\mathcal{U}(x,\gamma)$ are given by Table \ref{table:values_k3} and the function $A^3$ is
\be
A^3 = \big( 0, 1, 0, 0, 0, 0 \big)
\ee
viewed as a vector in $\mathbb{R}^X$.
\end{lemma}
\begin{proof}
Recall $\mathcal{U}(x, \gamma)$ is the result of executing the UTM with error syndrome $\gamma$ on input $x$. When we write $\gamma = e_j$ in the context of the current lemma, we are implicitly considering the error syndrome that assigns $0$ to all computation paths \emph{other} than the ones starting at the position in the description tape indexed by $k = 3$, and just describing the values assigned by $\gamma$ to the computation paths starting at $k = 3$.

Note that since the head never starts over a blank symbol, $\gamma = e_j$ for $1 \le j \le N$ does not affect the output of the UTM. Hence all but the last two columns of Table \ref{table:values_k3} contain whatever output $M$ would give on that input. An error on the computation path $\Omega$ leads to error only when $x = B$ so that after the first step the machine has $\sigma_0 = \Box, q = \reject$ and then it erroneously transitions to $\accept$. That is,
\be
\mathcal{U}(x, e_{N+1}) = \begin{cases}
   M(x) & x \neq B\\
   \accept & x = B
\end{cases}
\ee
The computation path $\Xi$ is never taken for the inputs considered here, so $\gamma = e_8$ does not affect the output of the UTM.

This completes the calculation of Table \ref{table:values_k3} and the values of $A^k(x)$ follow by summing the errors in each row.
\end{proof}

\begin{table}[htbp]
  \centering
  \begin{tabular}{|c|c|c|c|c|c|c|c|c|}
    \hline
    \diagbox{$x$}{$\gamma$} & $e_1$ & $e_2$ & $e_3$ & $e_4$ & $e_5$ & $e_6$ & $e_7$ & $e_8$ \\
    \hline
    $A$ & $\tikzcmark$ & $\tikzcmark$ & $\tikzcmark$ & $\tikzcmark$ & $\tikzcmark$ & $\tikzcmark$ & $\tikzcmark$ & $\tikzcmark$ \\
    \hline
    $B$ & $\tikzxmark$ & $\tikzxmark$ & $\tikzxmark$ & $\tikzxmark$ & $\tikzxmark$ & $\tikzxmark$ & ${\color{red} \tikzcmark}$ & $\tikzxmark$ \\
    \hline
    $AB$ & $\tikzcmark$ & $\tikzcmark$ & $\tikzcmark$ & $\tikzcmark$ & $\tikzcmark$ & $\tikzcmark$ & $\tikzcmark$ & $\tikzcmark$\\
    \hline
    $BA$ & $\tikzcmark$ & $\tikzcmark$ & $\tikzcmark$ & $\tikzcmark$ & $\tikzcmark$ & $\tikzcmark$ & $\tikzcmark$ & $\tikzcmark$\\
    \hline
    $AA$ & $\tikzcmark$ & $\tikzcmark$ & $\tikzcmark$ & $\tikzcmark$ & $\tikzcmark$ & $\tikzcmark$ & $\tikzcmark$ & $\tikzcmark$\\
    \hline
    $BB$ & $\tikzxmark$ & $\tikzxmark$ & $\tikzxmark$ & $\tikzxmark$ & $\tikzxmark$ & $\tikzxmark$ & $\tikzxmark$ & $\tikzxmark$\\
    \hline
  \end{tabular}
  \caption{\label{table:values_k3} Values of $\mathcal{U}(x, \gamma)$ for different $x$ and $\gamma$ and $k = 3$. To save space we use $\protect \tikzcmark$ for $\accept$, $\protect \tikzxmark$ for $\reject$, and we colour an output red if this is incorrect, i.e. an error.}
\end{table}

\begin{lemma}\label{lemma:detectA_k13} For $k = 13$ the relevant computation paths are $\Lambda_j$ for $1 \le j \le N,\Omega, \Xi$. With the notation of \eqref{eq:lemma_detectA_k3ej}, \eqref{eq:lemma_detectA_k3enplusone} the function $A^{13}$ is
\be
A^{13} = \big( 0, 0, 0, 1, 0, 0 \big)
\ee
viewed as a vector in $\mathbb{R}^X$.
\end{lemma}
\begin{proof}
For $x \in \{B, BB\}$ none of these computation paths are used. If $x = BA$ then errors in the $\Lambda_j$ computation paths cannot affect the output, but
\[
\mathcal{U}(BA, e_7) = \reject
\]
is an output error induced by an error in $\Omega$. Otherwise the first symbol of $x$ is $A$, in which case errors on $\Omega$ cannot lead to output errors and we need only consider the computation paths $\Lambda_j$. The two relevant tuples on the description tape are $j \in \{3,4\}$.

We have $\mathcal{U}(A*, e_3) = \accept$ which is correct, even though an error occurred on the computation path $\Lambda_3$. In detail: if we use the modified tuple $(A, \reject, A, \reject, \text{Stay})$ in the first step we will have $\sigma_0 = A, q = \reject$ after one timestep, which the UTM then compares with $(A, \reject)$ on the description tape during the scan phase and finds a match. It then writes $\accept$ to the $q'$ part of the staging tape (following the unmodified form of the third tuple). Note that the UTM finds a \emph{second} match with $(A, \accept)$ and writes $\accept$ to the staging tape for a second time.

In contrast, we have $\mathcal{U}(A*, e_4) = \accept$. When checking $A, \accept$ for a match on the description tape with $\sigma_0 = A, q = \reject$ does not find a match. Hence the UTM does not match to \emph{any} tuple and at the end of the scan phase $X$ remains on the staging tape, so it leaves the TM state unchanged with $\accept$.

When $\gamma = e_8$, the computation path $\Xi$ is never taken for the inputs considered here, so this error does not affect the output of the UTM. These calculations are summarised in Table \ref{table:values_k13}.
\end{proof}

\begin{table}[htbp]
  \centering
  \begin{tabular}{|c|c|c|c|c|c|c|c|c|}
    \hline
    \diagbox{$x$}{$\gamma$} & $e_1$ & $e_2$ & $e_3$ & $e_4$ & $e_5$ & $e_6$ & $e_7$ & $e_8$\\
    \hline
    $A$ & $\tikzcmark$ & $\tikzcmark$ & $\tikzcmark$ & $\tikzcmark$ & $\tikzcmark$ & $\tikzcmark$ & $\tikzcmark$ & $\tikzcmark$ \\
    \hline
    $B$ & $\tikzxmark$ & $\tikzxmark$ & $\tikzxmark$ & $\tikzxmark$ & $\tikzxmark$ & $\tikzxmark$ & $\tikzxmark$ & $\tikzxmark$ \\
    \hline
    $AB$ & $\tikzcmark$ & $\tikzcmark$ & $\tikzcmark$ & $\tikzcmark$ & $\tikzcmark$ & $\tikzcmark$ & $\tikzcmark$ & $\tikzcmark$\\
    \hline
    $BA$ & $\tikzcmark$ & $\tikzcmark$ & $\tikzcmark$ & $\tikzcmark$ & $\tikzcmark$ & $\tikzcmark$ & ${\color{red} \tikzxmark}$ & $\tikzcmark$\\
    \hline
    $AA$ & $\tikzcmark$ & $\tikzcmark$ & $\tikzcmark$ & $\tikzcmark$ & $\tikzcmark$ & $\tikzcmark$ & $\tikzcmark$ & $\tikzcmark$\\
    \hline
    $BB$ & $\tikzxmark$ & $\tikzxmark$ & $\tikzxmark$ & $\tikzxmark$ & $\tikzxmark$ & $\tikzxmark$ & $\tikzxmark$ & $\tikzxmark$\\
    \hline
  \end{tabular}
  \caption{\label{table:values_k13} Values of $\mathcal{U}(x, \gamma)$ for different $x$ and $\gamma$ and $k = 13$.}
\end{table}

\begin{lemma}\label{lemma:detectA_k18} For $k = 18$ the relevant computation paths are $\Lambda_j$ for $1 \le j \le N, \Omega$ and $\Xi$. With the notation of \eqref{eq:lemma_detectA_k3ej}, \eqref{eq:lemma_detectA_k3enplusone} the function $A^{18}$ is
\be
A^{18} = \big( 1, 0, 1, 0, 1, 0 \big)
\ee
viewed as a vector in $\mathbb{R}^X$.
\end{lemma}
\begin{proof}
Again there can be no output errors when $x \in \{ B, BB, BA \}$. However when $x$ begins with $A$ and we have the error of type $k = 18$ on computation path $\Omega$ we have $\mathcal{U}(A*, e_7) = \reject$ which is an output error. Lastly, the computation path $\Xi$ is never taken. These results are in Table \ref{table:values_k18}.
\end{proof}

\begin{table}[htbp]
  \centering
  \begin{tabular}{|c|c|c|c|c|c|c|c|c|}
    \hline
    \diagbox{$x$}{$\gamma$} & $e_1$ & $e_2$ & $e_3$ & $e_4$ & $e_5$ & $e_6$ & $e_7$ & $e_8$ \\
    \hline
    $A$ & $\tikzcmark$ & $\tikzcmark$ & $\tikzcmark$ & $\tikzcmark$ & $\tikzcmark$ & $\tikzcmark$ & ${\color{red} \tikzxmark}$ & $\tikzcmark$ \\
    \hline
    $B$ & $\tikzxmark$ & $\tikzxmark$ & $\tikzxmark$ & $\tikzxmark$ & $\tikzxmark$ & $\tikzxmark$ & $\tikzxmark$ & $\tikzxmark$ \\
    \hline
    $AB$ & $\tikzcmark$ & $\tikzcmark$ & $\tikzcmark$ & $\tikzcmark$ & $\tikzcmark$ & $\tikzcmark$ & ${\color{red} \tikzxmark}$ & $\tikzcmark$ \\
    \hline
    $BA$ & $\tikzcmark$ & $\tikzcmark$ & $\tikzcmark$ & $\tikzcmark$ & $\tikzcmark$ & $\tikzcmark$ & $\tikzcmark$ & $\tikzcmark$\\
    \hline
    $AA$ & $\tikzcmark$ & $\tikzcmark$ & $\tikzcmark$ & $\tikzcmark$ & $\tikzcmark$ & $\tikzcmark$ & ${\color{red} \tikzxmark}$ & $\tikzcmark$\\
    \hline
    $BB$ & $\tikzxmark$ & $\tikzxmark$ & $\tikzxmark$ & $\tikzxmark$ & $\tikzxmark$ & $\tikzxmark$ & $\tikzxmark$ & $\tikzxmark$\\
    \hline
  \end{tabular}
  \caption{\label{table:values_k18} Values of $\mathcal{U}(x, \gamma)$ for different $x$ and $\gamma$ and $k = 18$.}
\end{table}

\begin{lemma}\label{lemma:detectA_k23} For $k = 23$ the relevant computation paths are $\Lambda_j$ for $1 \le j \le N,\Omega$ and $\Xi$. With the notation of \eqref{eq:lemma_detectA_k3ej}, \eqref{eq:lemma_detectA_k3enplusone} the function $A^{23}$ is
\be
A^{23} = \big( 0, 1, 0, 1, 0, 2 \big)
\ee
viewed as a vector in $\mathbb{R}^X$.
\end{lemma}
\begin{proof}
Again there can be no output errors when $x \in \{ A, AB, AA \}$. The behaviour of $\mathcal{U}$ on inputs beginning with $B$ is quite complex. Firstly if $x \in \{B, BA\}$ we have $\mathcal{U}(x, e_7) = M(x)$. But $\mathcal{U}(BB, e_7) = \accept$. For the $\Lambda_j$ computation paths, we can ignore most columns depending on the second symbol in the input.
\begin{align*}
\mathcal{U}(B, e_1) &= \reject\,, \quad \mathcal{U}(B, e_2) = \accept\\ %
\mathcal{U}(BA, e_3) &= \reject \,, \quad \mathcal{U}(BA, e_4) = \accept\\ %
\mathcal{U}(BB, e_5) &= \reject \,, \quad \mathcal{U}(BB, e_6) = \accept\,. %
\end{align*}
An error on the computation path $\Xi$ alone does not affect the output. These calculations are summarised in Table \ref{table:values_k23}.
\end{proof}

\begin{table}[htbp]
  \centering
  \begin{tabular}{|c|c|c|c|c|c|c|c|c|}
    \hline
    \diagbox{$x$}{$\gamma$} & $e_1$ & $e_2$ & $e_3$ & $e_4$ & $e_5$ & $e_6$ & $e_7$ & $e_8$ \\
    \hline
    $A$ & $\tikzcmark$ & $\tikzcmark$ & $\tikzcmark$ & $\tikzcmark$ & $\tikzcmark$ & $\tikzcmark$ & $\tikzcmark$ & $\tikzcmark$ \\
    \hline
    $B$ & $\tikzxmark$ & ${\color{red} \tikzcmark}$ & $\tikzxmark$ & $\tikzxmark$ & $\tikzxmark$ & $\tikzxmark$ & $\tikzxmark$ & $\tikzxmark$ \\
    \hline
    $AB$ & $\tikzcmark$ & $\tikzcmark$ & $\tikzcmark$ & $\tikzcmark$ & $\tikzcmark$ & $\tikzcmark$ & $\tikzcmark$ & $\tikzcmark$\\
    \hline
    $BA$ & $\tikzcmark$ & $\tikzcmark$ & ${\color{red} \tikzxmark}$ & $\tikzcmark$ & $\tikzcmark$ & $\tikzcmark$ & $\tikzcmark$ & $\tikzcmark$\\
    \hline
    $AA$ & $\tikzcmark$ & $\tikzcmark$ & $\tikzcmark$ & $\tikzcmark$ & $\tikzcmark$ & $\tikzcmark$ & $\tikzcmark$ & $\tikzcmark$\\
    \hline
    $BB$ & $\tikzxmark$ & $\tikzxmark$ & $\tikzxmark$ & $\tikzxmark$ & $\tikzxmark$ & ${\color{red} \tikzcmark}$ & ${\color{red} \tikzcmark}$ & $\tikzxmark$\\
    \hline
  \end{tabular}
  \caption{\label{table:values_k23} Values of $\mathcal{U}(x, \gamma)$ for different $x$ and $\gamma$ and $k = 23$.}
\end{table}

\begin{lemma}\label{lemma:detectA_k24} For $k = 24$ the relevant computation paths are $\Theta_j$ for $1 \le j \le N$ so the error syndromes with weight one can be described by tuples
\be\label{eq:error_syndrome_k24_weightone}
e_j = \Big(\overset{\Theta_1}{0}, \ldots, \overset{\Theta_j}{1}, \ldots, \overset{\Theta_N}{0} \Big)\,\qquad 1 \le j \le N\,.
\ee
The values of $\mathcal{U}(x,\gamma)$ are given in Table \ref{table:values_k24}
\be
A^{24} = \big( 0, 0, 0, 2, 0, 0 \big)
\ee
viewed as a vector in $\mathbb{R}^X$.
\end{lemma}
\begin{proof}
For $x \notin \{B, BA, BB\}$ we have $\mathcal{U}(x, e_j) = M(x)$ for all $j$. For any $x$ that begins with $B$, the effect of introducing a $k = 24$ type error is to stay rather than move right in the first step, so that $\sigma_0 = B$ after one step. Thus $\mathcal{U}(BB, e_j) = M(BB) = \reject$ for all $j$, as this is any case the symbol under the head under normal execution. We can ignore $j \in \{2, 4, 6\}$ since the machine is in state $\reject$ after one step. The remaining cases are
\begin{align*}
\mathcal{U}(B, e_1) &= \reject\,, \quad \mathcal{U}(B, e_3) = \reject\,, \quad %
\mathcal{U}(B, e_5) = \reject\\
\mathcal{U}(BA, e_1) &= \accept \,, \quad \mathcal{U}(BA, e_3) = \reject\,, \quad \mathcal{U}(BA, e_5) = \reject %
\end{align*}
These calculations are summarised in Table \ref{table:values_k24}.
\end{proof}

\begin{table}[htbp]
  \centering
  \begin{tabular}{|c|c|c|c|c|c|c|}
    \hline
    \diagbox{$x$}{$\gamma$} & $e_1$ & $e_2$ & $e_3$ & $e_4$ & $e_5$ & $e_6$ \\
    \hline
    $A$ & $\tikzcmark$ & $\tikzcmark$ & $\tikzcmark$ & $\tikzcmark$ & $\tikzcmark$ & $\tikzcmark$ \\
    \hline
    $B$ & $\tikzxmark$ & $\tikzxmark$ & $\tikzxmark$ & $\tikzxmark$ & $\tikzxmark$ & $\tikzxmark$ \\
    \hline
    $AB$ & $\tikzcmark$ & $\tikzcmark$ & $\tikzcmark$ & $\tikzcmark$ & $\tikzcmark$ & $\tikzcmark$ \\
    \hline
    $BA$ & $\tikzcmark$ & $\tikzcmark$ & ${\color{red} \tikzxmark}$ & $\tikzcmark$ & ${\color{red} \tikzxmark}$ & $\tikzcmark$ \\
    \hline
    $AA$ & $\tikzcmark$ & $\tikzcmark$ & $\tikzcmark$ & $\tikzcmark$ & $\tikzcmark$ & $\tikzcmark$ \\
    \hline
    $BB$ & $\tikzxmark$ & $\tikzxmark$ & $\tikzxmark$ & $\tikzxmark$ & $\tikzxmark$ & $\tikzxmark$ \\
    \hline
  \end{tabular}
  \caption{\label{table:values_k24} Values of $\mathcal{U}(x, \gamma)$ for different $x$ and $\gamma$ and $k = 24$.}
\end{table}

\begin{lemma}\label{lemma:detectA_k25} For $k = 25$ the relevant computation paths are $\Theta_j$ for $1 \le j \le N$ so the error syndromes with weight one can be described by tuples as in \eqref{eq:error_syndrome_k24_weightone}. Then
\be
A^{25} = \big( 0, 0, 0, 1, 0, 0 \big)
\ee
viewed as a vector in $\mathbb{R}^X$.
\end{lemma}
\begin{proof}
For $x \notin \{B, BA, BB\}$ we have $\mathcal{U}(x, e_j) = M(x)$ for all $j$. For any $x$ beginning with $B$ introducing a $k = 25$ type error causes the head to move left rather than right in the first timestep, so $\mathcal{U}(B, e_j) = M(B) = \reject$ for all $j$ since the symbol to the left of the first symbol on the work tape is a blank. It is also easy to see that $\mathcal{U}(BB, e_j) = M(BB) = \reject$ for all $j$. When $x = BA$ we have
\[
\mathcal{U}(BA, e_2) = \mathcal{U}(BA, e_4) = \mathcal{U}(BA, e_6) = M(BA) = \accept
\]
since these tuples never match to the state $\reject$ after one timestep. Finally
\[
\mathcal{U}(BA, e_1) = \accept\footnote{Note that here the UTM executes two tuples.}\,, \quad \mathcal{U}(BA, e_3) = \reject\,, \quad \mathcal{U}(BA, e_5) = \accept\,. %
\]
These calculations are summarised in Table \ref{table:values_k25}.
\end{proof}

\begin{table}[htbp]
  \centering
  \begin{tabular}{|c|c|c|c|c|c|c|}
    \hline
    \diagbox{$x$}{$\gamma$} & $e_1$ & $e_2$ & $e_3$ & $e_4$ & $e_5$ & $e_6$ \\
    \hline
    $A$ & $\tikzcmark$ & $\tikzcmark$ & $\tikzcmark$ & $\tikzcmark$ & $\tikzcmark$ & $\tikzcmark$ \\
    \hline
    $B$ & $\tikzxmark$ & $\tikzxmark$ & $\tikzxmark$ & $\tikzxmark$ & $\tikzxmark$ & $\tikzxmark$ \\
    \hline
    $AB$ & $\tikzcmark$ & $\tikzcmark$ & $\tikzcmark$ & $\tikzcmark$ & $\tikzcmark$ & $\tikzcmark$ \\
    \hline
    $BA$ & $\tikzcmark$ & $\tikzcmark$ & ${\color{red} \tikzxmark}$ & $\tikzcmark$ & $\tikzcmark$ & $\tikzcmark$ \\
    \hline
    $AA$ & $\tikzcmark$ & $\tikzcmark$ & $\tikzcmark$ & $\tikzcmark$ & $\tikzcmark$ & $\tikzcmark$ \\
    \hline
    $BB$ & $\tikzxmark$ & $\tikzxmark$ & $\tikzxmark$ & $\tikzxmark$ & $\tikzxmark$ & $\tikzxmark$ \\
    \hline
  \end{tabular}
  \caption{\label{table:values_k25} Values of $\mathcal{U}(x, \gamma)$ for different $x$ and $\gamma$ and $k = 25$.}
\end{table}

Given the above we compute the functions $\{ g_3, g_{13}, g_{18}, g_{23}, g_{24}, g_{25} \}$ and then the matrix $P$ of Section \ref{section:geometry} is
\be\label{eq:matrix_P_detectA0}
\begin{array}{c|cccccc|c}
 & 3 & 13 & 18 & 23 & 24 & 25 & \text{(24 zero columns)} \\
\hline
 A & 0 & 0 & 1 & 0 & 0 & 0 & \\
 B & 1 & 0 & 0 & 1 & 0 & 0 & \\
AB & 0 & 0 & 1 & 0 & 0 & 0 & \\
BA & 0 & 1 & 0 & 1 & 2 & 1 & \\
AA & 0 & 0 & 1 & 0 & 0 & 0 & \\
BB & 0 & 0 & 0 & 2 & 0 & 0 & \\
\end{array}
\ee
If we assume all the inputs have equal probability $\tfrac{1}{6}$ then by \eqref{eq:hessian_formula_P} we have
\be
\operatorname{Hess}(H, [\detectA^{(0)}]) = \tfrac{1}{3} P^T P = \tfrac{1}{3} \begin{pmatrix} P' & 0 \\ 0 & 0 \end{pmatrix}
\ee
where (ordering the variables as in \eqref{eq:matrix_P_detectA0})
\be
P' = \begin{pmatrix}
1 & 0 & 0 & 1 & 0 & 0 \\
0 & 1 & 0 & 1 & 2 & 1 \\
0 & 0 & 3 & 0 & 0 & 0 \\
1 & 1 & 0 & 6 & 2 & 1 \\
0 & 2 & 0 & 2 & 4 & 2 \\
0 & 1 & 0 & 1 & 2 & 1 \end{pmatrix}\,.
\ee
From this we see that $e_{18}$ is an eigenvector of eigenvalue $\lambda_1 = 3$, and that
\[
P'( e_{24} - 2 e_{13} ) = P'( e_{25} - e_{13} ) = 0\,.
\]
Hence zero is an eigenvalue of multiplicity two. Indeed the characteristic polynomial is
\[
p(\lambda) = \lambda^2( \lambda - 3 )( \lambda^3 - 13 \lambda^2 + 41 \lambda - 24 )\,.
\]
Let $\lambda_2, \lambda_3, \lambda_4$ denote the roots of the cubic factor, which are approximately $0.76, 3.73, 8.52$. Then in local coordinates $v_1,\ldots,v_{30}$ at $[\detectA^{(0)}]$ we have
\be\label{eq:decomp_H_detect}
H = v_1^2 + \sum_{i=2}^4 \tfrac{1}{3} \lambda_i v_i^2 + O(v^3)\,.
\ee
Remarkably, there are \emph{only four out of thirty} directions in the space of noisy TMs in which the function $H$ is nondegenerate. The majority of the possible weight one errors do not affect the performance of the machine (on inputs in $X$ when run for two steps).

Note that the behaviour of the machine when it is in state $\accept$ and reading an $A$ (which is the only way to be in state $\accept$) is critical in our two step machine; introducing an error in the state part of this tuple is therefore one of the \emph{most obvious ways to cause an output error} and we would therefore expect $H$ to be strongly curved in this direction. This is what the geometric analysis shows: the direction $v_1$ corresponding to $e_{18}$ is one of the natural coordinates associated to the eigenvectors of the Hessian of $H$, and it has a large positive eigenvalue. Thus we see a direct relationship between the local geometry of $H$ at $[\detectA^{(0)}]$ and the structure of the Turing machine.

\begin{remark} Recall that it is the geometry of the average negative log-likelihood $L(w)$ that we are ultimately interested in, and the polynomial $H(w)$ is just a convenient way to access (some of) that geometry, since $H$ is comparable to $K = L-L([\detectA^{(0)}])$, which has the same Hessian as $L$. Thus the above observations about $H$ in local coordinates should not be overinterpreted. However, it \emph{is} true that \eqref{eq:decomp_H_detect} implies that the Hessian of $L$ also has rank $4$ at $[\detectA^{(0)}]$ and moreover, we can relate their eigenvalues to some extent (see Appendix \ref{section:compare_functions}).
\end{remark}

\begin{remark}\label{remark:kernel_hessian_example} It is clear by examining $P'$ that a basis for the kernel is $-2e_{13} + e_{24}, -e_{13} + e_{25}$. This gives two additional zero curvature directions in addition to the $24$ identified in Lemma \ref{lemma:irrelevant_entries_channel}. By Lemma \ref{lemma:comparable_hessian} applied to $H,K$, these are also directions in the kernel of the Hessian of $K$, hence of $L$, that is, they are genuinely part of the geometry of the statistical learning problem.

In the notation of Figure \ref{fig:detectA_left_coords} these directions are
\begin{gather*}
-2\big( \overset{13}{\accept \rightarrow \reject} ) + \overset{24}{\text{Right} \rightarrow \text{Stay}},\\
-1 \big( \overset{13}{\accept \rightarrow \reject} ) + \overset{25}{\text{Right} \rightarrow \text{Left}}\,.
\end{gather*}
Note that at the boundary point $[\detectA^{(0)}]$ of $W$ these tangent vectors point \emph{outside} of the space of noisy TM codes. Nonetheless they indicate an interesting and non-obvious way to trade-off these two tuples, reflecting non-trivial structure of the machine.
\end{remark}

\section{Conclusion}\label{section:conclusion}

One of the central lessons of logic and computer science for mathematics is to pay attention to the distinction between an \emph{algorithm} and the \emph{function} it computes. In statistical learning theory we often treat models as functions which assign probabilities to data, and ignore how those functions are implemented. In this section we explain how, within the framework of SLT, there is a natural sense in which the internal structure of \emph{probabilistic models as algorithms} play an important role in Bayesian statistics.

To compare the evidence for two candidate models $p(y|x,w^1), p(y|x,w^2)$ of a conditional distribution $q(y|x)$ with respect to a dataset $D_n$, we integrate the posterior over neighbourhoods $\mathcal{W}_i$ of $w^i$. If for simplicity we assume that $w^1, w^2$ are both true parameters, then we know from SLT that asymptotically in $n$ this comparison is determined by the local geometry of the KL divergence $K(w)$ at each $w^i$ (see Section \ref{section:free_energy}). This geometry is a reflection of \emph{the pattern of variation of the probabilities} as $w$ is varied near $w^i$
\be\label{eq:local_variation_p}
p(y|x,w) \approx p(y|x,w^i) \text{ for } w \approx w^i\,.
\ee
In many cases, such as when the probabilistic model $p(y|x,w)$ is parametrised by a neural network or as here by a noisy TM, variations in the parameter $w$ are localised microscopic changes in the algorithm that computes these probabilities (e.g. weights associated to connections between specific neurons, or probabilities of specific errors in the UTM reading the Turing machine description). Thus the geometry of $K(w)$, and consequently the preferences of the posterior, are related to how such variations in the algorithm affect (or do \emph{not} affect) the output probabilities. If the probabilities are computed by algorithms which have hidden variables, hierarchical structure or are composed of multiple information processing modules, then we expect these statistical models to be singular and that ``the knowledge or grammar to be discovered corresponds to singularities'' \cite[\S 1.2]{watanabeAlgebraicGeometryStatistical2009}.
\\

To make this more precise: some changes in the algorithm underlying our model are either just irrelevant to the outputs, or can be compensated by coordinated changes in other parts, leading to changes in $w$ that do not change the output probabilities. These totally flat directions in the geometry of $K(w)$, or \emph{degeneracies}, are signatures of the possible trade-offs among parts of the algorithm, and in turn reflect the internal structure. More refined information is obtained by considering coordinated changes in the algorithm that do not change the probabilities \emph{up to some order} in a small perturbation parameter $\varepsilon$ measuring how far we deviate from the original algorithm.

In the framework developed in this paper we have studied concrete examples of both kinds of changes: in our study of the Turing machine $\detectA^{(0)}$ we found that changes in $24$ out of $30$ directions at this code are irrelevant to outputs to second order, and in Remark \ref{remark:kernel_hessian_example} we exhibit two linearly independent coordinated changes in the code that together do not change the output probabilities up to second order.

\subsection{Structural Bayesianism}\label{sec:structural_bayesianism}

All this leads to the hypothesis that the internal computational structure of Turing machines, and potentially other singular models such as neural networks, is encoded in the geometry of the KL divergence and may be discoverable from data \cite{lau2024locallearningcoefficientsingularityaware,hoogland2024developmental,wang2024differentiation}. 

We view this as a strong motivation for taking \emph{interpretability} of singular models, including neural networks, seriously as a mathematical science. This philosophy, which we refer to as \emph{structural Bayesianism}, says that the model selection principles at the heart of Bayesian statistics dictate in many cases that the posterior can ``see'' the internal structure of the algorithms computing probabilities in the model, and consequently the posterior can prefer some kinds of algorithmic structures over others. That is, the posterior concentrates probability mass on programs that not only predict \emph{well} but predict \emph{properly}. 

\begingroup
\let\originalthebibliography\thebibliography
\renewcommand{\thebibliography}[1]{\originalthebibliography{#1}\setlength{\itemsep}{2pt}}
\bibliographystyle{plain}
\bibliography{bibliography.bib}
\endgroup

\appendix

\section{Frequently Asked Questions}\label{section:faq}

So, you have questions.

\subsection{Does the geometry depend on the encoding?}

We have described how to associate the germ of a singularity to a Turing machine $M$, or more precisely to the code of this machine $[M]$ for a particular pseudo-UTM. There are many choices involved in producing this geometric object, including: the choice of UTM, the choice of encoding of the step function of that UTM into linear logic, and the choice to use linear logic at all to construct the smooth extension. In this section we address some of the natural questions about these choices.

\subsubsection{Choice of UTM}

There is no question of this geometry being \emph{canonically} associated to $M$, and indeed it is our opinion that it is probably unreasonable to expect any such canonical association: at some conceptual level the geometry of $L$ (or $H$) at $[M]$ is determined by the particular way we choose to introduce errors into the execution, and this in turn is determined by \emph{how} we think about executing that code. Different UTMs will lead to different geometries. In our understanding this is a reflection of the fact that the notion of ``internal structure'' in a Turing machine is not intrinsic to the tuples themselves, but to the dynamical system that represents their execution (that is, what a UTM does with that code for the TM).

The reader should recall that there are UTMs other than the ``simulation type'' UTM. For example, given any particular TM we can augment a simulation-type UTM $\mathcal{U}^{\text{sim}}$ with special subroutines that recognise some part of the TM and consult a look-up table encoded in the UTM's specification in order to skip an arbitrary amount of computation and directly return the answer. Let us call this $\mathcal{U}^{\text{cheat}}$. Any structure internal to the part of the TM $M$ that is skipped will not be visible in the geometry for $[M]$ on $\mathcal{U}^{\text{cheat}}$.

Even if we stick to simulation-type UTMs, there are many variations. Some, like our staged pseudo-UTM, are convenient for doing calculations but may strike the reader as unnatural in various ways. We could hypothesise that some aspects of the geometry are independent of these variations to some degree, just like Kolmogorov complexity is independent of the chosen UTM to some degree, but there is no present theory which makes such claims and it is unclear what statement one could even hope for.

\subsubsection{Choice of Linear Logic}

We choose to compute the propagation of uncertainty about Turing machine codes through the execution of a UTM by encoding the step function of the UTM into linear logic and using the Sweedler semantics \cite{murfet2015sweedler,clift2020encodings,clift2018derivatives}. This is somewhat roundabout and the reader is forgiven for wondering if ``all this linear logic stuff is really necessary''. The answer is No! 

There is a (perhaps kinder) version of this paper which eschews linear logic completely, presents the computation paths of Appendix \ref{section:comp_path_utm} as a \emph{fait accompli} and presents a model of noisy execution of TM codes in those terms and derives geometry in that way. The reader is encouraged to get their hands dirty with the calculations in Section \ref{section:full_geometry} to see how the associated mode of thought engages with computation paths and does not depend on linear logic at all. This is an entirely valid way of approaching the theory.
\\

Nonetheless, our opinion is that while this presentation is easier to approach it is ultimately shallower, as it obscures the influence of Ehrhard-Regnier's idea of derivatives of proofs. Note that by Theorem \ref{theorem:main_error_syndrome_general} the geometry considered here is \emph{an object of proof theory in its own right} (modulo the business with minus signs and the antipode). Moreover in that version of the paper, dear reader, you are asking where this notion of computation path comes from. There is some elegance (in the Kolmogorov complexity sense) of simply saying ``encode the UTM in linear logic'' and then following that to its logical conclusion, but we accept that this elegance is only visible if you know linear logic (there's that UTM dependence of $K$). From a practical point of view one can think of linear logic as a useful term language to describe the algorithm for executing the UTM.

\subsubsection{Choice of Encoding}

Once one has fixed a UTM and decided to use linear logic to construct a smooth relaxation of its step function, there are still multiple encodings into linear logic \cite{clift2020encodings}. Variations are possible which would conceivably change the geometry of the final germ $([M], L)$. Once again, we think of these encodings as a compact \emph{specification} of the computation paths which in turn determine how error can affect the simulation of the machine. We have provided one specification that we find natural, but we have no ground to assert that there are no other natural specifications. The claim of this paper is to construct \emph{an example} of a structure preserving map from Turing machines to singularities. There is no claim, explicitly or implicitly, that this homomorphism is \emph{unique}.

\subsection{Is everything in the geometry computationally meaningful?}

Since every coefficient in the Taylor series expansion of $H$ has been given a computational interpretation in terms of correlations of error syndrome statistics in Theorem \ref{theorem:partial_formula_explicit}, it is possible to answer Yes in a formal sense provided we restrict this statement to the aspects of the geometry of $H$ that are shared with $L$ (Section \ref{sec:geom_comparable}).

Our present intuition is that this answer seems overly strong, but since we don't have a clear definition of what ``counts'' as computational structure of a Turing machine it seems impossible to give a definitive answer. On the other hand, it seems plausible that geometry of the singularity $([M], L)$ that initially seems very remote from computer science may eventually be seen to be intimately connected to basic concepts in that field.

\subsection{Why is there a uniform number of time steps $t$}\label{section:faq_timesteps}

In our conceptual framing in terms of synthesis problems and statistical learning, we are interested in two classical solutions $M_1, M_2$ and how the Bayesian posterior distinguishes them. Let us suppose that for all $x \in I$ not only do $M_1(x) = y(x), M_2(x) = y(x)$ in the sense of the main text, but the machines both \emph{halt} (i.e. enter some static state).

Then, given a dataset size $n$ the question is: does the posterior assign higher probability to a neighbourhood of $M_1$ or $M_2$? Since the dataset is finite and all inputs are finite strings, and \emph{since $M_1, M_2$ are assumed to be classical solutions} and to halt, it may seem that there is some $t$ such that comparing them using this number of steps is sufficient. This is almost correct, but note that the local geometry is sensitive to execution traces of the machines that involve errors, and such traces \emph{may fail to halt}. Thus the geometry can be sensitive to arbitrary numbers of timesteps even if the machines would usually halt on all allowed inputs within a given timeout.

More general synthesis problems are considered in \cite{clift2018derivatives,waring2021geometric} but this introduces additional complexity and in this paper we follow \cite{clift2021geometry} in making the simplification of fixing $t$. We admit this as a limitation of the current presentation. It would be interesting if there was an argument to the effect that non-stably halting machines are exponentially suppressed in the posterior, which would provide a more principled reason for fixing $t$.

\subsection{Why is the input distribution given finite support?}

For a fixed finite number of timesteps, only finitely many input patterns on the inspected tape window can affect the computation. We assume finite support for simplicity. Note that finite numbers can be big: take $t \gg 0$ and $I \supseteq \Sigma^k$ for all $k \le K$ for $K \gg 0$.

\subsection{How should I think about this model of noisy computation?}\label{section:think_noisy}

We refer the reader to Section \ref{appendix:stagedpseudoutm} for a description of the staged pseudo-UTM, herein just ``UTM''. The remarks in this section expand on those in Section \ref{section:overview_ind_inf}.

Recall that the squares of the description tape are of the form $\sigma, q, \sigma', q', d$. To compute the probabilities $p(y|x,w)$ as in Section \ref{section:background_iid} the string $x$ is placed on the work tape of the UTM, $w$ is arranged in the squares of type $\sigma', q', d$ on the description tape and then the (smooth relaxation of the) UTM is executed for $t$ cycles ($tP$ UTM steps) before reading off the distribution over symbols $y$ on the state tape (meaning the tape which records the state of the simulated machine).

Note that we do not allow uncertainty in the inputs $x$. There is a distribution $q(x)$ over inputs, but this is irrelevant: the UTM only ever begins execution with a definite string $x$ on its work tape. We of course allow uncertainty on the squares of the description tape coming from $w$ (there is never uncertainty about the squares of type $\sigma, q$). While there is initially no uncertainty on the work, state or staging tapes, the uncertainty on the description tape propagates to these tapes as well. There may also be some uncertainty in the state of the UTM (e.g. between \text{copySymbol} and $\neg\,$\text{copySymbol}) but the design of the staged pseudo-UTM mean that this uncertainty is ``resolved'' at the end of every cycle.
\\

The usphot of using the encoding of the step function of the UTM into linear logic to propagate uncertainty is that, following the logic of Section \ref{section:eval_prog_vec}, \emph{whenever the UTM wants to read from a square on the description, work or staging tapes with uncertainty, it draws a sample}. But there is a catch: these samples must all be \emph{independent}. What that means is that the UTM must recurse backwards to sample, in the previous step, all tape squares necessary to produce the tape square it originally wanted to read. These ``requests for samples'' propagate backwards through all previous steps until they cash out in (many) requests for samples from the original sources of uncertainty on the description tape. 

These ``chains of requests'' from the final output state $y$ to the squares on the description tape are precisely the computation paths (Appendix \ref{section:comp_path_utm}). Reading them forwards, these computation paths specify how one of a number of independent samples from some error distribution on the description tape is to be routed through the execution of the UTM to be ``consumed''. This backwards chaining of requests is made precise by the cut-elimination rules for promotion and contraction, and this explains why linear logic is a particularly natural language to reason about this kind of noisy computation model.

\subsection{Why use TMs? Why not proofs or lambda terms?}

The parameter space $W$ is part of the space of denotations of a linear logic formula, which has as proofs configurations of the description tape of a UTM. For a linear logician, it might seem natural to instead consider general proof synthesis for a general type $A$ and $W \subseteq \den{A}$ relative to some proof $\psi: {!} A, {!} C \vdash B$ which plays the role of the (iterated) step function of the UTM in the present paper with some function $y: {!} C \vdash B$ as the analogue of the synthesis problem. And indeed, this geometrisation of proof synthesis is what \cite{clift2020encodings,clift2018derivatives} were originally setting out to do!

There are however two problems: firstly, there is the technical problem of the propagation of uncertainty through non-plain proofs \cite{clift2020encodings}. It is not currently known if this holds for proofs $\psi$ which are not plain. Putting that aside, the above theory could be constructed exactly as in this paper for any plain proof $\psi$, but then the question becomes: what are the natural examples we can actually compute with? Turing machine codes provide a natural ``flat'' class of examples that are short, can be constructed by hand, and which have a clear conceptual content.

\subsection{Why not a real UTM?}\label{section:faq_realutm}

The staged pseudo-UTM $\mathcal{U}$ is not universal, because it assumes the alphabet symbols and states of the simulated machine can be encoded as single symbols and states of the UTM. That is, $\Sigma \cup Q \subseteq \Sigma_{\text{UTM}}$. In a similar spirit to Appendix \ref{section:faq_timesteps}, note that by enlarging $\Sigma, Q$ we can in this way accommodate any \emph{finite} set of Turing machines as codes on our pseudo-UTM and this suffices for our purposes.

\subsection{Why not smooth relaxation $\mathcal{X}$?}

There are many smooth relaxations of a UTM in the sense of Section \ref{section:background_iid}. Maybe you have one in mind. Any such smooth relaxation gives rise to a KL divergence $K$ whose geometry is relevant to the learning process defined with respect to that smooth relaxation. For many such relaxations there will be a relation between the local learning coefficient and invariants like program length, as described in \cite{clift2021geometry,waring2021geometric} because un-used bits will often correspond to degeneracy. 

However, generically a smooth relaxation $\mathcal{X}$ will have no deeper relation between the geometry of $K: W \lto \mathbb{R}$ and the internal structure of the machine being simulated $[M] \in W$, because there will be no \emph{precise} and general relation between perturbations away from $[M]$ within $W$ and computationally meaningful perturbations of the execution of $M$ as a dynamical system.

In our opinion, the nontrivial contribution being made in the series of papers \cite{clift2018derivatives,clift2021geometry,waring2021geometric} and continued here is to construct a smooth relaxation where such a precise relation \emph{does} exist. One way of observing this is that the geometry of $K$ is here related to the derivatives of $H$ which are, by Theorem \ref{theorem:main_error_syndrome_general}, denotations of \emph{proofs in differential linear logic}. In some sense the geometry of the synthesis problem is an object of proof theory. Arguably the degeneracies observed in Section \ref{section:full_geometry} arise from patterns that can also be expressed in the syntax. These are strong statements about the logical status of the geometry of $K$ which are not true of general smooth relaxations.

\section{Comparable Functions}\label{section:compare_functions}

We make some basic remarks about comparable functions (see Definition \ref{defn:comparable}).

\begin{lemma}\label{lemma:comparable_hessian}
Let $f, g: \mathbb{R}^d \to \mathbb{R}$ be real analytic functions on an open neighbourhood $U$ of the origin, with local minima at the origin and $f(0) = g(0) = 0$. Suppose there exist constants $c, d > 0$ such that 
\begin{equation}
c \cdot f(x) \leq g(x) \leq d \cdot f(x)
\end{equation}
for all $x \in U$. Let $Q_f$ and $Q_g$ be the Hessian matrices of $f$ and $g$ at the origin, respectively. Then:
\begin{enumerate}
\item[(i)] $\ker(Q_f) = \ker(Q_g)$.
\item[(ii)] For any unit eigenvector $v$ of $Q_f$ with eigenvalue $\lambda_f > 0$, we have $c \cdot \lambda_f \leq v^T Q_g v \leq d \cdot \lambda_f$.
\end{enumerate}
\end{lemma}

\begin{proof}
By hypothesis $f(0) = g(0) = 0$. Since $f$ and $g$ have local minima at $0$, we have $\nabla f(0) = \nabla g(0) = 0$, and both $Q_f$ and $Q_g$ are positive semidefinite. The Taylor expansions around the origin are:
\begin{align}
f(x) &= \frac{1}{2}x^T Q_f x + O(\|x\|^3)\\
g(x) &= \frac{1}{2}x^T Q_g x + O(\|x\|^3)
\end{align}
For (i) it suffices to show $\ker(Q_f) \subseteq \ker(Q_g)$. Let $v \in \ker(Q_f)$, so $Q_f v = 0$ and $v^T Q_f v = 0$. For sufficiently small $t$, consider $x = tv \in U$:
\begin{align}
f(tv) &= \frac{1}{2}t^2 v^T Q_f v + O(t^3) = O(t^3)
\end{align}

By our growth rate assumption, we have:
\begin{align}
g(tv) &\geq c \cdot f(tv) = c \cdot O(t^3) = O(t^3)\\
g(tv) &\leq d \cdot f(tv) = d \cdot O(t^3) = O(t^3)
\end{align}

However, we also know:
\begin{align}
g(tv) &= \frac{1}{2}t^2 v^T Q_g v + O(t^3)
\end{align}
For these equations to be consistent, we must have $v^T Q_g v = 0$. Since $Q_g$ is positive semidefinite, this implies $Q_g v = 0$, and therefore $v \in \ker(Q_g)$.

For (ii) let $v$ be a unit eigenvector of $Q_f$ with eigenvalue $\lambda_f > 0$, so $Q_f v = \lambda_f v$ and $v^T Q_f v = \lambda_f$.

For small $t$, consider $x = tv$:
\begin{align}
f(tv) &= \frac{1}{2}t^2 v^T Q_f v + O(t^3) = \frac{1}{2}t^2 \lambda_f + O(t^3)\\
g(tv) &= \frac{1}{2}t^2 v^T Q_g v + O(t^3)
\end{align}
From our growth rate constraint, we have:
\begin{align}
c \cdot f(tv) \leq g(tv) \leq d \cdot f(tv)
\end{align}
Substituting the expansions:
\begin{align}
c \cdot \left(\frac{1}{2}t^2 \lambda_f + O(t^3)\right) \leq \frac{1}{2}t^2 v^T Q_g v + O(t^3) \leq d \cdot \left(\frac{1}{2}t^2 \lambda_f + O(t^3)\right)
\end{align}
Dividing by $\frac{1}{2}t^2$ and taking the limit as $t \to 0$, we obtain:
\begin{align}
c \cdot \lambda_f \leq v^T Q_g v \leq d \cdot \lambda_f
\end{align}
This establishes the desired bounds on the quadratic form $v^T Q_g v$ for any unit eigenvector $v$ of $Q_f$ with positive eigenvalue $\lambda_f$.
\end{proof}

\begin{remark}
Lemma \ref{lemma:comparable_hessian} establishes that two real analytic functions with comparable growth rates near a common minimum point where both vanish have Hessian matrices with identical kernels. Thus the ellipsoids defined by the quadratic forms have the same ``flat directions'' and further their curvatures in other directions maintain bounded ratios: for a unit eigenvector $v$ of $Q_f$ with eigenvalue $\lambda_f > 0$,
\[
c \le \frac{v^T Q_g v}{\lambda_f} \le d \,.
\]
\end{remark}

Suppose that $v, v'$ are two unit eigenvectors of $Q_f$ with distinct positive eigenvalues $\lambda, \lambda'$. Define the curvatures of $Q_g$ in these directions as
\[
\mu = v^T Q_g v\,, \qquad \mu' = (v')^T Q_g v'\,.
\]
Then we have
\begin{align*}
\frac{\mu}{\mu'} = \frac{\mu}{\lambda} \cdot \frac{\lambda'}{\mu'} \cdot \frac{\lambda}{\lambda'} \le \frac{d}{c} \cdot \frac{\lambda}{\lambda'}
\end{align*}
and similarly with the lower bound, giving
\begin{equation}\label{eq:ratio_curve}
\frac{c}{d} \cdot \frac{\lambda}{\lambda'} \le \frac{\mu}{\mu'} \le \frac{d}{c} \cdot \frac{\lambda}{\lambda'}\,.
\end{equation}
This is applied in the following way: suppose the ratio of $\lambda$ to all other positive eigenvalues of $Q_f$ is large (i.e. $v$ is the most curved direction). Then the ratio of the curvatures $\frac{\mu}{\mu'}$ for $Q_g$ also can't be too small. In our applications, while $c$ can be chosen close to $1$, $d$ may be large, so that \eqref{eq:ratio_curve} need not be very informative.

\begin{definition}
    Let $S: \bb{R}_{> 0} \lto \bb{R}$ be defined by
    \[
    S(u) = -\log u + u - 1\,.
    \] 
\end{definition}

We will prove that $S$ is comparable to $\tfrac{1}{2}(u-1)^2$.

\begin{lemma}\label{lemma:S_bounds} Given $0 < c < 1$ and $d > 1$ there exists $\delta > 0$ such that
\begin{equation}\label{eq:S_bound}
c \cdot \tfrac{1}{2}(u-1)^2 \le S(u) \le d \cdot \tfrac{1}{2} (u-1)^2
\end{equation}
whenever $|1 - u| < \delta$.
\end{lemma}
\begin{proof}
We compute $S'(u) = -\tfrac{1}{u} + 1, S''(u) = \tfrac{1}{u^2}, S^{(3)}(u) = -\tfrac{2}{u^3}$ so $S(1) = 0, S'(1) = 0, S''(1) = 1$ and so for some $a$ between $1$ and $u$
\[
S(u) = \tfrac{1}{2}(u-1)^2 - \tfrac{1}{3 a^3}(u-1)^3\,.
\]
For a constant $0 < c < 1$ we have $c \cdot \tfrac{1}{2}(u-1)^2 \le S(u)$ if and only if
\[
\tfrac{1}{3 a^3}(u-1)^3 \le \tfrac{1}{2}[1-c](u-1)^2
\]
which is automatic if $u \le 1$ and for $u > 1$ holds if and only if
\begin{equation}\label{eq:u_lowerbound}
u \le 1 + \tfrac{3a^3}{2}\big[1-c\big]\,.
\end{equation}
Similarly if $1 < d$ we have $S(u) \le d \cdot \tfrac{1}{2}(u-1)^2$ if and only if
\[
\tfrac{1}{2}\big[1 - d\big](u-1)^2 \le \tfrac{1}{3 a^3}(u-1)^3
\]
which is automatic if $u = 1$ and for $u \neq 1$ is equivalent to
\begin{equation}\label{eq:u_upperbound}
1 + \tfrac{3a^3}{2}\big[1 - d\big] \le u\,.
\end{equation}
It follows that given $c,d$ as described, \eqref{eq:S_bound} will hold provided
\be\label{eq:bound_delta}
|1 - u| < \delta = \min\left\{\tfrac{1}{2}, \tfrac{3}{16} \min\{1-c,d-1\}\right\}\,.
\ee
Indeed, this ensures $a > \tfrac{1}{2}$, so $\tfrac{3a^3}{2} > \tfrac{3}{16}$ and both \eqref{eq:u_lowerbound} and \eqref{eq:u_upperbound} hold, as claimed.
\end{proof}

Note that we can make the constants $c,d$ as close to $1$ as we wish (respectively from below and above) by taking $\delta$ sufficiently small.

\subsection{Geometry of Comparable Functions}\label{sec:geom_comparable}

Let $V$ be an open neighbourhood of $0 \in \mathbb{R}^d$ and $f: V \lto \mathbb{R}$ a non-negative and non-constant analytic function with $f(0) = 0$. Let $g: V \lto \mathbb{R}$ be a comparable analytic function, say with $c, d > 0$ such that
\begin{equation}\label{eq:geom_comparable_1}
c f(x) \le g(x) \le d f(x)\,, \qquad \forall x \in V\,.
\end{equation}
By resolution of singularities (e.g. \cite[Theorem 2.3]{watanabeAlgebraicGeometryStatistical2009}) there is a triple $(V, U, h)$ where $V$ is an open set in $\mathbb{R}^d$ which contains $0$, $U$ is a $d$-dimensional real analytic manifold, $h: U \lto V$ is a real analytic map, such that $h$ is proper and with
\[
V_0 = \{ x \in V \mid f(x) = 0 \}\,, \qquad U_0 = \{ u \in U \mid f(h(u)) = 0 \}
\]
$h$ induces an analytic isomorphism of $U \setminus U_0$ with $V \setminus V_0$. Further, for every point $P \in U_0$ there are local coordinates $u_1,\ldots,u_d$ of $U$ with origin $P$ such that
\[
f(h(u)) = u_1^{2k_1} u_2^{2k_2} \cdots u_d^{2k_d}
\]
where $k_1,\ldots,k_d$ are non-negative integers, and the Jacobian determinant of $x = h(u)$ is
\[
h'(u) = b(u) u_1^{h_1} \cdots u_d^{h_d}
\]
where $b(u)$ is nonzero in the local chart and $h_1,\ldots,h_d$ are non-negative integers.

By \eqref{eq:geom_comparable_1} we have in any such local chart that
\[
c u_1^{2k_1} u_2^{2k_2} \cdots u_d^{2k_d} \le g(h(u)) \le d u_1^{2k_1} u_2^{2k_2} \cdots u_d^{2k_d}
\]
and hence by \cite[Theorem 2.6]{watanabeAlgebraicGeometryStatistical2009} there exists an open set within the coordinate chart around $P$ such that, on this open set, there is defined a real analytic function $t(u)$ with
\[
g(h(u)) = t(u) f(h(u))\,.
\]
That is to say, the ratio $t(u) = g(h(u))/f(h(u))$ is well-defined near $P$ on the resolution, \emph{even though} it may not be defined on the original space. Since $c \le t(u) \le d$ on this neighbourhood of $P$ we have that $t(u)$ is a unit in the ring of analytic functions. In particular, in the local ring $\mathcal{O}^{\text{an}}_P$ of germs of analytic functions on $U$ at $P$ we have that $f \circ h$ and $g \circ h$ generate the same ideal
\[
\langle g(h(u)) \rangle = \langle f(h(u)) \rangle
\]
and thus $\mathcal{O}^{\text{an}}_P/\langle f \circ h \rangle \cong \mathcal{O}^{\text{an}}_P / \langle g \circ h \rangle$. Thus any geometric property of $f$ that factors through the resolution is shared by $g$. This includes the real log canonical threshold (i.e. learning coefficient).

\subsection{Functions Comparable to the KL divergence}
\label{sec:KL_Squared_Distance}

Let $q(x)$ be a probability distribution on a non-empty finite set $I$ with $q(x) > 0$ for all $x \in I$. Let $Q$ be a non-empty finite set. We suppose given a conditional distribution $q(y|x)$ over $y \in Q$ given $x \in I$ such that $q(y|x) > 0$ for all $y \in Q$. Let $p(y|x,w)$ be a model defined on some $W \subseteq \mathbb{R}^d$ with open interior. Define
\begin{align*}
H(w) &= \sum_{x \in I} q(x) H(w,x)\,, \qquad H(w,x) = \sum_{y \in Q}\Big( p( y | x, w) - q(y|x) \Big)^2\,,\\
L(w) &= \sum_{x \in I} q(x) L(w,x)\,, \qquad L(w,x) = - \sum_{y \in Q} q(y|x) \log p(y|x,w) \,.\\
K(w) &= \sum_{x \in I} q(x) K(w,x)\,, \qquad K(w,x) = \sum_{y \in Q} q(y| x)\log\frac{q(y| x)}{p(y| x, w)}\,.
\end{align*}
We assume $p(y|x,w)$ is continuous in $w \in W$ and that $p(y|x,w) > 0$ for all $x \in I, y \in Q$ and $w \in W$, so that the above expressions are well-defined.

\begin{remark}\label{remark:realisable_app}
When the true distribution $q(x,y) = q(y|x) q(x)$ is realisable by a statistical model $q(y|x) = p(y|x,w^*)$ for some $w^*$ (which may be on the boundary of $W$) and we may replace $W$ by an infinitesimal neighbourhood of $w^*$. Since $p(y|x,w)$ is a continuous function of $w$ and there are only finitely many $x \in I$ and $y \in Q$, this means we may assume that given any $\varepsilon > 0$ we have simultaneously $\big| p(y|x,w) - q(y|x) \big| < \varepsilon$ for all $x,y$ and $w$. Since $q(y|x) > 0$ this means we may assume $p(y|x,w) > 0$ for all $w$, explaining the reason we adopt this hypothesis. \end{remark} 

\begin{lemma}\label{lemma:H_comparable_K} Given $0 < c < 1$ and $d > 1$ let $\delta > 0$ be as in \eqref{eq:bound_delta}. If 
\be\label{lemma:appendix_u_hyp}
\Bigg| \frac{p(y|x,w)}{q(y|x)} - 1 \Bigg| < \delta \qquad \forall x \in I, y \in Q, w \in W
\ee
then $\tfrac{1}{2} H(w)$ is comparable to $K(w)$.
\end{lemma}
\begin{proof}
We can apply Lemma \ref{lemma:S_bounds} to $u = \frac{p(y|x,w)}{q(y|x)}$ to obtain
\[
\tfrac{c}{2}\Bigg( \frac{p(y|x,w)}{q(y|x)} - 1 \Bigg)^2 \le -\log  \frac{p(y|x,w)}{q(y|x)} +  \frac{p(y|x,w)}{q(y|x)} - 1 \le \tfrac{d}{2}\Bigg( \frac{p(y|x,w)}{q(y|x)} - 1 \Bigg)^2
\]
which we can multiply by $q(y|x)$ and sum over $y \in Q$ to obtain
\be
\tfrac{c}{2} H'(w,x) \le K(w,x) \le \tfrac{d}{2} H'(w,x)
\ee
where
\[
H'(w,x) = \sum_{y \in Q} q(y|x)\Bigg( \frac{p(y|x,w)}{q(y|x)} - 1 \Bigg)^2\,.
\]
Thus $K(w,x)$ is comparable to $H'(w,x)$. Now
\[
H'(w,x) = \sum_{y \in Q} q(y|x)^{-1}\Big( p(y|x,w) - q(y|x) \Big)^2
\]
and if we let $C_x = \inf_{y \in Q} q(y|x)^{-1}$ and $D_x = \sup_{y \in Q} q(y|x)^{-1}$ then
\[
C_x H(w,x) \le H'(w,x) \le D_x H(w,x)
\]
so $H'(w,x)$ is comparable to $H(w,x)$. By transitivity $H(w,x)$ is comparable to $K(w,x)$, or more explicitly we have for all $x,w$
\[
\frac{c \, C_x}{2} H(w,x) \le K(w,x) \le \frac{d \, D_x}{2} H(w,x)\,.
\]
Multiplying through and taking the sum over $x \in I$
\[
c \, C \cdot \tfrac{1}{2} H(w) \le K(w) \le d \, D \cdot \tfrac{1}{2} H(w)
\]
where $C = \inf_{x \in I} C_x$, $D = \sup_{x \in I} D_x$ as desired. Of course if $K(w)$ is comparable to $\tfrac{1}{2} H(w)$ it is comparable to $H(w)$ but since we control $c, C, d, D$ this is more natural.
\end{proof}

\begin{remark} In the proof note that $q(y|x) \le 1$ so $q(y|x)^{-1} \ge 1$ and so $1 \le C_x$ and thus $1 \le C$. Hence, possibly by shrinking $\delta$ we can take the constant $c\,C$ in the lower bound as close as we like to $1$. The upper bound is less trivial to deal with, because if there are $y$ for which $q(y|x)$ is very small, $D$ is large. For this reason the comparability of $K(w)$ to a sum of squared differences is not trivial when $q$ does not have full support. In Section \ref{section:boundary_away} we explain one way to deal with this.
\end{remark}

\begin{remark}
As in Remark \ref{remark:realisable_app}, we typically arrange for \eqref{lemma:appendix_u_hyp} to hold by shrinking $W$ to a small neighbourhood of a true parameter.
\end{remark}

\begin{remark} Recall that $K(w) = L(w) - C$ where $C$ is the entropy of the true distribution (which does not depend on $w$). Thus $L(w)$ is comparable to $\tfrac{1}{2} H(w) + C$. If we assume we are working in a neighbourhood of a true parameter $w^*$ so that $K(w^*) = 0$ then $L(w^*) = C$ so that $L(w)$ is comparable to $\tfrac{1}{2} H(w) + L(w^*)$.
\end{remark}

\subsection{Moving Away From the Boundary}\label{section:boundary_away}

We define for $0 < \mu < 1$ the interpolation between a probability distribution and the barycenter $\bold{b} = \sum_{q \in Q} \frac{1}{|Q|} q$ of the probability simplex by
\begin{gather*}
\varepsilon_\mu: \Delta Q \lto \Delta Q\,,\\
\varepsilon_\mu( \bold{x} ) = (1-\mu) \bold{x} + \mu \bold{b}
\end{gather*}
where we write probability distributions as vectors in $\Delta Q \subseteq \mathbb{R}^Q$ and denote them with bold letters. Note that
\begin{align*}
\Vert \varepsilon_\mu(\bold{x}) - \varepsilon_\mu(\bold{y}) \Vert^2 &= \Vert (1-\mu)(\bold{x} - \bold{y}) \Vert^2 = (1-\mu)^2 \Vert \bold{x} - \bold{y} \Vert^2\,.
\end{align*}
Let $I, Q$ and $q(x)$ be as in Section \ref{sec:KL_Squared_Distance}. We assume given a conditional distribution $q(y|x)$ but we \emph{do not} assume that $q(y|x) > 0$ for all $y \in Q$. Nonetheless the distribution
\[
q_\mu(y|x) = \varepsilon_\mu q(y|x) \in \Delta Q
\]
does have this property. Similarly, let $p(y|x,w)$ be a model parametrised by $w \in W$ which is continuous, but not necessarily positive on all $y$. We define
\[
p_\mu(y|x,w) = \varepsilon_\mu p(y|x,w)
\]
which once again, is positive on all $y$. Assuming the hypothesis of Lemma \ref{lemma:H_comparable_K} for $q_\mu, p_\mu$ we obtain that the KL divergence $K_\mu(w)$ between $q_\mu$ and $p_\mu$ is comparable to
\begin{align*}
H_\mu(w) &= \sum_{x \in I} q(x) \Vert \varepsilon_\mu p(y|x,w) - \varepsilon_\mu q(y|x) \Vert^2\\
&= \sum_{x \in I} q(x) (1-\mu)^2 \Vert p(y|x,w) - q(y|x) \Vert^2\\
&= (1-\mu)^2 H(w)\,.
\end{align*}

\begin{remark}
The idea of moving distributions $p$ off the boundary in the context of program synthesis was introduced in \cite{clift2018derivatives}, but the elegant trick of applying it to both $p,q$ is due to Waring \cite{waring2021geometric}.
\end{remark}

\section{Details of the Step Function of $\mathcal{U}$}
\label{section:multi_tape}

The proof ${}_h\underline{\text{relstep}}$ of Proposition \cite[5.5]{clift2020encodings} encodes one step of a Turing machine in linear logic. In Section \ref{appendix:stagedpseudoutm} we considered simulated Turing machines with tape alphabet given by some fixed set $\Sigma$, with set of states some fixed set $Q$. Here we fix these sets to be as follows, where $n,m \geq 1$:
\begin{equation}
\label{eq:simulated_machines}
    \Sigma = \{0,\ldots, n-1\}, Q = \{0,\ldots, m-1\}, D = \{0,1,2\} = \{\texttt{Left}, \texttt{Stay}, \texttt{Right}\}.
\end{equation}
and we consider a staged pseudo-UTM $\mathcal{U}$. There are $N = nm$ tuples associated to the transition function of $M$. The aim of this section is to adapt ${}_h\underline{\text{relstep}}$, for $h > 0$, to a proof ${}_h\underline{\text{$\mathcal{U}$relstep}}$ sufficiently capturing and simulating the computation of $\mathcal{U}$ for a single simulated step of $M$ (which amounts to $10N+5$ steps of $\mathcal{U}$).

One way to approach this would be to come up with a generic encoding of Turing machines with $k > 0$ tapes in linear logic and then specialise to the particular case of $\mathcal{U}$. This is possible, but we simplify things by making use of the following properties of $\mathcal{U}$:

\begin{itemize}
    \item[(i)] The description tape, staging tape, and state tape are of determined, finite length.
    \item[(ii)] The head position of the description tape depends only on the time step $\mu$ of $\mathcal{U}$.
    \item[(iii)] The head positions of the staging tape depends only on the state of $\mathcal{U}$.
    \item[(iv)] The head position of the state tape never changes.
    \item[(v)] The symbols and states of the simulated machine can be encoded in single tape squares of the $\mathcal{U}$.
\end{itemize}
Let $\mu, h \geq 0$ be integers with $\mu \leq 10N+5$. We define a type ${}_h^\mu\underline{\mathcal{U}\textbf{tur}}$ and, for $\mu < 10N+5$, a proof
\begin{equation}
    {}_h^{\mu}\underline{\text{$\mathcal{U}$relstep}}: {}_h^\mu\underline{\mathcal{U}\textbf{tur}} \vdash {}_h^{\mu+1}\underline{\mathcal{U}\textbf{tur}}
\end{equation}
encoding a single step of $\mathcal{U}$ at position $\mu$ within the overall cycle of the staged pseudo-UTM (see Section \ref{appendix:stagedpseudoutm}). By properties (i)-(v), we can reduce to a finite set of squares which are used in a single, fixed timestep of $\mathcal{U}$. For each of these squares we will construct a plain proof and another plain proof encoding the updates of the states of $\mathcal{U}$. Following \cite{clift2020encodings}, the proof ${}_h^{\mu}\underline{\text{$\mathcal{U}$relstep}}$ will then be given by promoting these proofs, taking a tensor product, and performing contractions. We construct the proof associated to these squares one tape at a time. 
\\

In the following the base type $A$ is chosen so that all required denotations are linearly independent; this is possible by \cite[Appendix B]{clift2020encodings}. Nothing depends on the choice of $A$ and we generally drop it from the notation and write $\textbf{bool}$ for $\textbf{bool}_A$.

We will make extensive use of the following two lemmas:

\begin{lemma}\label{encoding functions of n bools mod}
Given $k \ge 1$ and any function $f: \{0, ..., n-1\}^k \to \{0, ..., m-1\}$, there exists a proof $F$ of $k\,\tNBool \vdash \tMBool$ which encodes $f$.
\end{lemma} 

\begin{proof}
The proof is by induction on $k$, with $k = 1$ being Lemma \cite[4.7]{clift2020encodings}. Assuming the lemma holds for $k-1$ we for each $0 \le z \le n-1$ let $F_z$ be the proof of $(k-1)\,\tNBool \vdash \tMBool$ encoding the function
\[
f(-, z): \{0, ..., n-1\}^{k-1} \to \{0, ..., m-1\}\,.
\]
Then the proof
\begin{center}
\AxiomC{$\proofvdots{F_0}$}
\noLine\UnaryInfC{$A^m, (k-1)\, \tNBool \vdash A$}
\AxiomC{}
\noLine\UnaryInfC{$...$}
\AxiomC{$\proofvdots{F_{n-1}}$}
\noLine\UnaryInfC{$A^m, (k-1)\, \tNBool \vdash A$}

\RightLabel{\scriptsize $\&R$}
\doubleLine\TrinaryInfC{$A^m, (k-1)\, \tNBool \vdash A^n$}

\AxiomC{}
\UnaryInfC{$A \vdash A$}

\RightLabel{\scriptsize $\multimap L$}
\BinaryInfC{$A^m, (k-1)\, \tNBool, \tNBool \vdash A$}
\RightLabel{\scriptsize $\multimap R$}
\UnaryInfC{$(k-1)\, \tNBool, \tNBool \vdash \tMBool$}

\DisplayProof
\end{center}
encodes $f$.
\end{proof}

Similarly it follows that:

\begin{lemma}
\label{lem:update_multiple_tape}
    Let $n_1,n_2,m \geq 1$ and consider a function
    \begin{equation}
        f: \{0, \ldots, n_1-1\} \times \{0, \ldots, n_2-1\}\lto \{0,\ldots, m-1\}
    \end{equation}
    there exists a proof $F$ of $\textbf{${_{n_1}}$bool}, \textbf{${_{n_2}}$bool} \vdash \textbf{${_m}$bool}$ which encodes $f$.
\end{lemma}

\subsection{The description tape}

The description tape of $\mathcal{U}$ is
\begin{equation}
    X \sigma_1 q_1 \sigma_1' q_1' d_1 \ldots \sigma_N q_N \sigma_N' q_N' d_N X
\end{equation}
The natural orders $0 < \ldots < n-1, 0 < \ldots < m-1$ on $\Sigma, Q$ respectively induce the lexicographical order on the set $\Sigma \times Q$ which means the information of the description tape can be compressed to tuples of size 3 instead of size 5 as follows:
\begin{equation}
\label{eq:description_compressed}
    \sigma'_1 q_1' d_1 \ldots \sigma'_N q'_N d_N
\end{equation}
where we have dropped the $X$ at the beginning and at the end of the tape as these are only used in the definition of $\mathcal{U}$ as a Turing machine and will not be used in the encoding. Here $\sigma'_i \in \Sigma$, $q'_i \in Q$, $d_i \in D$ for $1 \le i \le N$. We use the type \textbf{dscr} and proofs $\mathcal{P}^i_{\text{dscr}}$ introduced in Section \ref{section:encoding_utm_step}.

For each $i = 1, \ldots, 3N$ we define a proof $S^{\text{dscr}}_{\mu, i}$ of the following sequent:
\begin{equation}
    S^{\text{dscr}}_{\mu, i} :
    \begin{cases}
        {!}{}_{n}\textbf{bool} \vdash {!}{}_{n}\textbf{bool},& i = 1\text{ mod }3,\\
        {!}{}_{m}\textbf{bool} \vdash {!}{}_{m}\textbf{bool},& i = 2\text{ mod }3,\\
        {!}{}_{3}\textbf{bool} \vdash {!}{}_{3}\textbf{bool},& i = 0\text{ mod }3,\\
    \end{cases}
\end{equation}
Each of these consists of an axiom followed by dereliction and promotion, encoding the identity function. For example, $S_{\mu, 1}^{\text{dscr}}$ is the following proof:
\begin{prooftree}
    \AxiomC{}
    \RightLabel{$\ax$}
    \UnaryInfC{${}_{n}\textbf{bool} \vdash {}_{n}\textbf{bool}$}
    \RightLabel{$\der$}
    \UnaryInfC{${}_{n}!\textbf{bool} \vdash {}_{n}\textbf{bool}$}
    \RightLabel{$\prom$}
    \UnaryInfC{${}_{n}!\textbf{bool} \vdash {}_{n}!\textbf{bool}$}
\end{prooftree}
We define the proof associated to the description tape as the tensor of all of these:
\begin{equation}
    S^{\text{dscr}}_{\mu} = \bigotimes_{i = 1}^{3N}S^{\text{dscr}}_{\mu, i}: \textbf{dscr} \vdash \textbf{dscr}.
\end{equation}
Note that this proof is independent of $\mu$.

\subsection{The staging tape}

We label these squares $s_0, s_1, s_2$. If $\Sigma_{\text{stg}}^i$ denotes the set of symbols which may appear on square $s_i$ for $i =0,1,2$, then:
\begin{equation}
\Sigma_{\text{stg}}^i =
    \begin{cases}
        \Sigma \cup \{X\},& i =0,\\
        Q \cup \{X\},& i =1,\\
        D \cup \{X\},& i =2.
    \end{cases}
\end{equation}
We introduce the following type:
\begin{equation}
    \textbf{stg} = \underset{\sigma'}{\underbrace{{!}{}_{n+1}\textbf{bool}}} \otimes \underset{q'}{\underbrace{{!}{}_{m+1}\textbf{bool}}} \otimes \underset{d}{\underbrace{{!}{}_{4}\textbf{bool}}}
\end{equation}
where we have decorated each factor with the kind of information that is written to the staging tape during the execution of the UTM. We introduce sets of proofs:
\begin{equation}
    \mathcal{P}_i^{\text{stg}} =
    \begin{cases}
        \{ \llbracket \underline{j} \rrbracket \}_{j = 0, \ldots, n} \subseteq \llbracket {}_{n+1}\textbf{bool}\rrbracket,& i = 0,\\
        \{ \llbracket \underline{j} \rrbracket \}_{j = 0, \ldots, m} \subseteq \llbracket {}_{m+1}\textbf{bool}\rrbracket,& i = 1,\\
        \{ \llbracket \underline{j} \rrbracket \}_{0 \le j \le 3} \subseteq \llbracket {}_{4}\textbf{bool}\rrbracket,& i = 2.
    \end{cases}
\end{equation}
For each $i = 0,1,2$ we define a proof $S^{\text{stg}}_{\mu,i}$. By convention $X$ is encoded as $\underline{n}, \underline{m}, \underline{3}$.

We show the details for $i = 0$ and leave $i = 1,2$ to the reader. Say $0 \leq \mu \le 5N$. If $\mu \in \{0,1,3,4\}\text{ mod 5}$ then $S^{\text{stg}}_{\mu, 0}$ consists of a single Axiom-rule. Now say $\mu = 2\text{ mod 5}$. Recall that we have fixed a Turing machine $M$ with associated code $[M]$. Then $S^{\text{stg}}_{\mu,0}$ encodes the following function:
\begin{gather*}
    W^{\texttt{code}} \times (\Sigma \cup \{X\}) \times Q_{\text{UTM}} \lto \Sigma \cup \{X\}\\
    ([M], \sigma, \varphi) \longmapsto
    \begin{cases}
        [M]_{\hat{\mu}},& \varphi = \text{copySymbol},\\
        \sigma,& \varphi \neq \text{copySymbol}
    \end{cases}
\end{gather*}
where $\hat{\mu}$ is the entry of the description tape corresponding to the tuple in turn corresponding to $\mu$. This is a proof
\begin{equation}
    S^{\text{stg}}_{\mu,0}: \textbf{dscr} \otimes {!}{}_{n+1}\textbf{bool} \otimes {!}{}_{13}\textbf{bool} \vdash {!}{}_{n+1}\textbf{bool}.
\end{equation}
For existence, apply Lemma \ref{lem:update_multiple_tape} to encode the selected description symbol and staging symbol. Apply the lemma again to combine the resulting encoding with the UTM state, and compose the two proofs by Cut. Apply dereliction to the three hypotheses, add the unused description-tape components by weakening, and promote the conclusion. Finally, use Exchange and the left tensor rule to assemble the input of the displayed sequent.

If $\mu = 5N+1\text{ mod }P$ then $S^{\text{stg}}_{\mu, 0}$ encodes the constant function mapping everything to $X$. For $5N+2 \leq \mu < 10N+5$ we define $S^{\text{stg}}_{\mu, 0}$ to consist of an Axiom-rule. We define the proof associated to the staging tape as the tensor of all of these:
\begin{equation}
    S^{\text{stg}}_{\mu} = S_{\mu,0}^{\text{stg}} \otimes S_{\mu,1}^{\text{stg}} \otimes S_{\mu,2}^{\text{stg}}.
\end{equation}
This is a proof of the following sequents depending on $\mu$:
\begin{equation}
    \begin{cases}
        \textbf{stg} \vdash \textbf{stg},& \mu \leq 5N, \mu = 0,1\text{ mod }5,\\
        \textbf{dscr} \otimes \textbf{stg} \otimes {!}{}_{13}\textbf{bool} \vdash \textbf{stg},& \mu \leq 5N, \mu = 2,3,4\text{ mod }5,\\
        \textbf{stg} \vdash \textbf{stg},& \mu > 5N
    \end{cases}
\end{equation}

\subsection{The state tape}

We consider
\begin{equation}
    \mathcal{P}^{\text{state}} = \{\llbracket \underline{j}\rrbracket\}_{j = 0, \ldots, m-1} \subseteq \llbracket {}_m\textbf{bool}\rrbracket\,.
\end{equation}
We define a proof $S^{\text{state}}_\mu$. If $\mu \neq 5N+2$ then $S^{\text{state}}_\mu$ consists of a single Axiom-rule. If $\mu = 5N+2$ then we encode the function
\begin{gather*}
    (Q \cup \{X\}) \times Q \lto Q\\
    (q,q') \longmapsto
    \begin{cases}
        q, & q \neq X,\\
        q', & q = X.
    \end{cases}
\end{gather*}
This is a proof of the following sequents depending on $\mu$:
\begin{equation}
    \begin{cases}
        {!}{}_{m}\textbf{bool} \vdash {!}{}_{m}\textbf{bool},& \mu \neq 5N+2,\\
        {!}{}_{m+1}\textbf{bool} \otimes {!}{}_{m}\textbf{bool} \vdash {!}{}_{m}\textbf{bool},& \mu = 5N+2.
    \end{cases}
\end{equation}

\subsection{The working tape}

Each square on the working tape is occupied by a symbol in $\Sigma$ and so for all $i = -h-1, \ldots, h+1$ we take
\begin{equation}
    \mathcal{P}^{\text{work}}_{h,i} = \{\llbracket \underline{j}\rrbracket\}_{j = 0, \ldots, n-1} \subseteq \llbracket {}_n\textbf{bool}\rrbracket.
\end{equation}
For $\mu \leq 5N+2$ and $i = -h, \ldots, h$ we define a proof $S^{\text{work}}_{\mu,i}$. For $\mu > 5N+2$ and $i = -h-1, \ldots, h+1$ we define a proof $S^{\text{work}}_{\mu,i}$. For $0 \leq \mu < 10N+5$ and every admissible $i$, the proof $S^{\text{work}}_{\mu,i}$ consists of a single Axiom-rule whenever $\mu \notin \{5N+1,5N+3\}$. This also holds when $\mu = 5N+1$ and $i \neq 0$. Now say $\mu = 5N+1$ and $i = 0$. Then we encode the function
\begin{align*}
    (\Sigma \cup \{X\}) \times \Sigma \lto \Sigma\\
    (\sigma,\sigma') \longmapsto
    \begin{cases}
        \sigma, & \sigma \neq X,\\
        \sigma' & \sigma = X.
    \end{cases}
\end{align*}
Now say $\mu = 5N+3$ and $i \in \{-h+1, \ldots, h-1\}$. Then we need to encode the function
\begin{gather*}
    (D \cup \{X\}) \times \Sigma \times \Sigma \times \Sigma \lto \Sigma\\
    (d, \sigma, \sigma', \sigma'') \longmapsto
    \begin{cases}
        \sigma, & d = \text{Left},\\
        \sigma', & d \in \{\text{Stay},X\},\\
        \sigma'', & d = \text{Right}.\\
    \end{cases}
\end{gather*}
If $i = -h$ then we need to encode the function:
\begin{gather*}
    (D \cup \{X\}) \times \Sigma \times \Sigma \lto \Sigma\\
    (d, \sigma, \sigma') \longmapsto
    \begin{cases}
        0, & d = \text{Left},\\
        \sigma, & d \in \{\text{Stay},X\},\\
        \sigma', & d = \text{Right}.\\
    \end{cases}
\end{gather*}
and similarly for $i = h$. Finally, if $i = -h-1$ then we need to encode the function:
\begin{gather*}
    (D \cup \{X\}) \times \Sigma \lto \Sigma\\
    (d, \sigma) \longmapsto
    \begin{cases}
        0, & d = \text{Left},\\
        0, & d \in \{\text{Stay},X\},\\
        \sigma, & d = \text{Right}.\\
    \end{cases}
\end{gather*}
and similarly for $i = h+1$. The proof associated to the working tape is the tensor of all these (with promotions and contractions)
\begin{equation}
    \begin{cases}
        S^{\text{work}}_{\mu} = \bigotimes_{i=-h}^{h} S^{\text{work}}_{\mu,i},& \mu \leq 5N+2,\\
        S^{\text{work}}_{\mu} = \bigotimes_{i=-h-1}^{h+1} S^{\text{work}}_{\mu,i},& \mu = 5N+3,\\
        S^{\text{work}}_{\mu} = \bigotimes_{i=-h-1}^{h+1} S^{\text{work}}_{\mu,i},& \mu \geq 5N+4.
    \end{cases}
\end{equation}
which is of the following type:
\begin{equation}
    \begin{cases}
        {!}{}_{n}\textbf{bool}^{\otimes 2h+1} \vdash {!}{}_{n}\textbf{bool}^{\otimes 2h+1},& \mu \leq 5N \text{ or } \mu = 5N+2,\\
        {!}{}_{n+1}\textbf{bool} \otimes {!}{}_{n}\textbf{bool}^{\otimes 2h+1} \vdash {!}{}_{n}\textbf{bool}^{\otimes 2h+1},& \mu = 5N+1,\\
        {!}{}_{4}\textbf{bool} \otimes {!}{}_{n}\textbf{bool}^{\otimes 2h+1} \vdash {!}{}_{n}\textbf{bool}^{\otimes 2h+3},& \mu = 5N+3,\\
        {!}{}_{n}\textbf{bool}^{\otimes 2h+3} \vdash {!}{}_{n}\textbf{bool}^{\otimes 2h+3},& \mu \geq 5N+4.
    \end{cases}
\end{equation}

\subsection{The state of $\mathcal{U}$}

There are $13$ states of $\mathcal{U}$, so we consider
\begin{equation}
    \mathcal{P}^{\mathcal{U}\text{state}}_i = \{\llbracket \underline{j}\rrbracket\}_{j = 0, \ldots, 12} \subseteq \llbracket {}_{13}\textbf{bool}\rrbracket.
\end{equation}
We define a proof $S^{\mathcal{U}\text{state}}_\mu$ encoding the state at time $\mu+1$ from the configuration at time $\mu$, for $0 \leq \mu < P$.\\

Say $0 \leq \mu < 5N$, and let $j = \lfloor \mu/5 \rfloor + 1$ be the index of the tuple being scanned. Say $\mu = 0\text{ mod }5$. Then $S^{\mathcal{U}\text{state}}_\mu$ encodes the function
\begin{align*}
    \Sigma &\lto Q_{\text{UTM}}\\
    \sigma &\longmapsto
    \begin{cases}
        \text{compState}, & \sigma = \sigma_j,\\
        \neg\text{compState}, & \sigma \neq \sigma_j.
    \end{cases}
\end{align*}
Here $\sigma_j,q_j$ are the first two symbols of the $j$th tuple on the description tape.

Say $\mu = 1\text{ mod }5$. Then $S^{\mathcal{U}\text{state}}_\mu$ encodes the function
\begin{align*}
    Q \times Q_{\text{UTM}} &\lto Q_{\text{UTM}}\\
    (q,\varphi) &\longmapsto
    \begin{cases}
        \text{copySymbol}, & \varphi = \text{compState} \text{ and } q = q_j,\\
        \neg\text{copySymbol}, & \text{otherwise}.
    \end{cases}
\end{align*}

For $\mu = 2\text{ mod }5$, the proof encodes the function mapping $\text{copySymbol}$ to $\text{copyState}$ and every other state to $\neg\text{copyState}$. For $\mu = 3\text{ mod }5$, it maps $\text{copyState}$ to $\text{copyDir}$ and every other state to $\neg\text{copyDir}$. For $\mu = 4\text{ mod }5$, it maps every state to $\text{compSymbol}$.

Now say $5N \leq \mu < P$. Then each $S^{\mathcal{U}\text{state}}_\mu$ encodes a constant function, with value
\[
    \begin{cases}
        \text{updateSymbol}, & \mu = 5N,\\
        \text{updateState}, & \mu = 5N+1,\\
        \text{updateDir}, & \mu = 5N+2,\\
        \text{resetDescr}, & 5N+3 \leq \mu < 10N+4,\\
        \text{compSymbol}, & \mu = 10N+4.
    \end{cases}
\]
The proofs have the following types:
\begin{equation}
    \begin{cases}
        {!}{}_{n}\textbf{bool} \vdash {!}{}_{13}\textbf{bool},& \mu < 5N, \mu = 0\text{ mod }5,\\
        {!}{}_{m}\textbf{bool} \otimes {!}{}_{13}\textbf{bool} \vdash {!}{}_{13}\textbf{bool},& \mu < 5N, \mu = 1\text{ mod }5,\\
        {!}{}_{13}\textbf{bool} \vdash {!}{}_{13}\textbf{bool}, & \mu < 5N, \mu = 2,3,4\text{ mod }5,\\
        \vdash {!}{}_{13}\textbf{bool},& 5N \leq \mu < P.
    \end{cases}
\end{equation}

\subsection{The type ${}_h^\mu\underline{\mathcal{U}\textbf{tur}}$}\label{section:tape_type_utm}

First we define the type such that proofs of that type include encodings of the configuration of the staged pseudo-UTM at various steps within its cycle.

\begin{definition} Define for $h \ge 0$ and $0 \le \mu \le 10N + 5$
\begin{equation}
    {}_h^\mu\underline{\mathcal{U}\textbf{tur}} = \textbf{dscr} \otimes \textbf{stg} \otimes {!}{}_m\textbf{bool} \otimes {!}{}_n\textbf{bool}^{\otimes h'}\otimes {!}{}_{13}\textbf{bool}
\end{equation}
where
\begin{equation}
    h' =
    \begin{cases}
        2h+1,& \mu < 5N+4,\\
        2h+3,& \mu \geq 5N+4.
    \end{cases}
\end{equation}
\end{definition}

Next we define the proof which encodes the mapping of configurations of the staged pseudo-UTM at position $\mu$ to $\mu+1$ within the cycle.

\begin{definition} Define for $h \ge 0$ and $0 \le \mu < 10N + 5$
    \begin{equation}
    {}_h^{\mu}\underline{\text{$\mathcal{U}$relstep}}: {}_h^\mu\underline{\mathcal{U}\textbf{tur}} \vdash {}_h^{\mu+1}\underline{\mathcal{U}\textbf{tur}}.
    \end{equation}
    as follows. We use the proofs already constructed for the tapes and the state of $\mathcal{U}$:
    \begin{equation}
        S^{\text{dscr}}_{\mu},\quad S^{\text{stg}}_{\mu},\quad
        S^{\text{state}}_{\mu},\quad
        S^{\text{work}}_{\mu},\quad
        S^{\mathcal{U}\text{state}}_{\mu}
    \end{equation}
    Regard the tensor-product inputs of these proofs as contexts of individual exponential hypotheses. We tensor the proofs, apply Exchange and Contraction to identify repeated occurrences of each input component, and add unused input components by Weakening. Finally, we apply the left tensor rule to assemble the input configuration; this final application of the left tensor rule is implicit in the following display:
    \begin{align*}
    {}_h^{\mu}\underline{\text{$\mathcal{U}$relstep}} = &\operatorname{Weak}\Big(\operatorname{Ctr}\Big(\operatorname{Ex}\Big(S^{\text{dscr}}_{\mu} \otimes S^{\text{stg}}_{\mu}\\
    & \qquad \otimes S^{\text{state}}_{\mu} \otimes S^{\text{work}}_{\mu} \otimes S^{\mathcal{U}\text{state}}_{\mu} \Big)\Big)\Big)\,.
    \end{align*}
\end{definition}

Finally, we cut these proofs together for $\mu = 0, \mu = 1, \ldots, \mu = 10N + 4$ to get a proof computing a single cycle of the UTM.

\begin{definition}\label{defn:urelstep_t} We define
    \begin{equation}
        {}_h\underline{\text{$\mathcal{U}$relstep}} = {}_h^{10N+4}\underline{\text{$\mathcal{U}$relstep}} \mid \ldots \mid {}_h^{1}\underline{\text{$\mathcal{U}$relstep}}
        \mid
        {}_h^{0}\underline{\text{$\mathcal{U}$relstep}}
    \end{equation}
    and for any integer $t \ge 1$
    \begin{equation}
        {}_h\underline{\text{$\mathcal{U}$relstep}}^t = \underset{t \text{ proofs and } t - 1 \text{ cuts}}{\underbrace{{}_{h+t-1}\underline{\text{$\mathcal{U}$relstep}} \mid \ldots \mid {}_{h+1}\underline{\text{$\mathcal{U}$relstep}}
        \mid
        {}_h\underline{\text{$\mathcal{U}$relstep}}}}
    \end{equation}
where $\alpha \mid \beta$ denotes the cut of two proofs $\alpha, \beta$.
\end{definition}

\begin{remark}
    We recall that the cut of component-wise plain proofs is component-wise plain. Hence ${}_h\underline{\text{$\mathcal{U}$relstep}}^t$ is component-wise plain. If we perform cut-elimination on this proof we obtain a component-wise plain proof which is the tensor product of a family of proofs indexed by the active squares of $\mathcal{U}$. If we consider the plain proof associated to the only active square of the state tape, then we obtain a plain proof $\psi$ whose naive probabilistic extension is
    \begin{equation}
        \Delta \psi = \Delta \text{step}^t
    \end{equation}
    of the body of the paper.
\end{remark}

\section{Computation Paths for the UTM}\label{section:comp_path_utm}

We continue the treatment of computation paths for $\mathcal{U}$ from Section \ref{section:graphical_model_utm}. The DGM for one simulated step of $T$ is shown in Figure \ref{figure:symbol_write_error_diagram}. We perform the following simplification: in general, given a DGM, consider a vertex $B$ for which there exists a unique pair of edges $e_1,e_2$ such that $B$ is the target of $e_1$ and the source of $e_2$:
\begin{equation}
\label{eq:DGM}
\begin{tikzcd}
	A & B & C
	\arrow[from=1-1, to=1-2, "{e_1}"]
	\arrow[from=1-2, to=1-3, "{e_2}"]
\end{tikzcd}
\end{equation}
If $\operatorname{Pr}(B \vert C) = \delta_{B=C}$, the joint distribution factorises as:
\begin{equation}
\label{eq:simplifiable}
    \operatorname{Pr}(A,B,C) = \operatorname{Pr}(A\vert B)\operatorname{Pr}(B\vert C)\operatorname{Pr}(C).
\end{equation}
Marginalising \eqref{eq:simplifiable} over $B$, we obtain:
\begin{align*}
    \sum_B \operatorname{Pr}(A\vert B)\operatorname{Pr}(B\vert C)\operatorname{Pr}(C) &= \sum_B \operatorname{Pr}(A\vert B)\delta_{B = C}\operatorname{Pr}(C)\\
    &= \operatorname{Pr}(A\vert C)\operatorname{Pr}(C) = \operatorname{Pr}(A,C).
\end{align*}
In other words, in any DGM, a subgraph of the form \eqref{eq:DGM} with $\operatorname{Pr}(B \vert C) = \delta_{B=C}$ can be replaced by the following:
\begin{center}\begin{tikzcd}
	A & C.
	\arrow[from=1-1, to=1-2]
\end{tikzcd}\end{center}
In this case we say that we have \emph{collapsed} the DGM. 

\subsection{Symbol write error}

First we consider a symbol write error. Applying the above technique of collapsing the DGM we obtain the equivalent model shown in Figure \ref{figure:write_error_two_steps}, where in this and subsequent diagrams a dashed arrow means a collection of $N$ arrows and a squiggly arrow means a collection of $2N$ arrows.

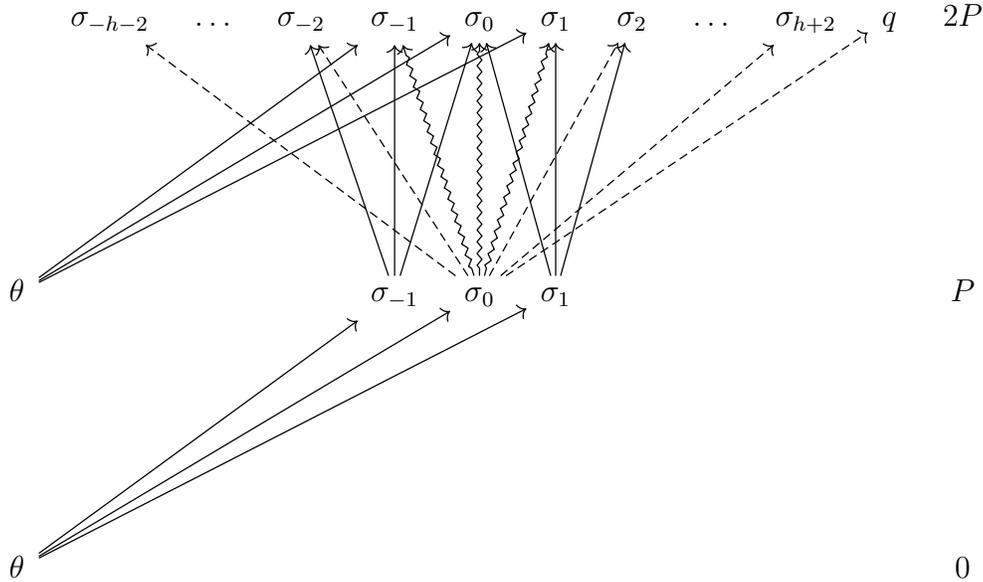
\begin{figure}[htpb]
    \centering
\[\begin{tikzcd}[column sep = tiny, row sep = huge]
	& {\sigma_{-h-2}} & \ldots & {\sigma_{-2}} & {\sigma_{-1}} & {\sigma_0} & {\sigma_1} & {\sigma_2} & \ldots & {\sigma_{h+2}} & q & 2P \\
	\\
	\theta &&&& {\sigma_{-1}} & {\sigma_{0}} & {\sigma_1} &&&&& P \\
	\\
	\theta &&&&&&&&&&& 0
	\arrow[from=3-1, to=1-5]
	\arrow[from=3-1, to=1-6]
	\arrow[from=3-1, to=1-7]
	\arrow[from=3-5, to=1-4]
	\arrow[from=3-5, to=1-5]
	\arrow[from=3-5, to=1-6]
	\arrow[dashed, from=3-6, to=1-2]
	\arrow[dashed, from=3-6, to=1-4]
	\arrow[shift left, squiggly, from=3-6, to=1-5]
	\arrow[shift right, from=3-6, to=1-5]
	\arrow[shift left, squiggly, from=3-6, to=1-6]
	\arrow[shift right, from=3-6, to=1-6]
	\arrow[shift left, squiggly, from=3-6, to=1-7]
	\arrow[shift right, from=3-6, to=1-7]
	\arrow[dashed, from=3-6, to=1-8]
	\arrow[dashed, from=3-6, to=1-10]
	\arrow[dashed, from=3-6, to=1-11]
	\arrow[from=3-7, to=1-6]
	\arrow[from=3-7, to=1-7]
	\arrow[from=3-7, to=1-8]
	\arrow[from=5-1, to=3-5]
	\arrow[from=5-1, to=3-6]
	\arrow[from=5-1, to=3-7]
\end{tikzcd}\]
\caption{\label{figure:write_error_two_steps} A collapsed version of two stacked diagrams as in Figure \ref{figure:symbol_write_error_diagram} (that is, two simulated timesteps of a TM, considering a possible error in the symbol to write square of the description tape) showing only paths that begin at $\theta$.}
\end{figure}

If we consider the simulated steps from time $P$ to time $3P$ we obtain Figure \ref{figure:write_error_four_steps}.

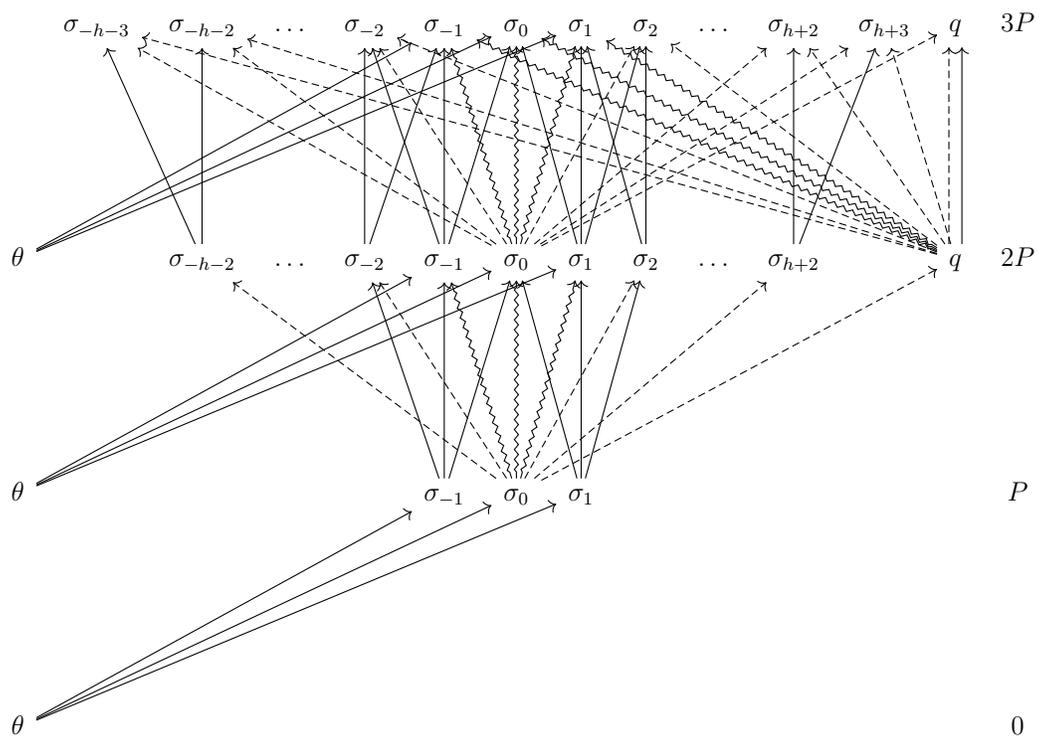
\begin{figure}[htpb]
    \centering
\[\adjustbox{scale=0.85}{\begin{tikzcd}[column sep = tiny, row sep = huge]
	& {\sigma_{-h-3}} & {\sigma_{-h-2}} & \ldots & {\sigma_{-2}} & {\sigma_{-1}} & {\sigma_0} & {\sigma_1} & {\sigma_2} & \ldots & {\sigma_{h+2}} & {\sigma_{h+3}} & q & 3P \\
	\\
	\theta && {\sigma_{-h-2}} & \ldots & {\sigma_{-2}} & {\sigma_{-1}} & {\sigma_0} & {\sigma_1} & {\sigma_2} & \ldots & {\sigma_{h+2}} && q & 2P \\
	\\
	\theta &&&&& {\sigma_{-1}} & {\sigma_0} & {\sigma_1} &&&&&& P \\
	\\
	\theta &&&&&&&&&&&&& 0
	\arrow[from=3-1, to=1-6]
	\arrow[from=3-1, to=1-7]
	\arrow[from=3-1, to=1-8]
	\arrow[from=3-3, to=1-2]
	\arrow[from=3-3, to=1-3]
	\arrow[from=3-5, to=1-5]
	\arrow[from=3-5, to=1-6]
	\arrow[from=3-6, to=1-5]
	\arrow[from=3-6, to=1-6]
	\arrow[from=3-6, to=1-7]
	\arrow[dashed, from=3-7, to=1-2]
	\arrow[dashed, from=3-7, to=1-3]
	\arrow[dashed, from=3-7, to=1-5]
	\arrow[shift left, squiggly, from=3-7, to=1-6]
	\arrow[shift right, from=3-7, to=1-6]
	\arrow[shift left, squiggly, from=3-7, to=1-7]
	\arrow[shift right, from=3-7, to=1-7]
	\arrow[shift left, squiggly, from=3-7, to=1-8]
	\arrow[shift right, from=3-7, to=1-8]
	\arrow[dashed, from=3-7, to=1-9]
	\arrow[dashed, from=3-7, to=1-11]
	\arrow[dashed, from=3-7, to=1-12]
	\arrow[dashed, from=3-7, to=1-13]
	\arrow[from=3-8, to=1-7]
	\arrow[from=3-8, to=1-8]
	\arrow[from=3-8, to=1-9]
	\arrow[from=3-9, to=1-8]
	\arrow[from=3-9, to=1-9]
	\arrow[from=3-11, to=1-11]
	\arrow[from=3-11, to=1-12]
	\arrow[dashed, from=3-13, to=1-2]
	\arrow[dashed, from=3-13, to=1-3]
	\arrow[dashed, from=3-13, to=1-5]
	\arrow[squiggly, from=3-13, to=1-6]
	\arrow[squiggly, from=3-13, to=1-7]
	\arrow[squiggly, from=3-13, to=1-8]
	\arrow[dashed, from=3-13, to=1-9]
	\arrow[dashed, from=3-13, to=1-11]
	\arrow[dashed, from=3-13, to=1-12]
	\arrow[shift left, dashed, from=3-13, to=1-13]
	\arrow[shift right, from=3-13, to=1-13]
	\arrow[from=5-1, to=3-6]
	\arrow[from=5-1, to=3-7]
	\arrow[from=5-1, to=3-8]
	\arrow[from=5-6, to=3-5]
	\arrow[from=5-6, to=3-6]
	\arrow[from=5-6, to=3-7]
	\arrow[dashed, from=5-7, to=3-3]
	\arrow[dashed, from=5-7, to=3-5]
	\arrow[shift left, squiggly, from=5-7, to=3-6]
	\arrow[shift right, from=5-7, to=3-6]
	\arrow[shift left, squiggly, from=5-7, to=3-7]
	\arrow[shift right, from=5-7, to=3-7]
	\arrow[shift left, squiggly, from=5-7, to=3-8]
	\arrow[shift right, from=5-7, to=3-8]
	\arrow[dashed, from=5-7, to=3-9]
	\arrow[dashed, from=5-7, to=3-11]
	\arrow[dashed, from=5-7, to=3-13]
	\arrow[from=5-8, to=3-7]
	\arrow[from=5-8, to=3-8]
	\arrow[from=5-8, to=3-9]
	\arrow[from=7-1, to=5-6]
	\arrow[from=7-1, to=5-7]
	\arrow[from=7-1, to=5-8]
\end{tikzcd}}\]
\caption{\label{figure:write_error_four_steps} A collapsed version of three stacked diagrams as in Figure \ref{figure:symbol_write_error_diagram} (that is, three simulated timesteps of a TM, considering a possible error in the symbol to write square of the description tape) showing only paths that begin at $\theta$.}
\end{figure}

The step from time $2P$ to time $3P$ gives the general pattern. The key point is that at time $2P$, there is now uncertainty in all of the squares on the working tape, and the square on the state tape. Except for the growth of the working tape, there is no further introduction of uncertainty, so we now have the general picture. If we only consider the paths which end at $q$ in the fifth simulated step, we obtain the equivalent model shown in Figure \ref{figure:write_error_five_steps}.

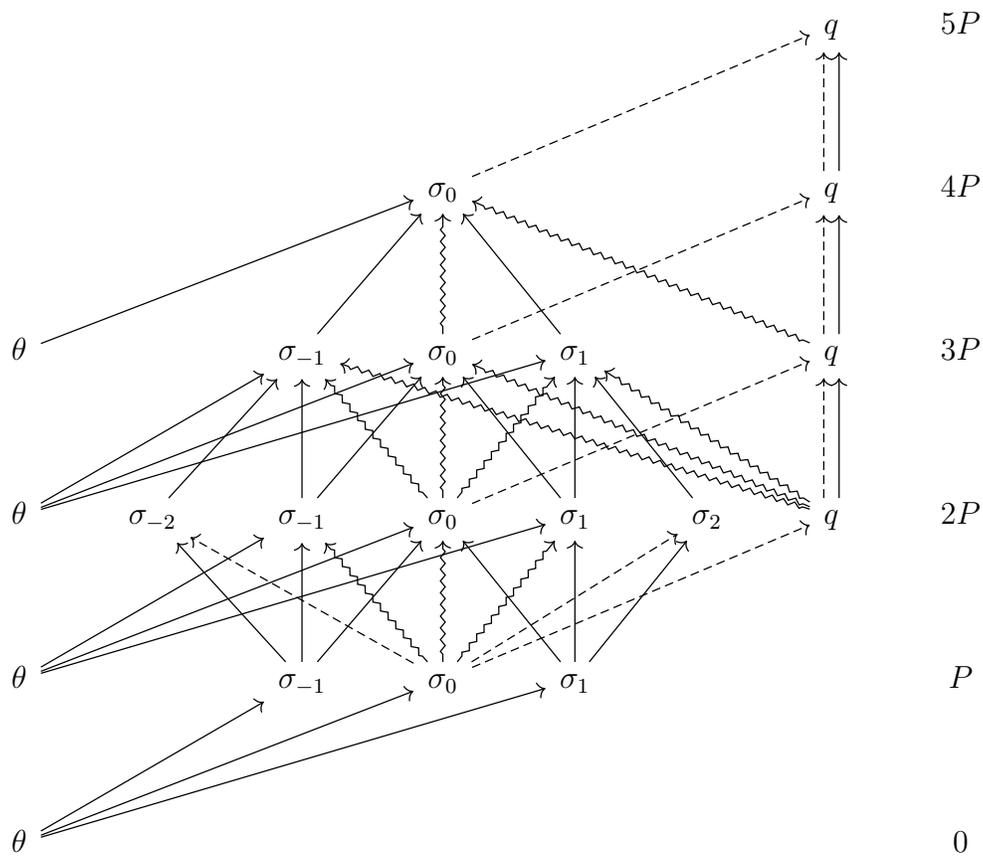
\begin{figure}[htpb]
    \centering
\[\begin{tikzcd}
	&&&&&& q & 5P \\
	\\
	&&& {\sigma_0} &&& q & 4P \\
	\\
	\theta && {\sigma_{-1}} & {\sigma_0} & {\sigma_1} && q & 3P \\
	\\
	\theta & {\sigma_{-2}} & {\sigma_{-1}} & {\sigma_0} & {\sigma_1} & {\sigma_2} & q & 2P \\
	\\
	\theta && {\sigma_{-1}} & {\sigma_0} & {\sigma_1} &&& P \\
	\\
	\theta &&&&&&& 0
	\arrow[dashed, from=3-4, to=1-7]
	\arrow[shift left, dashed, from=3-7, to=1-7]
	\arrow[shift right, from=3-7, to=1-7]
	\arrow[from=5-1, to=3-4]
	\arrow[from=5-3, to=3-4]
	\arrow[shift left, squiggly, from=5-4, to=3-4]
	\arrow[shift right, from=5-4, to=3-4]
	\arrow[dashed, from=5-4, to=3-7]
	\arrow[from=5-5, to=3-4]
	\arrow[squiggly, from=5-7, to=3-4]
	\arrow[shift left, dashed, from=5-7, to=3-7]
	\arrow[shift right, from=5-7, to=3-7]
	\arrow[from=7-1, to=5-3]
	\arrow[from=7-1, to=5-4]
	\arrow[from=7-1, to=5-5]
	\arrow[from=7-2, to=5-3]
	\arrow[from=7-3, to=5-3]
	\arrow[from=7-3, to=5-4]
	\arrow[shift left, squiggly, from=7-4, to=5-3]
	\arrow[shift right, from=7-4, to=5-3]
	\arrow[shift left, squiggly, from=7-4, to=5-4]
	\arrow[shift right, from=7-4, to=5-4]
	\arrow[shift left, squiggly, from=7-4, to=5-5]
	\arrow[shift right, from=7-4, to=5-5]
	\arrow[dashed, from=7-4, to=5-7]
	\arrow[from=7-5, to=5-4]
	\arrow[from=7-5, to=5-5]
	\arrow[from=7-6, to=5-5]
	\arrow[squiggly, from=7-7, to=5-3]
	\arrow[squiggly, from=7-7, to=5-4]
	\arrow[squiggly, from=7-7, to=5-5]
	\arrow[shift left, dashed, from=7-7, to=5-7]
	\arrow[shift right, from=7-7, to=5-7]
	\arrow[from=9-1, to=7-3]
	\arrow[from=9-1, to=7-4]
	\arrow[from=9-1, to=7-5]
	\arrow[from=9-3, to=7-2]
	\arrow[from=9-3, to=7-3]
	\arrow[from=9-3, to=7-4]
	\arrow[dashed, from=9-4, to=7-2]
	\arrow[shift left, squiggly, from=9-4, to=7-3]
	\arrow[shift right, from=9-4, to=7-3]
	\arrow[shift left, squiggly, from=9-4, to=7-4]
	\arrow[shift right, from=9-4, to=7-4]
	\arrow[shift left, squiggly, from=9-4, to=7-5]
	\arrow[shift right, from=9-4, to=7-5]
	\arrow[dashed, from=9-4, to=7-6]
	\arrow[dashed, from=9-4, to=7-7]
	\arrow[from=9-5, to=7-4]
	\arrow[from=9-5, to=7-5]
	\arrow[from=9-5, to=7-6]
	\arrow[from=11-1, to=9-3]
	\arrow[from=11-1, to=9-4]
	\arrow[from=11-1, to=9-5]
\end{tikzcd}\]
    \caption{\label{figure:write_error_five_steps} A collapsed version of five stacked diagrams as in Figure \ref{figure:symbol_write_error_diagram} (that is, five simulated timesteps of a TM, considering a possible error in the symbol to write square of the description tape) showing only paths that begin at $\theta$ and end at $q$.
    }
\end{figure}

\subsection{State write error}
Now we consider a state write error. This situation is similar to the above, but at time $P$ only the state square has uncertainty on it, not the $\sigma_0$ square on the working tape. However, at time $2P$, there is the same set of squares with uncertainty as in the symbol write error case. We obtain the DGM from time $P$ to time $3P$ as shown in Figure \ref{figure:state_error_three_steps}.

\begin{figure}[htpb]
    \centering
\[\adjustbox{scale=0.85}{\begin{tikzcd}[column sep = tiny, row sep = huge]
	& {\sigma_{-h-3}} & {\sigma_{-h-2}} & \ldots & {\sigma_{-2}} & {\sigma_{-1}} & {\sigma_0} & {\sigma_1} & {\sigma_2} & \ldots & {\sigma_{h+2}} & {\sigma_{h+3}} & q & 3P \\
	\\
	\theta && {\sigma_{-h-2}} & \ldots & {\sigma_{-2}} & {\sigma_{-1}} & {\sigma_0} & {\sigma_1} & {\sigma_2} & \ldots & {\sigma_{h+2}} && q & 2P \\
	\\
	\theta &&&&&&&&&&&& q & P \\
	\\
	\theta &&&&&&&&&&&&& 0
	\arrow[from=3-1, to=1-13]
	\arrow[from=3-3, to=1-2]
	\arrow[from=3-3, to=1-3]
	\arrow[from=3-5, to=1-5]
	\arrow[from=3-5, to=1-6]
	\arrow[from=3-6, to=1-5]
	\arrow[from=3-6, to=1-6]
	\arrow[from=3-6, to=1-7]
	\arrow[dashed, from=3-7, to=1-2]
	\arrow[dashed, from=3-7, to=1-3]
	\arrow[dashed, from=3-7, to=1-5]
	\arrow[shift left, squiggly, from=3-7, to=1-6]
	\arrow[shift right, from=3-7, to=1-6]
	\arrow[shift left, squiggly, from=3-7, to=1-7]
	\arrow[shift right, from=3-7, to=1-7]
	\arrow[shift left, squiggly, from=3-7, to=1-8]
	\arrow[shift right, from=3-7, to=1-8]
	\arrow[dashed, from=3-7, to=1-9]
	\arrow[dashed, from=3-7, to=1-11]
	\arrow[dashed, from=3-7, to=1-12]
	\arrow[dashed, from=3-7, to=1-13]
	\arrow[from=3-8, to=1-7]
	\arrow[from=3-8, to=1-8]
	\arrow[from=3-8, to=1-9]
	\arrow[from=3-9, to=1-8]
	\arrow[from=3-9, to=1-9]
	\arrow[from=3-11, to=1-11]
	\arrow[from=3-11, to=1-12]
	\arrow[dashed, from=3-13, to=1-2]
	\arrow[dashed, from=3-13, to=1-3]
	\arrow[dashed, from=3-13, to=1-5]
	\arrow[squiggly, from=3-13, to=1-6]
	\arrow[squiggly, from=3-13, to=1-7]
	\arrow[squiggly, from=3-13, to=1-8]
	\arrow[dashed, from=3-13, to=1-9]
	\arrow[dashed, from=3-13, to=1-11]
	\arrow[dashed, from=3-13, to=1-12]
	\arrow[shift left, dashed, from=3-13, to=1-13]
	\arrow[shift right, from=3-13, to=1-13]
	\arrow[from=5-1, to=3-13]
	\arrow[dashed, from=5-13, to=3-3]
	\arrow[dashed, from=5-13, to=3-5]
	\arrow[squiggly, from=5-13, to=3-6]
	\arrow[squiggly, from=5-13, to=3-7]
	\arrow[squiggly, from=5-13, to=3-8]
	\arrow[dashed, from=5-13, to=3-9]
	\arrow[dashed, from=5-13, to=3-11]
	\arrow[shift left, dashed, from=5-13, to=3-13]
	\arrow[shift right, from=5-13, to=3-13]
	\arrow[from=7-1, to=5-13]
\end{tikzcd}}\]
\caption{\label{figure:state_error_three_steps} A collapsed version of three stacked diagrams as in Figure \ref{figure:symbol_write_error_diagram} (that is, three simulated timesteps of a TM, considering a possible error in the state to transition to on the description tape) showing only paths that begin at $\theta$.}
\end{figure}
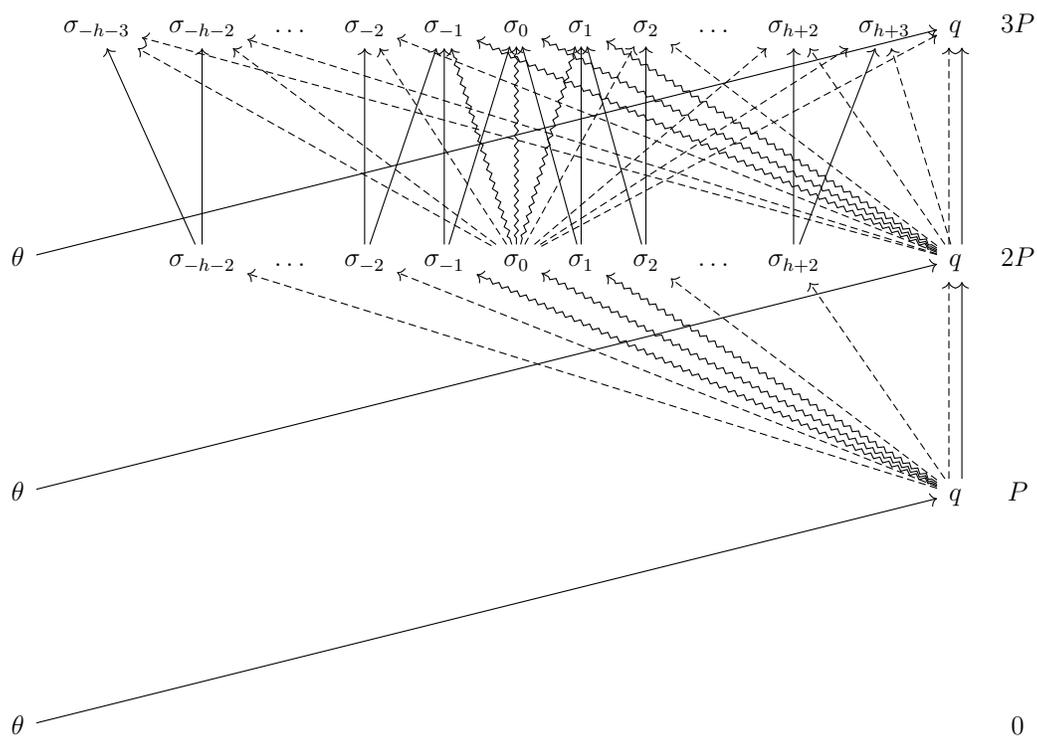

If we consider $5$ simulated steps and only paths which end at $q$, we obtain Figure \ref{figure:state_error_five_steps}.

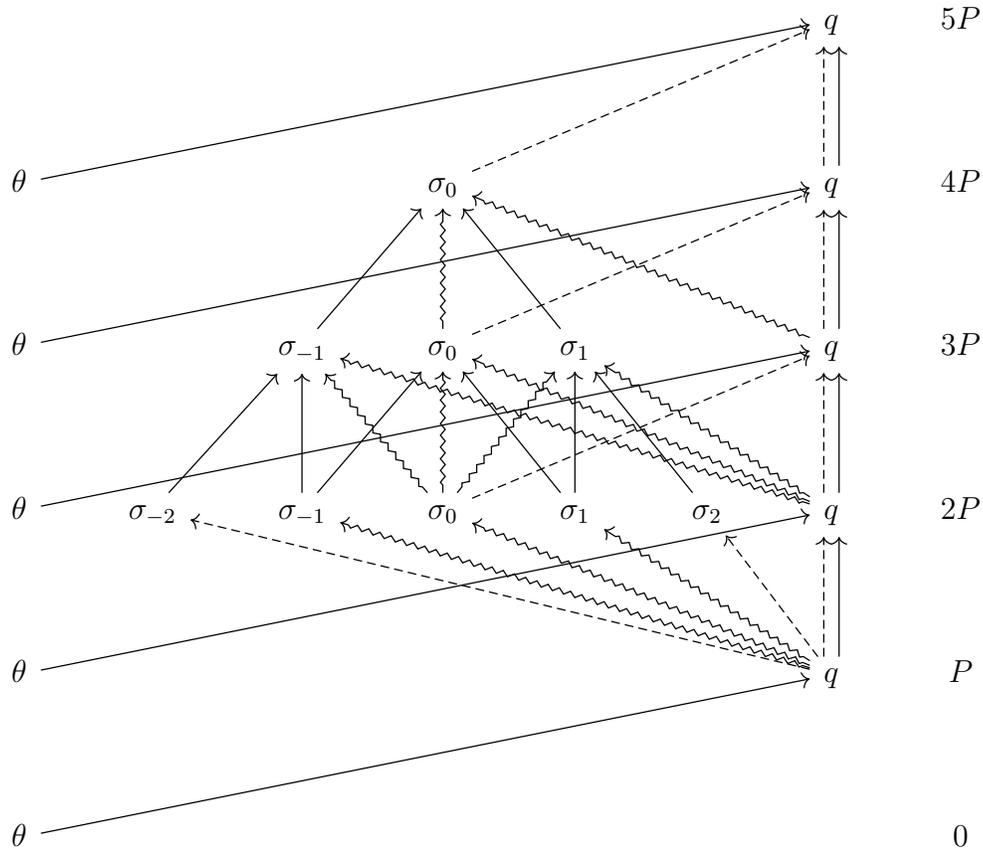
\begin{figure}[htpb]
    \centering
\[\begin{tikzcd}
	&&&&&& q & 5P \\
	\\
	\theta &&& {\sigma_0} &&& q & 4P \\
	\\
	\theta && {\sigma_{-1}} & {\sigma_0} & {\sigma_1} && q & 3P \\
	\\
	\theta & {\sigma_{-2}} & {\sigma_{-1}} & {\sigma_0} & {\sigma_1} & {\sigma_2} & q & 2P \\
	\\
	\theta &&&&&& q & P \\
	\\
	\theta &&&&&&& 0
	\arrow[from=3-1, to=1-7]
	\arrow[dashed, from=3-4, to=1-7]
	\arrow[shift left, dashed, from=3-7, to=1-7]
	\arrow[shift right, from=3-7, to=1-7]
	\arrow[from=5-1, to=3-7]
	\arrow[from=5-3, to=3-4]
	\arrow[shift left, squiggly, from=5-4, to=3-4]
	\arrow[shift right, from=5-4, to=3-4]
	\arrow[dashed, from=5-4, to=3-7]
	\arrow[from=5-5, to=3-4]
	\arrow[squiggly, from=5-7, to=3-4]
	\arrow[shift left, dashed, from=5-7, to=3-7]
	\arrow[shift right, from=5-7, to=3-7]
	\arrow[from=7-1, to=5-7]
	\arrow[from=7-2, to=5-3]
	\arrow[from=7-3, to=5-3]
	\arrow[from=7-3, to=5-4]
	\arrow[shift left, squiggly, from=7-4, to=5-3]
	\arrow[shift right, from=7-4, to=5-3]
	\arrow[shift left, squiggly, from=7-4, to=5-4]
	\arrow[shift right, from=7-4, to=5-4]
	\arrow[shift left, squiggly, from=7-4, to=5-5]
	\arrow[shift right, from=7-4, to=5-5]
	\arrow[dashed, from=7-4, to=5-7]
	\arrow[from=7-5, to=5-4]
	\arrow[from=7-5, to=5-5]
	\arrow[from=7-6, to=5-5]
	\arrow[squiggly, from=7-7, to=5-3]
	\arrow[squiggly, from=7-7, to=5-4]
	\arrow[squiggly, from=7-7, to=5-5]
	\arrow[shift left, dashed, from=7-7, to=5-7]
	\arrow[shift right, from=7-7, to=5-7]
	\arrow[from=9-1, to=7-7]
	\arrow[dashed, from=9-7, to=7-2]
	\arrow[squiggly, from=9-7, to=7-3]
	\arrow[squiggly, from=9-7, to=7-4]
	\arrow[squiggly, from=9-7, to=7-5]
	\arrow[dashed, from=9-7, to=7-6]
	\arrow[shift left, dashed, from=9-7, to=7-7]
	\arrow[shift right, from=9-7, to=7-7]
	\arrow[from=11-1, to=9-7]
\end{tikzcd}\]
\caption{\label{figure:state_error_five_steps} A collapsed version of five stacked diagrams as in Figure \ref{figure:symbol_write_error_diagram} (that is, five simulated timesteps of a TM, considering a possible error in the state to transition to on the description tape) showing only paths that begin at $\theta$ and end at $q$.}
\end{figure}

\subsection{Direction write error}
Finally, we consider a direction write error. The simplified DGM for three simulated steps of $T$ is given in Figure \ref{figure:direction_error_three_steps}.

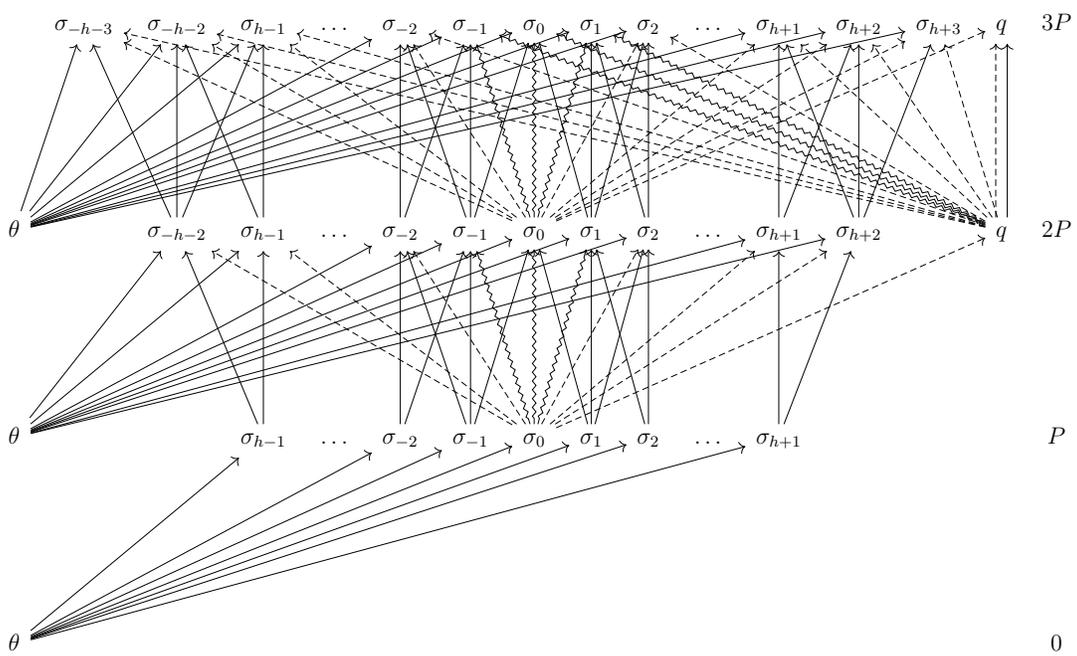
\begin{figure}[htpb]
    \centering
\[\adjustbox{scale=0.75}{\begin{tikzcd}[column sep = tiny, row sep = huge]
	& {\sigma_{-h-3}} & {\sigma_{-h-2}} & {\sigma_{-h-1}} & \ldots & {\sigma_{-2}} & {\sigma_{-1}} & {\sigma_0} & {\sigma_1} & {\sigma_2} & \ldots & {\sigma_{h+1}} & {\sigma_{h+2}} & {\sigma_{h+3}} & q & 3P \\
	\\
	\theta && {\sigma_{-h-2}} & {\sigma_{-h-1}} & \ldots & {\sigma_{-2}} & {\sigma_{-1}} & {\sigma_0} & {\sigma_1} & {\sigma_2} & \ldots & {\sigma_{h+1}} & {\sigma_{h+2}} && q & 2P \\
	\\
	\theta &&& {\sigma_{-h-1}} & \ldots & {\sigma_{-2}} & {\sigma_{-1}} & {\sigma_0} & {\sigma_1} & {\sigma_2} & \ldots & {\sigma_{h+1}} &&&& P \\
	\\
	\theta &&&&&&&&&&&&&&& 0
	\arrow[from=3-1, to=1-2]
	\arrow[from=3-1, to=1-3]
	\arrow[from=3-1, to=1-4]
	\arrow[from=3-1, to=1-6]
	\arrow[from=3-1, to=1-7]
	\arrow[from=3-1, to=1-8]
	\arrow[from=3-1, to=1-9]
	\arrow[from=3-1, to=1-10]
	\arrow[from=3-1, to=1-12]
	\arrow[from=3-1, to=1-13]
	\arrow[from=3-1, to=1-14]
	\arrow[from=3-3, to=1-2]
	\arrow[from=3-3, to=1-3]
	\arrow[from=3-3, to=1-4]
	\arrow[from=3-4, to=1-3]
	\arrow[from=3-4, to=1-4]
	\arrow[from=3-6, to=1-6]
	\arrow[from=3-6, to=1-7]
	\arrow[from=3-7, to=1-6]
	\arrow[from=3-7, to=1-7]
	\arrow[from=3-7, to=1-8]
	\arrow[dashed, from=3-8, to=1-2]
	\arrow[dashed, from=3-8, to=1-3]
	\arrow[dashed, from=3-8, to=1-4]
	\arrow[dashed, from=3-8, to=1-6]
	\arrow[shift left, squiggly, from=3-8, to=1-7]
	\arrow[shift right, from=3-8, to=1-7]
	\arrow[shift left, squiggly, from=3-8, to=1-8]
	\arrow[shift right, from=3-8, to=1-8]
	\arrow[shift left, squiggly, from=3-8, to=1-9]
	\arrow[shift right, from=3-8, to=1-9]
	\arrow[dashed, from=3-8, to=1-10]
	\arrow[dashed, from=3-8, to=1-12]
	\arrow[dashed, from=3-8, to=1-13]
	\arrow[dashed, from=3-8, to=1-14]
	\arrow[dashed, from=3-8, to=1-15]
	\arrow[from=3-9, to=1-8]
	\arrow[from=3-9, to=1-9]
	\arrow[from=3-9, to=1-10]
	\arrow[from=3-10, to=1-9]
	\arrow[from=3-10, to=1-10]
	\arrow[from=3-12, to=1-12]
	\arrow[from=3-12, to=1-13]
	\arrow[from=3-13, to=1-12]
	\arrow[from=3-13, to=1-13]
	\arrow[from=3-13, to=1-14]
	\arrow[dashed, from=3-15, to=1-2]
	\arrow[dashed, from=3-15, to=1-3]
	\arrow[dashed, from=3-15, to=1-4]
	\arrow[dashed, from=3-15, to=1-6]
	\arrow[squiggly, from=3-15, to=1-7]
	\arrow[squiggly, from=3-15, to=1-8]
	\arrow[squiggly, from=3-15, to=1-9]
	\arrow[dashed, from=3-15, to=1-10]
	\arrow[dashed, from=3-15, to=1-12]
	\arrow[dashed, from=3-15, to=1-13]
	\arrow[dashed, from=3-15, to=1-14]
	\arrow[shift left, dashed, from=3-15, to=1-15]
	\arrow[shift right, from=3-15, to=1-15]
	\arrow[from=5-1, to=3-3]
	\arrow[from=5-1, to=3-4]
	\arrow[from=5-1, to=3-6]
	\arrow[from=5-1, to=3-7]
	\arrow[from=5-1, to=3-8]
	\arrow[from=5-1, to=3-9]
	\arrow[from=5-1, to=3-10]
	\arrow[from=5-1, to=3-12]
	\arrow[from=5-1, to=3-13]
	\arrow[from=5-4, to=3-3]
	\arrow[from=5-4, to=3-4]
	\arrow[from=5-6, to=3-6]
	\arrow[from=5-6, to=3-7]
	\arrow[from=5-7, to=3-6]
	\arrow[from=5-7, to=3-7]
	\arrow[from=5-7, to=3-8]
	\arrow[dashed, from=5-8, to=3-3]
	\arrow[dashed, from=5-8, to=3-4]
	\arrow[dashed, from=5-8, to=3-6]
	\arrow[shift left, squiggly, from=5-8, to=3-7]
	\arrow[shift right, from=5-8, to=3-7]
	\arrow[shift left, squiggly, from=5-8, to=3-8]
	\arrow[shift right, from=5-8, to=3-8]
	\arrow[shift left, squiggly, from=5-8, to=3-9]
	\arrow[shift right, from=5-8, to=3-9]
	\arrow[dashed, from=5-8, to=3-10]
	\arrow[dashed, from=5-8, to=3-12]
	\arrow[dashed, from=5-8, to=3-13]
	\arrow[dashed, from=5-8, to=3-15]
	\arrow[from=5-9, to=3-8]
	\arrow[from=5-9, to=3-9]
	\arrow[from=5-9, to=3-10]
	\arrow[from=5-10, to=3-9]
	\arrow[from=5-10, to=3-10]
	\arrow[from=5-12, to=3-12]
	\arrow[from=5-12, to=3-13]
	\arrow[from=7-1, to=5-4]
	\arrow[from=7-1, to=5-6]
	\arrow[from=7-1, to=5-7]
	\arrow[from=7-1, to=5-8]
	\arrow[from=7-1, to=5-9]
	\arrow[from=7-1, to=5-10]
	\arrow[from=7-1, to=5-12]
\end{tikzcd}}\]
\caption{\label{figure:direction_error_three_steps} A collapsed version of three stacked diagrams as in Figure \ref{figure:symbol_write_error_diagram} (that is, three simulated timesteps of a TM, considering a possible error in the direction to move) showing only paths that begin at $\theta$.}
\end{figure}

If we consider 5 simulated steps and only paths which end at $q$, we obtain Figure \ref{figure:direction_error_five_steps}.

\begin{figure}[htpb]
    \centering
\[\begin{tikzcd}[column sep = small]
	&&&&&&&& q & 5P \\
	\\
	&&&& {\sigma_0} &&&& q & 4P \\
	\\
	\theta &&& {\sigma_{-1}} & {\sigma_0} & {\sigma_1} &&& q & 3P \\
	\\
	\theta && {\sigma_{-2}} & {\sigma_{-1}} & {\sigma_0} & {\sigma_1} & {\sigma_2} && q & 2P \\
	\\
	\theta & {\sigma_{-3}} & {\sigma_{-2}} & {\sigma_{-1}} & {\sigma_0} & {\sigma_1} & {\sigma_{2}} & {\sigma_3} && P \\
	\\
	\theta &&&&&&&&& 0
	\arrow[dashed, from=3-5, to=1-9]
	\arrow[shift left, dashed, from=3-9, to=1-9]
	\arrow[shift right, from=3-9, to=1-9]
	\arrow[from=5-1, to=3-5]
	\arrow[from=5-4, to=3-5]
	\arrow[shift left, squiggly, from=5-5, to=3-5]
	\arrow[shift right, from=5-5, to=3-5]
	\arrow[dashed, from=5-5, to=3-9]
	\arrow[from=5-6, to=3-5]
	\arrow[squiggly, from=5-9, to=3-5]
	\arrow[shift left, dashed, from=5-9, to=3-9]
	\arrow[shift right, from=5-9, to=3-9]
	\arrow[from=7-1, to=5-4]
	\arrow[from=7-1, to=5-5]
	\arrow[from=7-1, to=5-6]
	\arrow[from=7-3, to=5-4]
	\arrow[from=7-4, to=5-4]
	\arrow[from=7-4, to=5-5]
	\arrow[shift left, squiggly, from=7-5, to=5-4]
	\arrow[shift right, from=7-5, to=5-4]
	\arrow[shift left, squiggly, from=7-5, to=5-5]
	\arrow[shift right, from=7-5, to=5-5]
	\arrow[shift left, squiggly, from=7-5, to=5-6]
	\arrow[shift right, from=7-5, to=5-6]
	\arrow[dashed, from=7-5, to=5-9]
	\arrow[from=7-6, to=5-5]
	\arrow[from=7-6, to=5-6]
	\arrow[from=7-7, to=5-6]
	\arrow[squiggly, from=7-9, to=5-4]
	\arrow[squiggly, from=7-9, to=5-5]
	\arrow[squiggly, from=7-9, to=5-6]
	\arrow[shift left, dashed, from=7-9, to=5-9]
	\arrow[shift right, from=7-9, to=5-9]
	\arrow[from=9-1, to=7-3]
	\arrow[from=9-1, to=7-4]
	\arrow[from=9-1, to=7-5]
	\arrow[from=9-1, to=7-6]
	\arrow[from=9-1, to=7-7]
	\arrow[from=9-2, to=7-3]
	\arrow[from=9-3, to=7-3]
	\arrow[from=9-3, to=7-4]
	\arrow[from=9-4, to=7-3]
	\arrow[from=9-4, to=7-4]
	\arrow[from=9-4, to=7-5]
	\arrow[dashed, from=9-5, to=7-3]
	\arrow[shift left, squiggly, from=9-5, to=7-4]
	\arrow[shift right, from=9-5, to=7-4]
	\arrow[shift left, squiggly, from=9-5, to=7-5]
	\arrow[shift right, from=9-5, to=7-5]
	\arrow[shift left, squiggly, from=9-5, to=7-6]
	\arrow[shift right, from=9-5, to=7-6]
	\arrow[dashed, from=9-5, to=7-7]
	\arrow[dashed, from=9-5, to=7-9]
	\arrow[from=9-6, to=7-5]
	\arrow[from=9-6, to=7-6]
	\arrow[from=9-6, to=7-7]
	\arrow[from=9-7, to=7-6]
	\arrow[from=9-7, to=7-7]
	\arrow[from=9-8, to=7-7]
	\arrow[from=11-1, to=9-2]
	\arrow[from=11-1, to=9-3]
	\arrow[from=11-1, to=9-4]
	\arrow[from=11-1, to=9-5]
	\arrow[from=11-1, to=9-6]
	\arrow[from=11-1, to=9-7]
	\arrow[from=11-1, to=9-8]
\end{tikzcd}\]
\caption{\label{figure:direction_error_five_steps} A collapsed version of five stacked diagrams as in Figure \ref{figure:symbol_write_error_diagram} (that is, five simulated timesteps of a TM, considering a possible error in the direction to move) showing only paths that begin at $\theta$ and end at $q$.}
\end{figure}
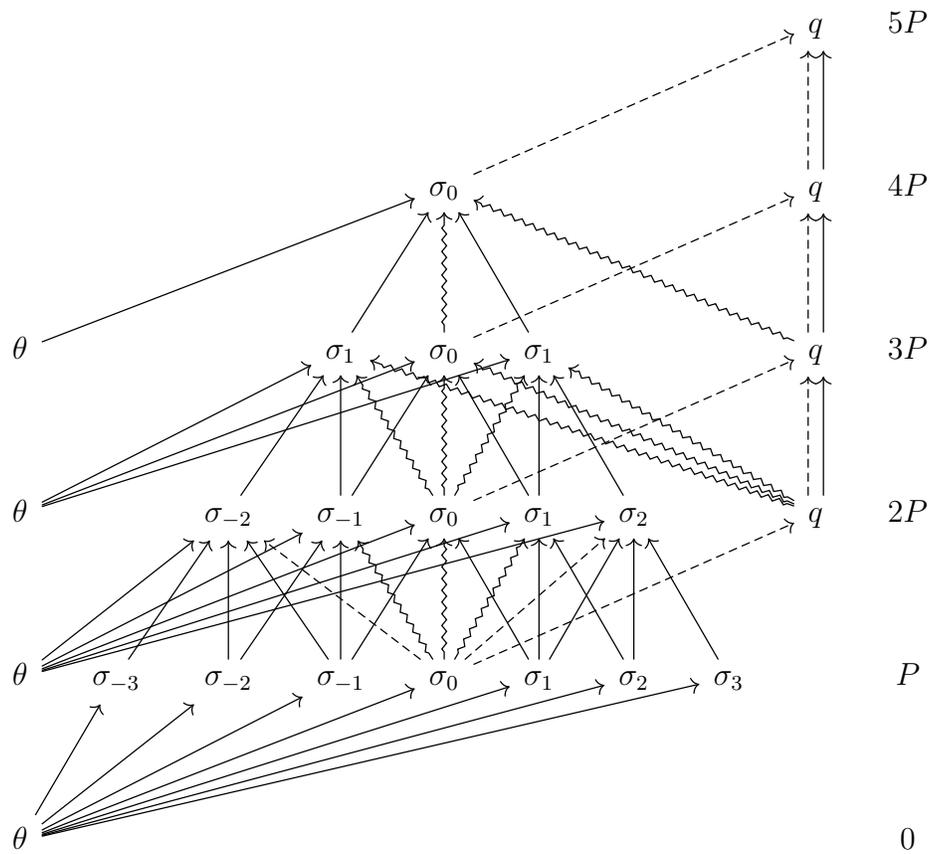

\begin{remark}
    In all three cases, the trees from time step $2P$ and above are identical except for the arrows with source $\theta$.
\end{remark}

\end{document}